\documentclass[11pt]{article}
\usepackage[margin=0.75in]{geometry}
\RequirePackage{amsthm,amsmath,amsfonts,amssymb, bbm}
\RequirePackage[authoryear]{natbib}
\RequirePackage[colorlinks,citecolor=blue,urlcolor=blue]{hyperref}
\RequirePackage{graphicx}
\usepackage{wrapfig}
\usepackage{subcaption}

\theoremstyle{plain}

\newtheorem{theorem}{Theorem}[section]
\newtheorem{lemma}[theorem]{Lemma}
\newtheorem{proposition}[theorem]{Proposition}
\theoremstyle{definition}
\newtheorem{definition}[theorem]{Definition}
\newtheorem*{example}{Example}

\theoremstyle{remark}

\newtheorem{algorithm}{Algorithm}
\newtheorem{remark}{Remark}
\usepackage{amssymb,mathtools}

\DeclarePairedDelimiter{\Norm}{\lVert}{\rVert}

\def\indist{\rightsquigarrow}
\def\ind{\perp\!\!\!\perp}

\newcommand{\var}{\text{var}}

\newcommand{\Pb}{\mathbb{P}}

\newcommand{\Pn}{\mathbb{P}_n}
\newcommand{\Un}{\mathbb{U}_n}

\newcommand{\E}{\mathbb{E}}
\newcommand{\R}{\mathbb{R}}

\newcommand{\pihat}{\widehat\pi}
\newcommand{\muhat}{\widehat\mu}

\newcommand{\phat}{\widehat{p}}
\newcommand{\mhat}{\widehat{m}}
\newcommand{\hhat}{\widehat{h}}
\newcommand{\rhat}{\widehat{r}}

\newcommand{\Qhat}{\widehat Q}

\newcommand{\ghat}{\widehat{g}}

\newcommand{\shat}{\widehat{s}}
\newcommand{\Ehat}{\widehat{\E}}
\newcommand{\fhat}{\widehat{f}}

\newcommand{\thetahat}{\widehat{\theta}}

\DeclareMathOperator*{\argmin}{arg\,min}

\DeclareMathOperator{\sgn}{sgn}
\newcommand{\one}{\mathbbm{1}}

\newcommand{\blind}{1}

\begin{document}

	\def\spacingset#1{\renewcommand{\baselinestretch}%
		{#1}\small\normalsize} \spacingset{1}

	
	\if1\blind
	{
		\title{\bf Fast convergence rates for dose-response estimation}
		\author{Matteo Bonvini\thanks{Department of Statistics, Rutgers, The State University of New Jersey. Email: mb1662@stat.rutgers.edu} \and Edward H. Kennedy\thanks{Associate Professor, Department of Statistics \& Data Science, Carnegie Mellon University. Email: edward@stat.cmu.edu.}}
		\maketitle
	} \fi
	
	\if0\blind
	{
		\bigskip
		\bigskip
		\bigskip
		\begin{center}
			{\bf Title}
		\end{center}
		\medskip
	} \fi
	
	\bigskip
\begin{abstract}
We consider the problem of estimating a dose-response curve. Continuous treatments arise often in practice, e.g. in the form of time spent on an operation, distance traveled to a location or dosage of a drug. Letting $A$ denote a continuous treatment variable, the target of inference is the expected outcome if everyone in the population takes treatment level $A=t$. Under standard assumptions, the dose-response function takes the form of a partial mean. Building upon the recent literature on nonparametric regression with estimated outcomes, our first contribution is to study global and local estimators of the dose-response based on empirical risk minimization. Our second and main contribution is to construct a $m^{\text{th}}$-order estimator based on the theory of higher-order influence functions. Under certain conditions, this higher order estimator achieves the fastest rate of convergence that we are aware of for this problem. However, the other two approaches are easier to implement using off-the-shelf software, since they are formulated as two-stage regression tasks. For each estimator, we provide an upper bound on the mean-square error and investigate its finite-sample performance through simulations and an empirical application. Finally, the supplementary material introduces a flexible, nonparametric approach for sensitivity analysis to violations of the no-unmeasured-confounding assumption with continuous treatments.
\end{abstract}

\section{Introduction}
\color{black}
In this work we study causal effects defined by a continuous treatment in non-randomized studies. Continuous or multi-valued treatments arise frequently in practice; examples include time, distance traveled, or dosage of a drug. Estimating the effect of such treatments on an outcome poses statistical challenges that do not arise with binary or discrete treatments. 

As in the discrete-treatment case, estimating causal effects in non-randomized studies typically requires estimation of complex nuisance functions. With discrete treatments, these nuisances consist of the outcome regression and the probability of treatment given the covariates (the propensity score). When these functions are estimated sufficiently accurately, average causal effects can be estimated at the fast parametric rate $n^{-1/2}$ even in nonparametric models. When the treatment is continuous, the propensity score is replaced by the conditional density of the treatment given the covariates. However, in this case even the average effect corresponding to a fixed treatment level is no longer $n^{-1/2}$-estimable in nonparametric models.

The estimand of interest is the dose–response curve $\theta_0(t)$, which maps each treatment level to the expected outcome in the population if every subject were assigned that level. Under causal assumptions, it equals $\E\{\E(Y \mid A = t, X)\}$, where $Y$ is the outcome, $A$ is the treatment and $X \in \R^p$ are measured confounders needed for identification. As noted in \cite{kennedy2017nonparametric}, $\theta_0(t)$ is a hybrid parameter sharing features of both functional estimation and nonparametric regression. In particular, their work shows that regressing a suitably constructed pseudo-outcome on $A$ yields an estimator of $\theta_0(t)$ with favorable statistical properties. Their theoretical analysis, however, was confined to local linear estimation. Our first contribution is therefore to build on recent results of \cite{foster2019orthogonal} and \cite{kennedy2020optimal} to extend the insights of \cite{kennedy2017nonparametric} to a broader class of estimators (Section \ref{sec:dr_est}). Section \ref{sec:lit_review_dr} reviews state-of-the-art approaches to estimate $\theta_0(t)$ by either regressing or averaging (locally) a carefully constructed signal.

An advantage of the estimators described in Section \ref{sec:lit_review_dr} and our proposals in Sections \ref{sec:dr_est} is that they do not require the nuisance functions to be estimated using a specific procedure or to belong to a particular function class. Borrowing terminology from \cite{balakrishnan2023fundamental}, they are suitable for \textit{structure-agnostic} estimation. However, their convergence rates can be potentially suboptimal when the nuisances possess additional structure. A canonical example--- also the focus of our work--- arises when the conditional density of $A$ given $X$ and the outcome regression are H\"{o}lder smooth. Our second and main contribution is to construct an estimator of $\theta_0(t)$ that builds upon, and extends the use of, the theory of higher-order influence functions \citep{robins2008higher, robins2017higher} (Section \ref{sec:higher_order}). To our knowledge, this estimator achieves the fastest convergence rate currently available in the literature for estimating $\theta_0(t)$ in H\"{o}lder models. We conjecture that the derived rate is minimax optimal, at least under certain conditions, but the construction of a tight lower bound in this setting remains elusive. Deriving the minimax rate for estimating $\theta_0(t)$ in smoothness models represents a natural and important next step, as recognized also in the recent review of open challenges in causality and causal inference \citep{cinelli2025challenges}.
\color{black}
\subsection{Notation \& setup}\label{section:notation}
\color{black}
We let  $A \in \mathcal{A} \subset \R$  denote the continuous treatment and $Y \in \mathcal{Y} \subset \R$ the outcome. Within the potential outcomes framework \citep{rubin1974estimating}, the expected outcome if everyone in the population takes $A = t$ is $\E(Y^t)$. Because $A$ is continuous, $t \mapsto \E(Y^t)$ is referred to as the \textit{dose-response function} (DRF). Under standard assumptions (see e.g. \cite{kennedy2017nonparametric}), the DRF takes the form:
\begin{align*}
		\theta_0(t) = \E\{\E(Y \mid A = t, X)\} = \int \E(Y \mid A = t, X = x) d\Pb(x),
\end{align*}
where $X \in \mathcal{X} \subset \R^d$ denotes measured confounders. We consider estimating $\theta_0(t)$ having access to $n$ iid copies of $Z = (Y, A, X)$ sampled from some distribution $\Pb$ with density $p$ with respect to the Lebesgue measure. To simplify the notation, we write $p(u)$ for the density of a generic random variable $U$ with respect to the Lebesgue measure and define
\begin{align*}
&  \pi(a \mid x) = \frac{p(a, x)}{p(x)}, \quad \mu(a, x) = \E(Y \mid A = a, X = x), & w(a, x) = \frac{p(a)}{\pi(a \mid x)}.
\end{align*}
That is, $\mu(a, x)$ is the outcome regression, and $\pi(a \mid x)$ is the conditional density of $A$ given $X = x$. We will sometimes denote the main nuisance functions by $\eta = \{p(a), \mu(a, x), \pi(a \mid x)\}$. With this notation, we have
\begin{align*}
	\theta_0(t) = \E\{\mu(t, X)\} = \E \left\{w(t, X) Y \mid A = t\right\}.
\end{align*}
For a kernel function $K(u)$, we define $K_{ht}(a) = h^{-1} K((a - t) / h)$. 

We use the notation $\Pb\{g(Z)\} = \int g(z) d\Pb(z)$ and $\Pn\{g(Z)\} = n^{-1} \sum_{i = 1}^n g(Z_i)$ to denote means (given $g$) and sample means. Further, we let $\Norm{f}^2 = \int f^2(z) d\Pb(z)$. We write $a \lesssim b$ if there exists a constant $C$ independent of the sample size such that $a \leq C b$. For shorthand notation, we also let $a \land b = \min(a, b)$, $a \lor b = \max(a,b)$ and $a \asymp b$ whenever $a \lesssim b$ and $b \lesssim a$.

\medskip
Throughout the paper, we will rely on the following assumptions. Additional assumptions will be introduced as needed. 
\begin{enumerate}
    \item Positivity: $\pi(a \mid x)$ and its estimator $\pihat(a \mid x)$ are bounded above and away from zero for all $a \in \mathcal{A}$ and $x \in \mathcal{X}$;
    \item Boundedness: $Y$, $A$ and $\muhat(a, x)$ (the estimator of $\mu(a, x)$) are uniformly bounded.
\end{enumerate}
Positivity is sufficient to ensure that $\theta_0(t)$ is well-defined, but it is not sufficient to interpret $\theta_0(t)$ causally. To interpret $\theta_0(t)$ as the average effect of assigning $A = t$ on $Y$, one needs to impose additional causal assumptions such as $Y^t \ind A \mid X$ and $A = t \implies Y^t = Y$, i.e., no unmeasured confounding and consistency (e.g. no interference). This paper is about estimating $\theta_0(t)$, regardless of its interpretation, and we refer to \cite{kennedy2017nonparametric} and references therein for more details on identification. 

Finally, our focus will be on estimation of the dose-response in nonparametric models where the dose-response itself and the nuisance functions possess varying degrees of smoothness. In particular, we will distinguish between the smoothness levels of the dose response $a \mapsto \theta_0(a)$, the conditional density of the treatment given the measured confounders $(a, x) \mapsto \pi(a \mid x)$, and the outcome regression $(a, x) \mapsto \mu(a, x)$. We will further refine this distinction when introducing the $m^{\text{th}}$-order estimator in the sense that we will consider models where $a \mapsto \mu(a, x)$, $x \mapsto \mu(a, x)$, $a \mapsto \pi(a \mid x)$ and $x \mapsto \pi(a \mid x)$ may have different smoothness levels. Note that it is reasonable to expect the smoothness of $a \mapsto \mu(a, x)$ to match that of $a \mapsto \theta_0(a) = \int \mu(a, x)d\Pb(x)$ in most applications.
\color{black}
\subsection{Literature review}
Because $A$ is continuous, $\theta_0(t)$ is generally not
$n^{-1/2}$-estimable in nonparametric models. In particular,
it can be written as $\theta_0(t) = \E\{ w(t,X)Y \mid A = t \}$, where $w(t,X)$ is a weighting function depending on the conditional treatment density. Thus, even in an idealized randomized study where the signal $w(t,X)Y$ would be fully observed, estimating $\theta_0(t)$ would reduce to a nonparametric regression problem with response $w(t,X)Y$ and predictor $A$. Consequently, the fastest convergence rate attainable would be the rate for nonparametric regression.

To compare the performance of different estimators, it will be useful to identify regimes under which they behave similarly to an oracle estimator that regresses $\varphi(Z)$ on $A$, where
	\begin{align}\label{eq:varphi}
		\varphi(Z) = w(A, X) \{Y - \mu(A, X)\} + \int \mu(A, x) , d\Pb(x).
	\end{align}
	It can be verified that $\theta_0(t) = \E\{\varphi(Z) \mid A = t\}$, so that regressing $\varphi(Z)$ on $A$ indeed recovers the dose–response function up to the ordinary nonparametric regression error.
\begin{definition}[Oracle rate]\label{def::oracle_rate}
Given an iid sample $Z_1, \ldots, Z_n$, let $\widetilde{\theta}(\cdot)$ be the (infeasible) estimator regressing $\varphi(Z)$ on $A$ and let $r_n$ be its convergence rate measured under a given loss. We refer to $r_n$ as the \textit{oracle rate}. 
\end{definition}
The definition of the oracle rate above depends on how the error in estimating the dose--response curve is measured. For example, $r_n$ my correspond to the to the risk of $\widetilde{\theta}(t)$ under squared loss, either pointwise, $\E[\{\widetilde\theta(t) - \theta_0(t)\}^2]$, or integrated, $\int \E[\{\widetilde\theta(t) - \theta_0(t)\}^2] p(t) dt$ where $p(t)$ denotes the density of $A$. If $\theta_0(t)$ is H\"{o}lder-smooth of order $\alpha$, and $\widetilde\theta(t)$ is a minimax optimal estimator of $\theta_0(t)$, then the square risk $r_n$ (either pointwise or integrated) is of order $n^{-2\alpha/(2\alpha + 1)}$ \citep{tsybakov2008introduction}.

One popular approach to dose–response estimation is to specify a \textit{marginal structural model} $\theta_0(t) = m(t; \beta)$ \citep{robins2000marginal}, or to change the target of inference from $\theta_0(t)$ to a \textit{projection} of $\theta_0(t)$ onto a finite-dimensional model $g(t; \beta)$. We refer to \cite{neugebauer2007nonparametric} and \cite{ai2018unified} for discussions of efficient estimation in the latter case. Estimation of $\beta$ in these settings is typically an easier statistical problem than estimating $\theta_0(t)$ fully nonparametrically. Another approach is to impose structural assumptions on the dose–response function. For instance, if it is known that the treatment cannot harm patients, one may impose a monotonicity assumption \citep{westling2020causal, westling2020unified}. Yet another approach is to select a candidate estimator of $\theta_0(t)$ by minimizing an estimate of its risk. The key insight is that, while $\theta_0(t)$ is generally not estimable at $n^{-1/2}$-rates, the integrated risk of a candidate estimator can be estimated at the parametric rate under certain conditions \citep{diaz2013targeted, van2003unified}.

In the context of nonparametric estimation of $\theta_0(t)$, \cite{newey1994kernel} derive sufficient conditions under which a two-stage kernel estimator is asymptotically normal and unbiased. Their estimator is of the plug-in type and takes the form $\widehat{\theta}(t) = n^{-1}\sum_{i=1}^n \widehat{\mu}(t, X_i)$, where $\widehat{\mu}(t,x)$ is a kernel-smoothed estimator of $\mu(t,x)$ based on a bandwidth $h$. Building on the estimator considered in \cite{newey1994kernel}, \cite{flores2007estimation} develop plug-in estimators for the maximum of $\theta_0(t)$ and the value of $t$ at which this maximum is attained. \cite{galvao2015uniformly} study estimation and testing of continuous treatment effects in general settings using inverse-probability-weighted estimators. Finally, \cite{singh2020reproducing} analyze plug-in estimators of general causal functions using reproducing kernel methods.

Conceptually aligned with the approach to dose-response estimation developed in this work are estimators derived from representing $\theta_0(t)$ as $\theta_0(t) =  \E\{\varphi(Z) \mid A = t\}$, where $\varphi(Z)$ is defined Eq. \eqref{eq:varphi}. This representation motivates estimators that regress the \textit{pseudo-outcome} $\varphi(Z)$ onto $A$. Because $\varphi(Z)$ depends on unknown nuisance functions, it must be estimated from the data. We thus refer to the regression of $\varphi(Z)$ on $A$ as a second-stage regression. The key observation is that $\varphi(Z)$ is constructed so that an estimator $\widehat\E_n\{\widehat\varphi(Z) \mid A = t\}$ can behave like the oracle $\widehat\E_n\{\varphi(Z) \mid A = t\}$ even when $\widehat\varphi(Z)$ converges to $\varphi(Z)$ at a slower rate than the rate at which $\widehat\E_n\{\varphi(Z) \mid A = t\}$ converges to $\theta_0(t)$. Here, the notation $\widehat{\E}_n(\cdot \mid A = t)$ refers to an estimator of the regression of $\widehat\varphi(Z)$ or $\varphi(Z)$ onto $A$, treating $\widehat\varphi(Z)$ as just an ordinary outcome. We conclude this section with a review of the use of $\varphi(Z)$ in the estimators proposed in \cite{kennedy2017nonparametric}, \cite{semenova2017debiased} and \cite{colangelo2020double}.

\subsubsection{Review of existing doubly-robust estimators \label{sec:lit_review_dr}}
We briefly review three doubly robust estimation strategies based on $\varphi(Z)$ defined in Eq.  \eqref{eq:varphi}, whose associated risk bounds decompose as \textit{oracle rate + second-order doubly robust remainder terms}. The doubly robust nature of these remainders is discussed below.

The quantity $\varphi(Z)$ depends on the collection of nuisance functions $\eta = \{p(A), \pi(A \mid X), \mu(A, X)\}$ and satisfies a Neyman-orthogonality condition:
 \begin{align*}
 	\partial r \E\{\varphi(Z; \eta + r \cdot(\overline\eta - \eta )) \mid A = a\}|_{r = 0} = 0 \text{ for all } a, \overline\eta.
 \end{align*}
 where $\eta$ denotes the true nuisance functions and $\eta + r \cdot (\overline\eta - \eta)$ denotes a perturbation in the direction $(\overline\eta - \eta)$.\footnote{More specifically, we have that \begin{align*}
	\frac{\partial}{\partial r_j} \E\left(\frac{Y - [\mu(A, X) + r_1\{\overline\mu(A, X) - \mu(A, X)\}]}{\pi(A \mid X) + r_2\{\overline\pi(A \mid X)- \pi(A \mid X)\}} \left[p(a) + r_3\{\overline{p}(A) - p(A)\}\right] \mid A = a\right)\Big|_{r_j = 0} = 0.
	\end{align*} 
 }
 This implies that the loss $\{\varphi(Z;\eta) - \theta(A)\}^2$ is universally Neyman-orthogonal in the sense that
 \begin{align*}
		& \partial r _2	\partial r_1 \E[\varphi(Z; \eta + r_2(\overline\eta - \eta)) - \theta(A) - r_1\{\theta(A) - \overline\theta(A)\}]^2|_{r_1 = r_2 = 0} \\
		& = -2\int \partial r \E\{\varphi(Z; \eta + r(\overline\eta - \eta)) \mid A = a\}_{r = 0}\{\theta(a) - \overline\theta(a) \} d\Pb(a) \\
		& = 0
 \end{align*}
 for any $\theta, \overline\theta$. See also Definition 2 and Assumption 1 in \cite{foster2019orthogonal}. This first-order insensitivity to perturbations of the nuisance functions explains why regressing $\widehat\varphi(Z)$ on $A$ can produce estimators of $\theta_0(t)$ whose remainder terms are second order.
\begin{remark}
	Constructing estimators satisfying Neyman-orthogonality conditions has a long history in statistics, albeit under different names. For example, in functional estimation, and, in particular, estimation of average treatment effects, estimators that are ``Neyman-orthogonal,'' ``bias-corrected,'' ``augmented-inverse-probability-weighted,'' or, more generally, constructed according to the ``double machine learning'' framework are all based on first-order functional Taylor expansions, also known as von-Mises expansions \citep{kennedy2022semiparametric}. In fact, underlying Neyman orthogonality is a first-order expansion of the target estimand $\psi(\Pb)$, viewed as a function of the unknown distribution $\Pb$, around an estimator $\widehat\Pb$ of $\Pb$. If the derivative term, say $\psi^{'}(\Pb - \widehat\Pb; \widehat\Pb)$, exists then the estimator consisting of the (estimated) derivative term plus the initial estimator $\psi(\widehat\Pb)$ should exhibit second-order error rates. For smooth functionals, the derivative term can be written as $\psi^{'}(\Pb - \widehat\Pb; \widehat\Pb) = \int \phi(z; \widehat\Pb) d\Pb(z)$, where $\phi(z)$ is the influence function and is mean-zero. For more complex parameters, such as the dose-response curve, this representation is generally not possible. However, one may try to express the derivative as an integral with respect to the conditional distribution of the observations given, for example, the treatment. 
\end{remark}

\cite{kennedy2017nonparametric} show that, when the second stage regression $\widehat\E_n(\widehat\varphi(Z) \mid A = a)$ is a local linear regression, the oracle rate is attained under negligibility of a remainder term, which can be upper bounded as
\begin{align}\label{eq::edward_r}
	\sup_{a: |a - t| \leq h}\| \pihat(a \mid X) - \pi(a \mid X)\| \|\muhat(a, X) - \mu(a, X)\|,
\end{align}
where $h$ is the bandwidth used in the second-stage regression. A similar requirement appears in \cite{colangelo2020double}. Notice that this term is \textit{second order}, as it takes the form of a product of estimation errors. As such, one nuisance function may be estimated at a relatively slow rate, provided the other is estimated sufficiently quickly to compensate. This property is often referred to as \textit{rate double robustness}. Borrowing terminology from \cite{kennedy2020optimal}, we refer to estimators constructed from $\varphi(Z)$ as \textit{DR-learners} (doubly robust learners), and say that their remainder terms exhibit \textit{doubly robust error}.

\cite{semenova2017debiased} study a DR-learner, in which the second-stage regression is implemented using a sieve estimator. Their estimator uses \textit{cross-fitting}, whereby, for a given fold $k$, the nuisance functions are estimated on all folds but $k$, and the second-stage regression is computed using observations from $k$. This construction bypasses the need to impose Donsker conditions on the nuisance functions.

The approach of \cite{colangelo2020double} is different. Instead of regressing $\widehat\varphi(Z)$ onto $A$, the estimator is
\begin{align}\label{eq:colangelo}
\thetahat_1(t) = \frac{1}{n}\sum_{i = 1}^n \left[\frac{K_{ht}(A_i)\{Y_i - \muhat(t, X_i)\}}{\pihat(t \mid X_i)} + \muhat(t, X_i)\right].
\end{align}
This estimator still enjoys second order rates, but it is less clear how it adapts to different levels of smoothnesss of $\theta_0(t)$. That is, their error rates may be of the form ``oracle + second-order terms'' only in certain smoothness regimes for $\theta_0(t)$. This is in contrast to estimators based on regressing $\widehat\varphi(Z)$ on $A$, which would behave like an oracle, and thus adapt to the smoothness of $\theta_0(t)$, as long as the second-order remainder terms are negligible. Their analysis focuses on low-smoothness regimes; viewed as a function of $a$, they assume that the joint density of the observations is three-times differentiable. Notice that this implies that both the outcome regression $a \mapsto \mu(a, x)$ and the conditional density $a \mapsto \pi(a \mid x)$ are three-times differentiable. In practice, however, it could be that $a \mapsto \mu(a, x)$ and thus the dose-response curve are smoother than $a \mapsto \pi(a \mid x)$. Our $m^{\text{th}}$-order estimator is an extension of \eqref{eq:colangelo} and appears to track the smoothness of the dose-response only in cases when this is no-greater than the smoothness of $a \mapsto \pi(a \mid x)$, which appears to be consistent with the results in \cite{colangelo2020double}.
\subsection{Our main contributions}
\color{black}
In this work, we study two approaches to dose-response estimation: one based on regressing an estimate of the pseudo-outcome $\varphi(Z)$ (Eq. \eqref{eq:varphi}) onto $A$ (Section \ref{sec:dr_est}) and one based on higher-order corrections (Section \ref{sec:higher_order}). For the first approach, building upon the work of \cite{kennedy2017nonparametric, semenova2017debiased, foster2019orthogonal} and \cite{kennedy2020optimal}, we analyze a broad class of empirical risk minimization estimators of $\theta_0(t)$ that correct the first-order bias of plug-in estimators. We study both global estimation of the function over its entire support and local estimation at a fixed treatment level.  The resulting estimators enjoy the double-robustness property, and we provide an explicit characterization of their remainder terms. This characterization implies faster convergence rates than those directly obtainable from the results of \cite{foster2019orthogonal} whenever the treatment and outcome models are estimated at different rates.

In the context of H\"{o}lder smooth models, our second and main contribution is to show how the convergence rates of these estimators can be substantially improved using kernel-smoothed, approximate higher-order influence functions (HOIFs) \citep{robins2008higher, robins2009quadratic, robins2017higher}. To the best of our knowledge, our higher-order estimator represents the first application of HOIFs to the estimation of a dose-response curve. Moreover, the resulting estimator achieves the fastest convergence rate currently available in the model considered.

Finally, extending the work of \cite{bonvini2022MSM} on sensitivity analysis in marginal structural models, the supplementary material describes a simple yet flexible framework for assessing the impact of potential unmeasured confounders on the dose-response estimates. We then analyze the performance of DR-learners for the bounds on $\theta_0(t)$ obtained under this sensitivity model.

\color{black}
\section{Doubly-robust estimators}\label{sec:dr_est}
\subsection{General doubly-robust estimation procedure}
In this section, we build upon the general procedure proposed in \cite{foster2019orthogonal}, as well as the work of \cite{kennedy2020optimal} on heterogeneous (binary) treatment effect estimation, to expand the class of estimators of $\theta_0(t)$ whose remainder terms are second-order and doubly robust. The work by \cite{foster2019orthogonal} already yields estimators with second-order remainder terms, but their analysis characterizes the rates in terms of $\Norm{\widehat\eta - \eta}_{\mathcal{F}}$, where $\Norm{f}_{\mathcal{F}}$ denotes a norm on the function space $\mathcal{F}$ in which the nuisance functions $\eta$ lie. We tailor their analysis to the dose-response setting and show that it is possible to obtain estimators whose remainders are also doubly robust. Establishing this property is particularly important when the nuisance estimators converge at different rates, since the product of the errors would be of smaller order than the sum of the squared errors.

Let $Z_1^n, Z_2^n$ and $Z_3^n$ denote three independent samples of size $n$. We will work with estimates of the pseudo-outcome $\varphi(Z_j)$ of the form
\begin{align*}
\widehat\varphi(Z_j) = \widehat{w}(A_j, X_j)\{Y_j - \muhat(A_j, X_j)\} + \frac{1}{n}\sum_{i = 1}^n \muhat(A_j, X_i),
\end{align*}
where $\muhat(a, x)$ and $\widehat{w}(a, x)$ are estimated using observations in $Z_1^n$, the observations $(X_i)_{i = 1}^n$, appearing in the second term, belong to $Z_2^n$ and $Z_j$ belongs to $Z_3^n$. An alternative approach, taken in \cite{semenova2017debiased}, is to consider only two samples, say $Z_1^n$ and $Z_2^n$, and compute
\begin{align*}
\widehat\varphi(Z_j) = \widehat{w}(A_j, X_j)\{Y_j - \muhat(A_j, X_j)\} + \frac{1}{n-1}\sum_{i \neq j}^n \muhat(A_j, X_i)
\end{align*}
for $Z_j$ and $(X_i)_{i = 1}^n$ in the same sample $Z_2^n$. We proceed by considering three separate samples to simplify the analysis of all our estimators, as we have $\widehat\varphi(Z_k) \ind \widehat\varphi(Z_l) \mid (Z_1^n, Z_2^n)$ for $k \neq l$. The roles of $Z_1^n$, $Z_2^n$ and $Z_3^n$ can be swapped, which results in three estimators of $\theta_0(t)$. One can then take their average as the final estimator. From a sample of iid observations, it is possible to obtain separate independent samples simply by randomly splitting the data into sub-samples. To keep the notation as light as possible, we analyze the theoretical properties of the estimators based a single split into three subsamples. An alternative strategy could be to divide the sample into $S$ folds. Using all but observations in fold $s$, one would estimate the nuisances $\omega(a, x)$ and $\mu(a, x)$. Then, an estimate of $\int \widehat\mu(a, x) p(x) dx$ would be computed using half of the observations in fold $s$, while the remaining half would be used to evaluate $\widehat\varphi(Z_j)$ and regress it on $A$. Repeating this procedure for each fold, one obtains $S$ estimates (or $2S$ estimates if the sub-folds in fold $s$ are swapped) that would be averaged. The further split of fold $s$ into two sub-splits could be replaced by the leave-one-out average considered in \cite{semenova2017debiased} mentioned above. \textcolor{black}{We expect the same results to hold for constructions based on multiple splits, akin to the one just described, although the arguments would need to be adjusted to account for the resulting dependence structure if the leave-one-out approach is employed.}

Our estimation procedure is summarized in the following algorithm. 
\begin{algorithm} \label{alg:first_order}
	Let $Z_1^n$, $Z_2^n$ and $Z_3^n$ denote three independent samples of $n$ iid observations of $Z = (Y, A, X)$. 
	\begin{enumerate}
		\item Nuisance training
		\begin{itemize}
			\item Using only observations in $Z_1^n$, estimate $\mu(A, X)$ with $\muhat(A, X)$ and $w(A, X)$ with $\widehat{w}(A, X)$;
			\item Using only observations in $Z_2^n$, estimate $m(a) = \int \mu(a, x) p(x) dx$ with
			$\mhat(a) = n^{-1}\sum_{i = 1}^n \muhat(a, X_i)$.
		\end{itemize}
		\item Pseudo-outcome construction: using observations in $Z_3^n$, construct the pseudo-outcome 
		\begin{align*}
			\widehat\varphi(Z) = \widehat{w}(A, X)\{Y - \muhat(A, X)\} + \mhat(A)
		\end{align*} 
		\item Second stage regression, either of the following:
		\begin{itemize}
		    \item[(a)] Global empirical-risk-minimization: Define $\thetahat$ to be the empirical risk minimizer 
		\begin{align*}
			\thetahat_{\text{erm}} = \argmin_{\theta \in \Theta} \frac{1}{n}\sum_{i \in Z_3^n}\{\widehat\varphi(Z_i) - \theta(A_i)\}^2
		\end{align*}
		where $\Theta$ is some function class.
		\item[(b)] DR-learner based on linear smoothing: Define
		\begin{align*}
		    \thetahat_{\text{ls}}(t) = \frac{1}{n} \sum_{i \in Z_3^n} W_i(t; A^n) \widehat\varphi(Z_i)
		\end{align*}
		where $W_i(t; A^n)$ are weights depending on the evaluation point $t$, $A^n = (A_1, \ldots, A_n) \subset Z_3^n$, and the choice of the linear smoother.
		\end{itemize}
		\item (Optional) Cross-fitting: swap the role of $Z_1^n$, $Z_2^n$ and $Z_3^n$ and repeat steps 1 and 2. Use the average of the estimators as an estimate of $\theta$. 
	\end{enumerate}
\end{algorithm}
In Algorithm \ref{alg:first_order}, Step 3b, the estimator of the dose-response consists of a weighted average of the pseudo-outcome $\widehat\varphi(Z)$, with weights depending on the choice of the second-stage regression estimators. Series methods, local polynomials, smoothing splines, and kernel ridge regression estimators are examples that can be written as linear smoothers in the form above for some choices of the weights; see also Example \ref{ex:poly}. In Proposition \ref{thm:dr}, we focus on the case where the weights localize the average around $A = t$ in order to derive a bound on the estimator's pointwise mean-square-error.

In the following two sections, we give error bounds for a procedure that generalizes Algorithm \ref{alg:first_order}. With some abuse of notation, we re-define $\theta_0(v)$ to denote a generic regression function of some pseudo-outcome $f(Z)$, which will be estimated by $\fhat(Z)$, onto some random variable $V \subset Z$. That is, $\theta_0(v) \equiv \E\{f(Z) \mid V = v\}$.  One can see that Algorithm \ref{alg:first_order} fits exactly this framework where $f(Z) = \varphi(Z)$, $V=A$ and $\theta_0(t)$ is the dose-response function. The estimator is $\thetahat(v) = \Ehat_n\{\fhat(Z) \mid V = v\}$, where $\Ehat_n(\cdot \mid V = v)$ is either computed as in Step 3a or 3b from a sample independent of that used to construct $\fhat(\cdot)$.

The bound on the integrated square risk (Proposition \ref{thm:erm}) and that on the pointwise risk (Proposition \ref{thm:dr}) will involve a bias term, \color{black}$\rhat(v) := \int \fhat(z) d\Pb(z \mid V = v)  - \theta_0(v)$, \color{black}that would need to be analyzed on a case-by-case basis. Lemma \ref{lemma:rhat} below shows that, for $f(Z) = \varphi(Z)$, $\rhat(v)$ is rate doubly-robust because it is a product of the error incurred in estimating $\omega$ times the error incurred in estimating $\mu$ (plus a smaller $O_\Pb(n^{-1/2})$ term).
\subsection{Upper bound on the risk of the global ERM-based estimator}
We start by considering estimating $\theta_0(t)$ via empirical loss minimization as in Algorithm \ref{alg:first_order} (a). We view this as a ``global'' method, as we estimate the function on its entire support. The error bound we describe in this section will be on the square loss and will be a specialization of the results described in \cite{foster2019orthogonal} and \cite{wainwright2019high}. \cite{foster2019orthogonal} provides a general framework for doing empirical risk minimization in the presence of nuisance components. Here, we take their approach and find that the oracle rate is achievable if $\E\|\widehat{r}\|^2$ is simply of smaller order. In particular, as shown in Lemma \ref{lemma:rhat} below, if the orthogonal signal $\varphi(Z)$ is used, $\widehat{r}$ consists of a product of errors, as opposed to simply being of second order, and thus the bound on the MSE of our procedure can be viewed as a refinement upon the bound from \cite{foster2019orthogonal}. 

The next proposition provides a bound on the error incurred by an estimator that uses an estimated outcome $\fhat(Z)$ in place of the true (unobservable) outcome $f(Z)$, when doing empirical risk minimization with the square loss to estimate a regression function $\theta_0(v) = \E\{f(Z) \mid V = v\}$. Let $\theta^* = \argmin_{\theta \in \Theta} \Norm{\theta - \theta_0}$ and $\Theta^* = \{\theta  - \theta^*: \theta \in \Theta\}$. For $\epsilon_1, \ldots, \epsilon_n$ iid Rademacher random variables (i.e., $P(\epsilon_i = -1) = P(\epsilon_i = 1) = 0.5$), which are independent of the sample, the local Rademacher complexity is defined as (see, e.g., Chapter 14 in \cite{wainwright2019high}):
\begin{align*}
	\mathcal{R}_n(\Theta^*, \delta) = \E\left\{\sup_{g \in \Theta^*: \|g\| \leq \delta} \left| \frac{1}{n} \sum_{i = 1}^n  \epsilon_i g(V_i) \right| \right\}.
\end{align*}
Let $\delta_n$ be defined as the smallest solution that satisfies
\begin{align}
	& \mathcal{R}_n(\Theta^*, \delta_n) \leq \delta_n^2 \quad \text{ and } \quad \delta_n^2 \gtrsim \frac{\log\log(n)}{n} \lor \frac{1}{2n}. \label{delta_n_def}
\end{align}
Following \cite{foster2019orthogonal}, we consider the case where $\Theta^*$ is \textit{star-shaped}. A class is star-shaped around the origin if, for any $\widetilde\theta \in \Theta^*$ and $\alpha \in [0, 1]$, it is the case that $\alpha \cdot \widetilde\theta \in \Theta^*$. For example, a convex set is star-shaped.
\begin{proposition}\label{thm:erm}
Consider two independent samples, $D^n = (Z_{01}, \ldots, Z_{0n})$ and $Z^n = (Z_1, \ldots, Z_n)$, consisting of $n$ iid copies of some generic observation $Z$ distributed according to $\Pb$. Let $V$ denote a generic random variable such that $V \subset Z$. Suppose $\fhat(\cdot)$ is constructed using only observations in $D^n$. Consider the estimator
	\begin{align*}
		\widehat\theta_{\rm{erm}} \equiv \argmin_{\theta \in \Theta} \frac{1}{n} \sum_{i = 1}^n \{\fhat(Z_i) - \theta(V_i)\}^2.
	\end{align*}
Suppose  $\Theta^*$ is star-shaped and $S \equiv \sup_{z \in \mathcal{Z}} |\fhat(z)| \ \lor \sup_{\theta \in \Theta} \| \theta \|_\infty$ is finite. Then, 
	\begin{align*}
		\E(\Norm{\thetahat_{\rm{erm}} - \theta_0}^2) \lesssim \Norm{\theta^* - \theta_0}^2 + \delta_n^2 + \E(\Norm{\rhat}^2)
	\end{align*}
	where $\Norm{f}^2 = \int f^2(z) d\Pb(z)$ and $\delta_n$ is defined in \eqref{delta_n_def}.
\end{proposition}
The error bound from Proposition \ref{thm:erm} takes the form of an oracle rate plus a term involving $\rhat$, which is controlled by Lemma \ref{lemma:rhat} when $f(Z) = \varphi(Z)$.

The assumptions underlying Proposition \ref{thm:erm} are rather mild. Appendix D in \cite{foster2019orthogonal} and Chapters 13 and 14 in \cite{wainwright2019high} describe common classes of functions for which the proposition applies, e.g., linear functions with constraints on the coefficients, functions satisfying Sobolev-type constraints or Reproducing Kernel Hilbert spaces. In order to apply Proposition \ref{thm:erm}, the class of functions considered has to be star-shaped. Following the discussion on page 424, Chapter 13, in \cite{wainwright2019high}, if the star-shaped condition is not met, statements similar to the proposition below can be established for $\delta_n$ defined in terms of the star-hull of $\Theta^*$ at zero, defined as
\begin{align*}
	\text{star}(\Theta^*) = \left\{\alpha\cdot(\theta - \theta^*), \ \theta \in \Theta, \alpha \in[0, 1] \right\}.
\end{align*} 
The boundedness assumption on $\widehat\varphi(Z)$ and $\Theta$ is used in various places in the proof, including in ensuring that the square-loss is globally Lipschitz; we expect this assumption to hold when the observations are bounded. Finally, the lower bound on $\delta_n$ from \eqref{delta_n_def} should often be satisfied; for instance, $\delta^2_n \geq 1 / (2n)$ as long as $\Theta^*$ contains the constant function $\theta(v) = 1$.\footnote{To see this, suppose that, for the sake of contradiction, $\delta_n < 1/\sqrt{2n}$. To start, because $\Theta^*$ is star-shaped, we have $g(V) = \delta_n \in \Theta^*$ because $\theta(V) = 1 \in \Theta^*$ and $\delta_n \in [0, 1]$. Then, $\Norm{g} = \delta_n$ so that
		$$\mathcal{R}_n(\Theta^*, \delta_n) \geq \delta_n \E\left(\left|\frac{1}{n}\sum_{i = 1}^n \epsilon_i \right|\right) \geq \frac{\delta_n}{\sqrt{2n}} > \delta_n^2$$
		where the second inequality is an application of the Khintchine inequality.
		 This is a contradiction because $\delta_n$ satisfies $\mathcal{R}_n(\Theta^*, \delta_n)  \leq \delta_n^2$.}
\begin{example}[Orthogonal series, Examples 13.14 and 13.15 in \cite{wainwright2019high}]\label{ex:basis}
	Suppose $\theta_0(v)$ is $\alpha$-times differentiable with $\theta^{(\alpha)}(v)$ satisfying $\int \{\theta^{(\alpha)}(v)\}^2 d\Pb(v) \leq B$ for some constant $B$. Let $\{p_j\}_{j=1}^\infty$ be an orthonormal basis of $L_2(\Pb)$, such as the sine / cosine basis (see \cite{belloni2015some} for a discussion on different basis choices). Consider estimating $\theta_0$ via ERM over the function class 
	\begin{align*}
		\Theta(k, b) = \left\{\theta_c(\cdot): \sum_{j=1}^k p_j(\cdot)c_j, \ \sum_{j = 1}^k c_j^2 \leq 1,  \text{ and } | \theta_c(\cdot) | \  \leq b \right\}
	\end{align*}
Writing $\theta_0(v) = \sum_{j = 1}^\infty p_j(u)c_{0j}$, we have $\theta^*(v) = \sum_{j = 1}^k p_j(u)c_{0j}$ and $\Norm{\theta^*-\theta_0}^2=\sum_{j=k+1}^\infty c_{0j}^2$. It can be shown that $\Norm{\theta^*-\theta_0}^2 \leq k^{-2\alpha}$. Furthermore, the function class $\Theta^*(k) = \{\theta - \theta^*, \theta \in \Theta(k)\} = \Theta(k, 2b)$ is convex and thus star-shaped and can be shown to satisfy $\delta^2_n \lesssim k / n$. Thus, Proposition \ref{thm:erm} provides an upper bound on the mean-square error of the order
\begin{align*}
	\E(\Norm{\thetahat_{\text{erm}} - \theta}^2) \lesssim k^{-2\alpha} + \frac{k}{n} + \E(\Norm{\rhat}^2)
\end{align*}
If $k$ is chosen optimally, i.e. $k \asymp n^{1 / (2\alpha + 1)}$, Proposition \ref{thm:erm} shows that the oracle rate is attained as long as $\E(\Norm{\rhat}^2)$ is of order $O(n^{-2\alpha / (2\alpha + 1)})$.
\end{example}
\subsection{Upper bound on the risk of the linear smoothing-based estimator}
In this section, we consider a DR-learner based on linear smoothing, focusing on the case where the linear smoother is localized (cf. \cite{van2006statistical} and \cite{kennedy2020optimal} for heterogeneous effects of binary treatments). Regressing $\widehat\varphi(Z)$ on $A$ via local polynomial regression represents an archetype of a localized DR-learner. \cite{kennedy2017nonparametric} propose using generic learners to regress the estimated pseudo-outcome $\widehat\varphi(Z)$ on $A$ but only analyze local linear estimators. Thus, our next proposition is an extension to their work, in the spirit of analyzing more general linear smoothers. Theorem 1 and Proposition 1 in \cite{kennedy2020optimal} yield the following proposition.
\begin{proposition}\label{thm:dr}
Consider two independent samples, $D^n = (Z_{01}, \ldots, Z_{0n})$ and $Z^n = (Z_1, \ldots, Z_n)$, consisting of $n$ iid copies of some generic observation $Z$ distributed according to $\Pb$. Let $\textcolor{black}{V}$ denote a generic variable such that $V \subset Z$. Let $\theta_0(v) \equiv \E\{f(Z) \mid V = v\}$ and suppose $\fhat(\cdot)$ is constructed using only observations in $D^n$. Consider the following estimator:
\begin{align*}
\thetahat_{\rm{ls}}(t) = n^{-1}\sum_{i = 1}^n W_i(t; A^n)\fhat(Z).
\end{align*}
Further suppose that the following regularity conditions hold:
	\begin{itemize}
		\item Minimum variance: $\var\{f(Z) \mid V = v\} \geq c > 0$ for all $v \in \mathcal{V}$ and some constant $c$;
		\item Consistency of nuisance estimators: $\sup_z|\fhat(z) - f(z)| = o_\Pb(1)$;
		\item Localized weights: $n^{-1}\sum_{i = 1}^n |W_i| \ \leq C$, for some constant $C$, and there exists a neighborhood $N_t$ around $V = t$ such that $W_i(t; V^n) = 0$ if $V_i \not\in N_t$. 
	\end{itemize}
	Then, letting $\widetilde{\theta}_{\text{ls}}(t) = n^{-1}\sum_{i = 1}^n W_i(t; V^n)f(Z_i)$ denote the oracle estimator:
	\begin{align*}
		|\widehat{\theta}_{\rm{ls}}(t) - \theta_0(t)|  & \leq \left|\widetilde{\theta}_{\rm{ls}}(t) - \theta_0(t)\right| + \sup_{u \in N_t} \left|\rhat(u)\right| + \color{black}{o_\Pb\left( \sqrt{\E\left[\left\{\widetilde{\theta}_{\text{ls}}(t) - \theta_0(t)\right\}^2\right]} \right)}.
	\end{align*}
\end{proposition}
As discussed in \cite{kennedy2020optimal}, the assumptions underlying Proposition \ref{thm:dr} are easily satisfied for linear smoothers of the local polynomial regression variety. In particular, the weights of the local polynomial regression satisfy the assumptions (\cite{tsybakov2008introduction}, Lemma 1.3). This proposition follows from the results contained in \cite{kennedy2020optimal} that apply to general linear smoothers, e.g. it does not require the weights to be localized. The use of localized weights simplifies the analysis of the point-wise risk. In Section 3 of the supplement, we also provide a bound on the $L_p$ risk of $\widehat\theta_{\rm ls}(t)$.
\begin{example}(DR-Learner based on local polynomial regression)\label{ex:poly}
Suppose $\theta_0(t) \equiv \E\{f(Z) \mid A = t\}$ belongs to a H\"{o}lder class of order $\alpha$. A DR-learner can be based upon local polynomial regression of order $p$, where $p$ is the largest integer strictly less than $\alpha$. The weights are $$W_i(t; A^n) = s(t)^T \Qhat^{-1} K_{ht}(A_i)s(A_i)^T,$$ where $K(\cdot)$ is a kernel function, $\Qhat = \Pn\{s(A)s(A)^T\}$ and 
\begin{align*}
    s(a) = \begin{bmatrix} 1 & \frac{a - t}{h} & \ldots & \left(\frac{a - t}{h}\right)^p \end{bmatrix}^T.
\end{align*} 
A standard calculation (see, for example, \cite{tsybakov2008introduction}), yields that
\begin{align*}
\E\left[\left\{\widetilde{\theta}_{\text{ls}}(t) - \theta_0(t)\right\}^2\right] = O(n^{-2\alpha / (2\alpha + 1)}).
\end{align*}
This means that the oracle rate is attainable if $\sup_{u \in N_t} \rhat^2(u)  = O_\Pb(n^{-2\alpha / (2\alpha + 1)})$, 
which is essentially the same requirement as for the estimator based on empirical-risk-minimization, see Example \ref{ex:basis}. 
\begin{remark}[Inference]
    From Proposition \ref{thm:dr}, inference can be carried out in the oracle regime, i.e., under the assumption that $\sup_{u \in N_t} |\rhat(u)|$ is of smaller order than $|\widetilde\theta_{\text{ls}}(t) - \theta(t)|$. If this holds, all inference tools for standard local nonparametric regression can be used. For example, let $\sigma^2(t)$ be the asymptotic variance of the local polynomial $\widetilde\theta_{\text{ls}}(t)$, $\widehat\sigma^2(t)$ its consistent estimator and $b_h(t)$ the asymptotic smoothing bias. If $\sup_{u \in N_t}\widehat{r}(u) = o_\Pb((nh)^{-1/2})$, we have
    \begin{align*}
    \frac{\sqrt{nh}\{\widehat{\theta}_{\text{ls}}(t) - \theta(t) - b_h(t)\}}{\widehat\sigma(t)} \indist N(0, 1),
    \end{align*}
    as shown, for instance, in Section 4 of \cite{fan2018local}. Without undersmoothing or bias-correction, a Wald interval based on the asymptotic statement above will cover the smoothed dose-response curve $\E\{\widetilde\theta_{\text{ls}}(t)\}$, rather than $\theta_0(t)$. We refer to Section 5.7 in \cite{wasserman2006all} for additional discussion, as well as recent work by \cite{calonico2018effect} and \cite{takatsu2024debiased} on correcting the smoothing bias of local polynomial estimators.
\end{remark}
\end{example}
\subsection{Upper bound on the conditional bias of the estimated pseudo outcome}
Both Propositions \ref{thm:erm} and \ref{thm:dr} yield bounds on the MSE that are of the form the oracle rate plus a term of the order of $\rhat^2(t)$. Lemma \ref{lemma:rhat} shows that, when $\fhat(Z) = \widehat{\varphi}(Z)$, $\rhat(v)$ is second-order and doubly robust.
\begin{lemma}\label{lemma:rhat}
    With $\widehat{r}(t) = \E\{\widehat\varphi(Z) \mid A = t, D^n\} - \theta_0(t)$, it holds that
    \begin{align*}
    |\widehat{r}(t)| \ \lesssim \Norm{w - \widehat{w}}_t \Norm{\mu - \muhat}_t + \left|(\Pn^{(2)} - \Pb)\{\muhat(t, X)\}\right|,
    \end{align*}
    where $\Norm{f}^2_t = \int f^2(z) d\Pb(z \mid A = t)$ and $\Pn^{(2)}$ denotes an average over observations in sample $Z_2^n$. 
\end{lemma}
\begin{proof}
Recall that $\theta_0(t) = \E\{\varphi(Z) \mid A = t\}$. We have
\begin{align*}
& \Pb\{\muhat(t, X)\} - \theta_0(t) = \int w(t, x) \{\muhat(t, x) - \mu(t, x)\} d\Pb(x \mid A = t), \quad \text{and} \\
& \E\{\widehat\varphi(Z) \mid A = t, D^n\} = \int \widehat{w}(t, x) \{\mu(t, x) - \muhat(t, x)\} d\Pb(x \mid A = t) \\
& \hphantom{\E\{\widehat\varphi(Z) \mid A = t, D^n\} = } + \Pn^{(2)} \{\widehat\mu(t, X)\}.
\end{align*}
Adding and subtracting $\Pb\{\muhat(t, X)\}$ and applying the Cauchy-Schwarz inequality yield the result. 
\end{proof}
The result from Lemma \ref{lemma:rhat} shows that $|\rhat(t)|$ can be bounded by the product of the $L_2$ errors in estimating $w(a, x)$ and $\mu(a, x)$ (yielding rate double robustness) plus a centered sample average, which would generally be of the smaller order $O_\Pb(n^{-1/2})$ if, for instance, the second moment of $\muhat(t, X)$ (conditional on $Z_1^n$) is bounded. In this respect, this term is effectively asymptotically negligible in nonparametric models where the rate of convergence is of slower order than $n^{-1/2}$. 

Standard results are generally calculated for $L_2(d\Pb(a, x))$ errors defined by the joint distribution of $(A, X)$, for example
\begin{align*}
\Norm{w - \widehat{w}}^2 \equiv \int \{w(a, x) - \widehat{w}(a, x)\}^2 d\Pb(a, x).
\end{align*}
In this case, optimal convergence rates for estimating $\mu(a, x)$ are well-understood for many classes. For instance, if $\mu(a, x)$ belongs to a H\"{o}lder class of order $\gamma$, then minimax-optimal convergence rates in $L_2(d\Pb(a, x))$ are of order $n^{-2\gamma /(2\gamma + d + 1)}$, where $d = \dim(X)$. The bound from Lemma \ref{lemma:rhat} is actually on an $L_2$ error with weight given by the conditional density of $X$ given $A = t$. In most settings, we expect the more conventional rate based on $L_2(d\Pb(a, x))$ to match that based on $L_2(d\Pb(x \mid A = t))$. For example, Result 1 from \cite{colangelo2020double} shows that the rate in $L_2(d\Pb(x \mid A = t) p(t))$ matches that for the point-wise risk (in $(A, X)$) under a mild boundedness assumption. Alternatively, we note that one can always upper bound (up to constants) $\Norm{f}_t$ by the supremum norm $\Norm{f}_\infty$ and the rate for estimating a regression function in $L_\infty$ matches that for estimating the function in $L_2(d\Pb(a, x))$ up to log factors in H\"{o}lder-smoothness models \citep{tsybakov2008introduction}. 

There are fewer results available for conditional density estimation compared to regression estimation. Recently \cite{ai2018unified} proposed a method to estimate $w(a, x)$ directly that, under certain conditions, exhibits a convergence rate in $L^2_2$ of order $n^{-2\gamma /(2\gamma + d + 1)}$ if $w(a, x)$ is $\gamma$-smooth (see their Theorem 3). Alternatively, one can estimate $p(a)$ and $\pi(a \mid x)$ and compute their ratio to estimate $w(a, x)$. We refer to \cite{colangelo2020double} for a discussion on ways to estimate $\pi(t \mid x)$. In particular, one approach is to estimate $\E\{G_{h_1 t}(A) \mid X = x\}$, where $G(u)$ is some kernel and $h_1$ some bandwidth of choice. As a third approach, because $w(a, x) = p(a)p(x) / p(a, x)$, one can estimate the marginals $p(x)$ and $p(a)$ and the joint density $p(a, x)$ and take the ratio as an estimate of $w(a, x)$. Estimating a joint density of $d$ variables that belongs to a H\"{o}lder-class of order $\gamma$ can be done with error scaling as $n^{-2\gamma /(2\gamma + d)}$. Thus, the MSE of this ratio would trivially be upper bounded by the MSEs for estimating $p(a, x)$, $p(x)$ and $p(x)$, which would depend on their respective smoothness levels. 

As investigated in more detail in the next section, an interesting setting is where $a \mapsto \mu(a, x)$ has a different smoothness level than $x \mapsto \mu(a, x)$, where we expect the former to match the smoothness of the dose-response $\theta_0(a)$ in many applications. This is an example of anisotropic regression. The optimal rate for estimating a $d+1$-dimensional regression in a H\"{o}lder class where each coordinate has its own level of smoothness $\gamma_j$ is of order $n^{-2\gamma / (1 + 2\gamma)}$, where $\gamma$ satisfies $\gamma^{-1} = \sum_{j = 1}^{d+1} \gamma_j^{-1}$ \citep{hoffman2002random, bertin2004asymptotically}. If $\mu(a, x)$ is in an anisotropic H\"{o}lder class of order $(\alpha, b, \ldots, b)$, the rate simplifies to $n^{-2b / (2b + b / \alpha + d)}$, where $d = \dim (X)$. If the treatment $A$ is categorical or $\alpha$ is much larger than $b$, the rate is essentially $n^{-2b / (2b + d)}$, i.e. the optimal rate for estimating a $d$-dimensional regression function that is $b$-smooth. In a similar fashion, we may think of $a \mapsto \pi(a \mid x)$ and $x \mapsto \pi(a \mid x)$ as having different smoothness levels; optimal convergence rates in this context typically depend too on the harmonic means of the smoothness levels of each coordinate \citep{efromovich2007conditional}. 

Finally, we conclude this section by deriving sufficient conditions under which the bounds from Propositions \ref{thm:erm} and \ref{thm:dr} yield the oracle rate. We focus on the case where $\theta_0(t)$ is H\"{o}lder of order $\alpha$, and $w$ and $\mu$ belong to anisotropic H\"{o}lder classes such that the smoothness with respect to the treatment variable is $\alpha$, while that with respect to the covariates $X$ is $s$. In this setting, the oracle mean square error rate is $n^{-1/(1 + 1/(2\alpha))}$. Based on the bounds on the mean square error derived above, the estimators $\widehat\theta_{\text{erm}}$ and $\widehat\theta_{\rm ls}$ behave like oracle estimators provided that
	\begin{align}\label{eq:dr_oracle_bar}
		\frac{1}{2}\left( 1 + \frac{1}{2\alpha}  + \frac{d}{2s} \right) \leq 1 + \frac{1}{2\alpha} 
		\quad \iff \quad 
		s \geq d \cdot \frac{\alpha}{2\alpha + 1}.
	\end{align}
	In Section \ref{sec:higher_order}, we will show that this oracle efficiency requirement can be relaxed, under certain conditions, by using higher-order estimators.
\begin{remark}
    Our discussion on the rates attained by the doubly-robust estimators discussed in Section \ref{sec:dr_est} is driven by the bound computed in Lemma \ref{lemma:rhat}. If $\muhat$ and $\widehat{w}$ are designed to optimally estimate $\mu$ and $w$, e.g. by selecting tuning parameters to minimize estimates of their MSEs, then generally the bound based on the Cauchy-Schwarz inequality is the best available. However, there are other techniques, such as particular forms of sample splitting coupled with undersmoothing, whereby the nuisance functions are estimated optimally with respect to the target of inference, and so the selected tuning parameters for the nuisance estimators may not minimize their corresponding MSEs. This approach has favorable theoretical properties, see e.g. \cite{newey2018cross, mcclean2024double, kennedy2020optimal}, although it can be challenging to implement in practice. We leave studying undersmoothing in the context of continuous treatments for future work. 
\end{remark}
\begin{remark}
    Compared to Theorem 2 in \cite{kennedy2017nonparametric}, Theorem \ref{thm:dr} and Lemma \ref{lemma:rhat} provide the same error bound, but under substantially weaker conditions. Sample-splitting circumvents the need to impose Donsker-type conditions on the nuisance functions' classes in the form of bounded uniform entropy integrals. Moreover, the use of local polynomial regression allows the estimator to track the smoothness of $\theta(t)$, thereby achieving the oracle rate in higher smoothness regimes (provided that the remainder term is negligible). 
\end{remark}
\section{Higher-order estimators}\label{sec:higher_order}
\subsection{Preliminaries}
Inspired by the seminal work of \cite{robins2008higher, robins2009quadratic, robins2017higher}, in this section we investigate the use of higher-order influence functions (HOIFs) for estimating continuous treatment effects. To the best of our knowledge, this is the first application of HOIFs in this context. We show that, under suitable smoothness conditions, higher-order estimators can improve the convergence rates of the estimators from Section \ref{sec:dr_est}.

For an introduction to higher-order influence functions, we refer to the main papers \citep{robins2009quadratic, robins2017higher} and give a brief overview here. Informally, an $m^{\text{th}}$-order estimator of a functional $\chi(p)$ (where $p$ is the density of the observations) takes the form, for some unknown function $\varphi_j$ estimated by $\widehat\varphi_j$ and whose specification depends on $\chi(p)$:
\begin{align}\label{gen_u_expansion}
\widehat{\chi}(p) = \chi(\widehat{p}) + \sum_{j = 1}^m \mathbb{U}_n\{ \widehat\varphi_j(Z_1, \ldots, Z_j)\},
\end{align}
where $\Un$ is the $U$-statistic measure so that
\begin{align*}
    \mathbb{U}_n\{\varphi_j(Z_1, \ldots, Z_j)\} = \frac{1}{n (n -1) \cdots (n - j + 1)} \sum_{1 \leq i_1 \neq i_2 \ldots \neq i_j \leq n} \varphi_j(Z_{i1}, \ldots, Z_{ij}).
\end{align*}
Letting $\Pb^j\{f(Z_1, \ldots, Z_j)\}= \int f(z_1, \ldots, z_j)d\Pb(z_1) \ldots d\Pb(z_j)$ denote the corresponding population measure, this implies an expansion:
\begin{align*}
\widehat{\chi}(p) - \chi(p) & = \chi(\widehat{p}) - \chi(p) + \sum_{j = 1}^m \Pb^j\{ \widehat\varphi_j(Z_1, \ldots, Z_j)\} \\
& + \sum_{j = 1}^m (\mathbb{U}_n - \Pb^j)\{ \widehat\varphi_j(Z_1, \ldots, Z_j)\}.
\end{align*}
Following \cite{robins2009quadratic},  \cite{van2014higher}, \cite{robins2017higher}, if $\varphi_j$ is chosen such that $\Pb^j\{\widehat\varphi_j(Z_1, \ldots, Z_j)\}$ acts as the $j^{\text{th}}$-order term in the functional Taylor expansion of $p \mapsto \chi(p)$ around $\widehat{p}$,   then 
\begin{align*}
\chi(\widehat{p})+ \sum_{j = 1}^m \Pb^j\{ \widehat\varphi_j(Z_1, \ldots, Z_j)\}  - \chi(p) = O(d(p - \phat)^{m + 1})
\end{align*}
for some distance $d(\cdot)$. The quantity $\varphi_j$ is referred to as the $j^{\text{th}}$-order influence function of $\chi(p)$. If the following holds
\begin{align*}
\var\left[\sum_{j = 1}^m (\mathbb{U}_n - \Pb^j)\{ \widehat\varphi_j(Z_1, \ldots, Z_j)\}\right] = O(n^{-1}),
\end{align*}
then this calculation would suggest that $\widehat{\chi}(p)$ is always root-$n$ consistent if $m$ is large enough. However, higher order influence functions do not exist for many functionals of interest, including the average treatment effect of a binary treatment. In our setting, the dose-response does not possess influence functions of \textit{any order}, in nonparametric models.While this means that generally it is not possible to construct root-$n$ consistent estimators, we show that an $m^{\text{th}}$-order estimator of $\theta_0(t)$ of the form \eqref{gen_u_expansion} can outperform the doubly-robust estimators from Section \ref{sec:dr_est} under certain smoothness conditions. Our estimator is tailored to models where $a \mapsto \mu(a, x)$ and $a \mapsto \pi(a \mid x)$ are H\"{o}lder smooth of order $\alpha$ and $\beta$, respectively.
\subsection{Notation}
Before describing our $m^{\text{th}}$-order estimator of the dose-response, we need to introduce some additional notation. Let $K(a)$ denote a kernel of order $l$, where $l$ is the largest integer strictly less than $\alpha \land \beta$. That is, $\int K(a) a^j da = 0$, for $j = 1, \ldots, l$, and $\int K(a) da = 1$ (see, e.g., Chapter 1 in \citep{tsybakov2008introduction}). Let $b(x)$ denote a $k$-dimensional vector representing an orthonormal basis of some user-specified finite-dimensional space; as clarified below, the space will hopefully well-approximate the space of functions where $\widehat\mu(t, x) - \mu(t, x)$ and $1 / \pihat(t \mid x) - 1 / \pi(t \mid x)$ reside. Define
\begin{align*}
\Pi_{i, j} \equiv \Pi(x_i, x_j) = b(x_i)^T \Omega^{-1} b(x_j)
\end{align*} 
where, for $g(x) = \int K_{ht}(a)p(a, x)da$:
\begin{align*}
    \Omega = \int b(x) b(x)^T g(x) dx.
\end{align*}
Thus, provided that $g(x)$ is positive and bounded away from zero and infinity, $\Pi_{i, j}$ is effectively the kernel of an orthogonal projection in $L_2(g)$ onto a $k$-dimensional subspace. That is, for a function $f(x)$, $\int \Pi(x_i, x) f(x) g(x) dx = b(x_i)^T \beta^*$, where $\beta^*$ solves the minimization problem
\begin{align*}
\beta^* = \argmin_{\beta \in \R^k} \int \left\{ f(x) - b(x)^T\beta\right\}^2 g(x) dx.
\end{align*}
The kernel $\Pi(x_i, x_j)$ has to be estimated in practice because $g(x)$ depends on the true density $p(a, x)$. When $X$ is multivariate, the basis can be taken to be the tensor product basis. Following \cite{robins2017higher}, by a slight abuse of notation, we will denote the projection operator associated with the kernel above using the same symbol $\Pi$. This way, we have $\Pi(f)(x_i) = \int \Pi(x_i, x) f(x) g(x) dx$. We also define $I$ to be the identity operator so that $(I - \Pi)(f)$ denotes the residual function, i.e., $(I - \Pi)(f)(x_i) = f(x_i) - b(x_i)^T \beta^*$.
\subsection{The estimator}\label{sec:higher_order_est}
In this section, we describe an estimator of $\theta_0(t)$ based on approximate, $m^{\text{th}}$-order influence functions. Define the first approximate influence function:
\begin{align*}
	& f_0(Z) = \frac{K_{ht}(A)\{Y - \mu(t, X)\}}{\pi(t \mid X)} + \mu(t, X)
\end{align*}
and the functions
\begin{align*}
	& f_1(Z) = K_{ht}(A)\{Y - \mu(A, X)\}  \\
	& f_2(Z) = \frac{K_{ht}(A)}{\pi(A \mid X)} - 1
\end{align*}
The function $f_0(Z)$ is a sum of a residual term involving $Y - \mu(t, X)$ and the outcome model $\mu(t, X)$. If $A$ was binary and $K_{ht}(A) = A$, $f_0(Z)$ would be exactly the influence function of $\int \mu(1, x) d\Pb(x)$, which equals $\E(Y^1)$ under standard causal assumptions. The terms $f_1(Z)$ and $f_2(Z)$ are kernel-weighted residuals; $f_2(Z)$ is a residual term in the sense that $\E\{K_{ht}(A) / \pi(A \mid X) \mid X\} = 1$ whenever $\int K_{ht}(a) da = 1$.

The $m^{\text{th}}$-order estimator of $\theta_0(t)$ that we study is
\begin{align*}
\widehat{\theta}_m(t) = \Pn\{\widehat{f}_0(Z)\} + \sum_{j = 2}^m \mathbb{U}_n \{\widehat\varphi_j(Z_1, \ldots, Z_j)\}
\end{align*}
where
\begin{align*}
	& \varphi_j(Z_1, \ldots, Z_j) = (-1)^{j-1} \sum_{S \subset \{1, \ldots, j\}} (-1)^{j - |S|} \E\left\{\overline{\varphi}_j(Z_1, \ldots, Z_j) \mid Z_i, i \in S\right\} \\
	& \overline{\varphi}_j(Z_1, \ldots, Z_j) = f_1(Z_1) \Pi_{1,2}K_{ht}(A_2)  \cdots \Pi_{j-2, j-1}K_{ht}(A_{j-1})\Pi_{j-1, j}f_2(Z_j)
\end{align*}
are the $m^{\text{th}}$-order approximate influence functions. Notice that $\varphi_j(Z_1, \ldots, Z_j)$ is simply the degenerate version of $\overline{\varphi}_j(Z_1, \ldots, Z_j)$, which ensures that $$\int \varphi_j(z_1, \ldots, z_j) d\Pb(z_i) = 0$$ for every $i$ and $(z_l: l \neq i)$. In addition, it holds that
\begin{align*}
& \int \Pi(x_{i-1}, x_i) K_{ht}(a_i) \Pi(x_i, x_{i + 1}) d\Pb(z_i) \\
& = b(x_{i-1})^T \Omega^{-1} \int b(x_i)b(x_i)^TK_{ht}(a_i) d\Pb(z_i) \Omega^{-1} b(x_{i + 1}) \\
& = \Pi(x_{i-1}, x_{i + 1})
\end{align*}
and, by degeneracy of $f_1(z)$ and $f_2(z)$, 
\begin{align*}
    \int f_1(z_1) \Pi(x_1, x_2) d\Pb(z_1) = \int \Pi(x_{j - 1}, x_j) f_2(z_j) d\Pb(z_j) = 0.
\end{align*}
This means that the first few approximate HOIFs take the form:
\begin{align*}
& \varphi_2(Z_1, Z_2) = - f_1(Z_1) \Pi_{1,2} f_2(Z_2), \\
& \varphi_3(Z_1, Z_2, Z_3) = f_1(Z_1) \Pi_{1, 2} K_{ht}(A_2)\Pi_{2, 3} f_2(Z_3) - f_1(Z_1) \Pi_{1, 3} f_2(Z_3), \\
& \varphi_4(Z_1, Z_2, Z_3, Z_4) = \\
& = - f_1(Z_1) \Pi_{1, 2} K_{ht}(A_2)\Pi_{2, 3} K_{ht}(A_3)\Pi_{3, 4}f_2(Z_4) \\
	& \hphantom{= \{  - } + f_1(Z_1) \Pi_{1, 2} K_{ht}(A_2)\Pi_{2, 4}f_2(Z_4) + f_1(Z_1) \Pi_{1, 3}K_{ht}(A_3)\Pi_{3, 4}f_2(Z_4) \\
	& \hphantom{= \{  - } - f_1(Z_1) \Pi_{1, 4}f_2(Z_4).
\end{align*}
\color{black}
By construction, we have that $\E\{\varphi_2(Z_1, Z_2) \mid Z_i\} = 0$ for $i = 1, 2$. Similarly, $\E\{\varphi_3(Z_1, Z_2, Z_3) \mid Z_i, Z_j\} = 0$ for $i, j = 1, 2, 3$. The other terms satisfy analogous degeneracy statements.

For parameters like the counterfactual mean defined by a discrete treatment, there are at least two ways to view the higher-order corrections: they can be derived as \textit{approximate} influence functions of the estimand or as the \textit{exact} influence functions of an estimand approximating the target of interest. In our context, the latter interpretation can be derived as follows. For shorthand notation, let $q(A, X) = 1 / \pi(A \mid X)$ and define
\begin{align*}
	q_k(x) = b(x)^T \Omega^{-1} \E\left\{ K_{ht}(A) b(X) q(A, X)\right\} \equiv b(x)^T \Omega^{-1} R_\pi \equiv b(x)^T \beta_{\pi}, \\
	\mu_k(x) = b(x)^T \Omega^{-1} \E\left\{ K_{ht}(A) b(X) \mu(A, X)\right\} \equiv b(x)^T \Omega^{-1} R_\mu \equiv b(x)^T \beta_{\mu}.
	\end{align*}
If $q(A, X)$ and $\mu(A, X)$ were replaced by $q(t, X)$ and $\mu(t, X)$ in $R_\pi$ and $R_\mu$, respectively, then $q_k(x)$ and $\mu_k(x)$ would be the $g$-weighted orthogonal projections of $q(t, X)$ and $\mu(t, X)$ onto the space spanned by $b(x)$. Consider the following approximating target:
\begin{align*}
\theta_{hk}(t) = \int q_k(x) \mu_k(x) g(x) dx = \beta_{\pi}^T \Omega \beta_{\mu}.
\end{align*}
Under smoothness conditions on $q(a, x)$ and $\mu(a, x)$, the approximation error  $|\theta_{hk}(t) - \theta(t)|$ shrinks as $h \to 0$ and $k \to \infty$. In this light, one may shift the focus from estimating $\theta_0(t)$ to estimating $\theta_{hk}(t)$, relying on the fact that, for fixed $k$ and $h$, $\theta_{hk}(t)$ possesses influence functions of any order. In particular, letting ${\rm IF}(\cdot)$ return the (unique) influence function of its argument in a nonparametric model, calculations outlined, for example, in \cite{kennedy2022semiparametric} yield
	\begin{align*}
		& {\rm IF} (\Omega) = K_{ht}(A) b(X) b(X)^T - \Omega, \\
		&   {\rm IF} (\Omega^{-1}) = - \Omega^{-1} K_{ht}(A) b(X) b(X)^T \Omega^{-1} + \Omega^{-1},  \\
		& {\rm IF}(\beta_\mu) = {\rm IF} (\Omega^{-1}) R_{\mu}  + \Omega^{-1} {\rm IF}(R_\mu) \\
		& \hphantom{{\rm IF}(\beta_\mu)} = -\Omega^{-1} K_{ht}(A) b(X) \mu_k(X) + \Omega^{-1} K_{ht}(A) b(X) Y, \\
		& {\rm IF}(\beta_\pi) =   {\rm IF} (\Omega^{-1}) R_{\pi}  + \Omega^{-1} {\rm IF}(R_\pi) \\
		 & \hphantom{{\rm IF}(\beta_\pi)} = -\Omega^{-1} K_{ht}(A) b(X) q_k(X) + \Omega^{-1} b(X).
	\end{align*}
This means that the first-order influence function of the approximating functional is
\begin{align*}
	{\rm IF}(\theta_{hk}(t)) = {\rm IF}(\beta_\pi^T\Omega \beta_\mu) = K_{ht}(A) q_k(X) \{Y - \mu_k(X)\} + \mu_k(X)  - \theta_{hk}(t).
\end{align*}
The second-order influence function can be computed as the first-order influence function of ${\rm IF}(\theta_{hk}(t))$ \citep{robins2009quadratic} (Lemma 1). We have
\begin{align*}
	& {\rm IF}(\mu_k(x_1)) = b(x_1)^T \Omega^{-1} K_{ht}(A_2) b(X_2) \{Y_2 - \mu_k(X_2)\}  \\
	& \hphantom{{\rm IF}(\mu_k(x_1))} = \Pi(x_1, X_2) K_{ht}(A_2) \{Y_2 - \mu_k(X_2)\}, \\
	& {\rm IF}(q_k(x_1)) = - \Pi(x_1, X_2) \{K_{ht}(A_2) q_k(X_2) - 1\}.
\end{align*}
Therefore, the second-order influence function of $\theta_{hk}(t)$ is
\begin{align*}
	& {\rm IF}({\rm IF}(\theta_{hk}(t))) = - K_{ht}(A_1) \{Y_1 - \mu_k(X_1)\} \Pi(X_1, X_2) \{K_{ht}(A_2) q_k(X_2) - 1\} \\
	& \hphantom{{\rm IF}({\rm IF}(\theta_{hj}(t))) =} - \{K_{ht}(A_1) q_k(X_1) -1\}\Pi(X_1, X_2) K_{ht}(A_2) \{Y_2 - \mu_k(X_2)\}.
\end{align*}
The term $\Pb^2\{\widehat\varphi_2(Z_1, Z_2)\}$ should act as the second term in the functional Taylor expansion of $p \mapsto \chi(p)$ at $p = \widehat{p}$, so we can see that ${\rm IF}({\rm IF}(\theta_{hk}(t))) / 2 $ agrees with our definition of $\varphi_2(Z_1, Z_2)$. The higher-order terms can be similarly derived by computing the first-order influence function of the previously derived term. 

Just like in the standard application of HOIF theory to pathwise differentiable parameters \citep{robins2009quadratic, robins2017higher}, the second-order correction subtracts off an estimate of the bias of the first-order estimator. Letting $v(x) = \mu(t, x) - \widehat\mu(t, x)$ and $q(x) = 1/\widehat\pi(t \mid x) - 1/\pi(t \mid x)$, the first-order bias equals 
\begin{align}\label{eq:first_order_bias}
	& \int \widehat\varphi_0(z) d\Pb(z) - \theta(t) \nonumber \\
	& = \int\int\frac{\mu(a, x) - \widehat\mu(t, x) }{\widehat\pi(t \mid x)} \cdot K_{ht}(a) p(a, x) da dx - \int\frac{\widehat\mu(t, x) - \mu(t, x)}{\pi(t \mid x)} \cdot p(t, x) dx \nonumber \\
	& = O(h^{\alpha \land \beta}) + \int v(x) q(x) g(x) dx.
\end{align}
Expressing the second term in \eqref{eq:first_order_bias} using an orthonormal basis $b$ in $L_2(g)$ and then trucating it at order $k$, one has that
\begin{align*}
\int \widehat\varphi_0(z) d\Pb(z) - \theta(t) =	O(h^{\alpha \land \beta}) + \beta_v^T \Omega \beta_q + \text{truncation error},
\end{align*}
where $\beta_f = \Omega^{-1} \int b(x) f(x) g(x) dx$ for $f = \{v, q\}$. The truncation error appears in the statement of Theorem 3.1 as the first term in the bound on the bias. The second-order correction aims to estimate $\beta_v^T \Omega \beta_q$ (unbiasedly up to order $h^{\alpha \land \beta}$).
\color{black}
The discussion above highlights how terms higher than the second would not be needed if $\Omega$, and thus $\Pi_{1, 2} = b^T(x_1) \Omega^{-1} b(x_2)$, was known. In fact, one may view higher-than-two order corrections as aiming to reduce the error in estimating the projection kernel. For example, consider $\widehat\varphi_2$ and $\widehat\varphi_3$, with $r_i(x_i)$ defined in the proof of Theorem 3.1. A direct calculation reveals:
	\begin{align*}
		& \left| -\E\{\widehat{f}_1(Z_1) \widehat\Pi_{1,2} \widehat{f}_2(Z_2) \mid D^n\} + \E\{\widehat{f}_1(Z_1) \Pi_{1,2} \widehat{f}_2(Z_2) \mid D^n\} \right. \\
		& \left. \quad + \E\{\widehat{f}_1(Z_1) \widehat\Pi_{1, 2} K_{ht}(A_2)\widehat\Pi_{2, 3} \widehat{f}_2(Z_3)  \mid D^n\} - \E\{\widehat{f}_1(Z_1) \widehat\Pi_{1, 3} \widehat{f}_2(Z_3) \mid D^n\}  \right| \\
		& \lesssim  h^{\alpha \land \beta} \\
		& \hphantom{\lesssim } + \left| - \int r_1(x_1)b^T(x_1)\widehat\Omega^{-1}(\Omega - \widehat\Omega)\Omega^{-1} b(x_2) r_2(x_2) g(x_1)dx_1 g(x_2) dx_2 \right. \\
		& \hphantom{\lesssim} \quad\quad \left. + \int  r_1(x_1) b^T(x_1)\widehat\Omega^{-1}(\Omega - \widehat\Omega)\widehat\Omega^{-1} b(x_2) r_2(x_2) g(x_1) dx_1 g(x_2) dx_2 \right| \\
		& =h^{\alpha \land \beta} \\
		& \hphantom{=} + \left| \int r_1(x_1)b^T(x_1)\widehat\Omega^{-1}(\Omega - \widehat\Omega)(\widehat\Omega^{-1} - \Omega^{-1}) b(x_2) r_2(x_2) g(x_1)dx_1 g(x_2) dx_2 \right|. 
	\end{align*}
Therefore, adding a third-order term reduces the error in estimating $\Pi_{1,2}$ from $(\Omega - \widehat\Omega)$ to $(\Omega - \widehat\Omega)^2$. Notice that in the statement of Theorem 3.1 below this error term appears as $\|f\|_\infty^{m-1} \equiv \sup_x |\widehat{p}(t, x) - p(t, x)|^{m-1}$ in the bias. 

\color{black}
In the proof of Theorem 3.1, we assume that all nuisance functions, namely $\pi(a \mid x), \mu(a, x)$ and $g(x)$ are estimated using a separate independent sample $D^n$. We derive the statement under the assumption that $\Pi(x_i, x_j)$ is estimated by $b(x_i)^T \widehat\Omega^{-1} b(x_j)$, where $\widehat\Omega = \int b(x)b(x)^T \widehat{g}(x) dx$, for some suitable estimator $\widehat{g}$ of $g$. The weight $g(x) = \int K_{ht}(a)p(a, x) da$ can be estimated as $\int K_{ht}(a) \phat(a, x) da$, for example. We note that estimating the joint density $p(a, x)$ and then integrating it over a multidimensional domain can be computationally prohibitive even in moderate dimensions. We are not aware of any procedure that could avoid this computational burden completely and in general. Finally, notice that, relative to the first-order methods described in Section \ref{sec:dr_est}, the higher-order estimator $\widehat\theta_m(t)$, $m \geq 2$, requires choosing the additional tuning parameter $k$ governing the dimension of the basis vector and, potentially, the order $m$ as well. This complicates implementation because, to the best of our knowledge, choosing $k$ and $m$ in a data-driven way remains elusive; please see Section \ref{sec:tuning} for additional discussion on computation and hyperparameter tuning.
\color{black}
\begin{remark}
   The estimator $\widehat\theta_1(t) = \Pn\{\widehat{f}_0(Z)\}$, corresponding to $m = 1$, is precisely the estimator studied in \cite{colangelo2020double}. Thus, we may view the $m^{\text{th}}$-order estimator as a higher-order generalization of their approach. 
\end{remark}
\begin{remark}
    The $m^{\text{th}}$-order estimator that we study has the same form as the $m^{\text{th}}$-order estimator of the functional $\psi = \int \E(Y \mid A = 1, X = x) p(x) dx$ studied in \cite{robins2017higher} (Section 8) except that terms of the form $Af(Z)$ for some function $f$ of the observations are replaced by $K_{ht}(A) f(Z)$. In fact, the rate described in Theorem \ref{thm:mthorder} is similar to that for $\psi$ from Theorem 8.1 in \cite{robins2017higher} with $n$ replaced by $nh$. Finally, Section 9 in \cite{robins2017higher} presents an estimator that is a modified version of that presented in Section 8.1 where certain terms in the influence functions are ``cut out'' to decrease the variances without substantially increasing the bias. This results in a more complex estimator that exhibits a better, and in fact minimax optimal under certain conditions, bias-variance trade-off. We plan to apply this refinement to the dose-response settings in future work, with the idea of first calculating a candidate minimax lower bound. 
\end{remark}
\subsection{Upper bound on the (conditional) risk}
The following theorem bounds the risk of the estimator $\thetahat_m(t)$ conditional on the training sample $D^n$. 
\begin{theorem}\label{thm:mthorder}
	Suppose Assumptions 1-2 hold and the following assumptions also hold:
	\begin{enumerate}
		\item The functions $a \mapsto \mu(a, x)$ (as well as $a \mapsto \widehat\mu(a, x)$) and $a \mapsto \pi(a \mid x)$ (as well as $a \mapsto \widehat\pi(a \mid x)$) are H\"{o}lder smooth of orders $\alpha$ and $\beta$, respectively, for any $x \in \mathcal{X}$;
		\item The kernel $K$ of order $l$, for $l$ the largest integer strictly less than $\alpha \land \beta$, is uniformly bounded and supported in $[-1, 1]$. Further, it satisfies $\int K_{ht}(a) da = 1$ and $\int K_{ht}(a) p(a, x) da \in [\underline{M}, \overline{M}]$ for some $\underline{M}> 0$, $\overline{M} < \infty$ and all $x \in \mathcal{X}$. 
		\item The orthogonal projection kernel $\Pi$ and its estimator $\widehat\Pi$ satisfy 
  \begin{align*}
      \sup_{x \in \mathcal{X}} \Pi(x, x) \lesssim k \quad \text{and} \quad \sup_{x \in \mathcal{X}} \widehat\Pi(x, x) \lesssim k;
  \end{align*}
		\item Boundedness: $\int K_{ht}(a) p(a, x) da / \int K_{ht}(a) \phat(a, x) da \in [\underline{M}', \overline{M}']$ for some $\underline{M}' > 0$, $\overline{M}' < \infty$ and all $x \in \mathcal{X}$. 
	\end{enumerate}
	Then
	\begin{align*}
		& \left| \E\{\thetahat_m(t) - \theta_0(t) \mid D^n\} \right| \lesssim \Norm{(I - \Pi) (v)}_{2, g} \Norm{(I - \Pi)(q)}_{2, g}  + h^{\alpha \land \beta} \\
		& \hphantom{ \left| \E\{\thetahat_m(t) - \theta(t) \mid D^n\} \right| \lesssim } \quad + \Norm{q}_{2, g}\Norm{v}_{2, g}\Norm{f}^{m-1}_\infty \\
		& \var\{\thetahat_m(t) \mid D^n\} \lesssim \sum_{l = 1}^m \left(\sum_{j = l}^m c^j j^{2l} \epsilon_n^{2(j-l)} \right) \frac{k^{l-1}}{(nh)^l}.
	\end{align*}
	where $c$ is some constant, $v(x) = \mu(t, x) - \muhat(t, x)$, $q(x) = 1 / \pihat(t \mid x) - 1 / \pi(t \mid x)$, $f(x) = \phat(t, x) - p(t, x)$, $\Norm{f}^k_{k, g} = \int |f|^k(x) g(x) dx$, and $\epsilon_n = h^{\alpha \land \beta} \ \lor \ \|v\|_{4,g} \ \lor \ \|q \|_{4, g} \ \lor \ \|f\|_\infty.$
\end{theorem}
The assumptions underlying Theorem \ref{thm:mthorder} are similar to those made in Propositions \ref{thm:erm} and \ref{thm:dr}. The main difference is that the higher order estimator $\thetahat_m(t)$ is specifically designed for nonparametric models where $a \mapsto \mu(a, x)$ and $a \mapsto \pi(a \mid x)$ possess some smoothness, which we encode in condition 1. The second condition ensures that the kernel $K$ accurately tracks the least smooth function between $a \mapsto \mu(a, x)$ and $a \mapsto \pi(a \mid x)$. A better estimator or a tighter bound would track just the smoothness of $\theta_0(t)$ or, at least, the smoothness of $a \mapsto \mu(a, x)$, as that should match the smoothness of $\theta_0(t)$ in most applications. We leave this for future work. In particular, we conjecture it might be possible to derive a tighter bound that would have, in place of the term $h^{\alpha \land \beta}$, terms of order $h^{\alpha \land (\beta + 1)}$ plus terms of order $h^{\alpha \land \beta}(\Norm{v} + \Norm{q} + o(h))$. This refined bound would also not track the smoothness of the dose-response and thus we preferred the simpler and more interpretable bound in terms of $h^{\alpha \land \beta}$. 

Because higher order kernels can take negative values on sets of non-zero Lebesgue measure (see, e.g. Proposition 1.3 in \cite{tsybakov2008introduction}), we require $g(x) = \int K_{ht}(a) p(a, x) da$ to be bounded away from zero since this is the weight used in the projection onto the finite space of dimension $k$. Condition 3 requires the kernels $\Pi$ and $\widehat\Pi$ to be bounded on the diagonal. This would be satisfied, for instance, if the basis elements are bounded. Condition 4 is a mild regularity condition on the estimator $\phat(a, x)$.

The first term in the bias is an upper bound on the truncation bias. It arises from approximating the target functional $\theta_0(t) = \int \mu(t, x) dP(x)$ by truncating the basis expansion of $x \mapsto \mu(t, x)$ at level $k$. This term would be present even if the matrix $\Omega$ in the orthogonal projection kernel was known and the $U$-statistic empirical measure replaced by the true expectation. It can also be viewed as an upper bound on the error incurred in approximating the first-order estimator's bias by expressing the residuals $v(x)$ and $q(x)$ with a $k$-dimensional basis. The second term of order $h^{\alpha \land \beta}$ upper bounds the error in approximating $a \mapsto \pi(a \mid x)$ and $a \mapsto \mu(a, x)$ with a local smoothing kernel with bandwidth $h$. Ideally, the upper bound would be of order $h^{\alpha}$, since a term of this order would also appear in the error expansion of the oracle estimator with access to the true nuisance functions. However, our analysis effectively relies on the smoothness of $a \mapsto \pi(a \mid x) \mu(a, x)$ which matches the lower smoothness order between that of $a \mapsto \pi(a\mid x)$ and $a \mapsto \mu(a, x)$. The third term in the bias arises from the approximation of the orthogonal projection kernel using the matrix $\widehat\Omega$, which relies on an estimator of the joint density of the covariates and treatment at level $A = t$. Finally, each term in the bound on the variance comes from a direct calculation of the second moments of each term in the higher-order $U$-statistic, relying on the assumption that the projection kernel is upper bounded by a constant multiple of $k$ on the diagonal (Condition 4); our proof follows closely that of \cite{robins2017higher} (Lemma 10.2 in their corrected version of the manuscript) and keeps track of the fact that the $U$-statistic terms are localized averages within a neighborhood of size $h$ around $A = t$.

Next, we discuss a few implications of Theorem \ref{thm:mthorder}, under the assumptions that 1) $\alpha \leq \beta$, i.e. $a \mapsto \pi(a \mid x)$ is smoother than $a \mapsto \mu(a, x)$, and 2) the dose-response is also $\alpha$-smooth.

First, we consider the case where $x \mapsto \widehat\mu(t, x)$ and $x \mapsto \mu(t, x)$ are H\"{o}lder-$\gamma_1$ and $x \mapsto \widehat\pi(t \mid x)$ and $x \mapsto \pi(t \mid x)$ are H\"{o}lder-$\gamma_2$. Given an appropriate basis, we suppose the approximation error satisfies 
	\begin{align*}
		\Norm{(I - \Pi) (v)}_{g} \Norm{(I - \Pi)(q)}_{g} \lesssim k^{-(\gamma_1 + \gamma_2) / d }, \text{ where } d = \dim(X).
	\end{align*}
	Each term in the bound on the variance is of order $k^{j-1} / (nh)^j$. If $k$ is of order $nh$ the variance is of order $(nh)^{-1}$. With this choice of $k$, provided that $\|f\|_\infty = o_P(1)$, the term $\Norm{q}_{g}\Norm{v}_{g}\Norm{f}^{m-1}_\infty$ can be made arbitrarily small by choosing $m$ large enough (but constant) and thus it is negligible relative to the other terms.\footnote{How large $m$ would need to be to make this term negligible depends on $\|q\|_g$, $\|v\|_g$ and $\|f\|_\infty$. For instance, for $m=2$, this term would also be negligible if the joint density of $(X, A)$ is estimable at sufficiently fast rates. } If this term is negligible, the bound from Theorem \ref{thm:mthorder} simplifies to $O(k^{-2(\gamma_1 + \gamma_2) / d } + h^{2\alpha} + (nh)^{-1})$. Thus, if the average nuisance functions' smoothness satisfies $(\gamma_1 + \gamma_2) / 2 \geq d / 4$ and $h \asymp n^{-1 / (2\alpha + 1)}$, one obtains the oracle rate $n^{-2\alpha / (2\alpha + 1)}$. In this light, $\widehat\theta_m(t)$ behaves like the oracle estimator that uses the true nuisance regression functions if $\alpha \leq \beta$ and $(\gamma_1 + \gamma_2) / 2 \geq d / 4$, provided that $m$ is chosen large enough. If $s = \gamma_1 = \gamma_2$, this means that $\widehat\theta_m(t)$ is oracle efficient for $s \geq d / 4$. Comparing this result to the oracle efficiency bound $s \geq (d\alpha) / (2\alpha + 1)$ (derived in \eqref{eq:dr_oracle_bar}), it can be seen that higher-order corrections, at least in the case where $1/2 \leq \alpha \leq \beta$, effectively lower the bar for oracle efficiency.
Theorem \ref{thm:mthorder} also sheds light on the rate of convergence in the ``low smoothness'' regime when $(\gamma_1 + \gamma_2)/2 < d/4$. Consider the second-order estimator $\widehat\theta_2(t)$. Suppose $k$ and $h$ are chosen optimally to minimize the upper bound on the mean-square error:
\begin{align*}
	k \asymp (nh)^{2d / ( d + 2\gamma_1 + 2\gamma_2)} \quad \text{ and } \quad h \asymp n^{-2(\gamma_1 + \gamma_2) / [\alpha\{2(\gamma_1 + \gamma_2) + d\} + 2(\gamma_1 + \gamma_2)]}.
\end{align*}
Then, the MSE of $\thetahat_2(t)$ is of order $n^{-2r_2}$ for
\begin{align*}
	r_2 = \left\{1 + \frac{d}{2(\gamma_1 + \gamma_2)} + \frac{1}{\alpha}\right\}^{-1} \land \Norm{v}_g\Norm{q}_g\Norm{f}_\infty.
\end{align*}
If the first term in $r_2$ dominates the rate, then the rate obtained by the quadratic estimator $\thetahat_2(t)$ is a combination of the oracle rate $1 / (2 + 1 / \alpha)$ and the minimax rate for estimating the dose-response when $A$ is categorical in the non-root-$n$ regime, namely $n^{-2r_f}$, for $r_f = [1 + d /\{2(\gamma_1 + \gamma_2)\}]^{-1}$, which is recovered as $\alpha \to \infty$.  Negligibility of the third-order term $\Norm{v}_g\Norm{q}_g\Norm{f}_\infty$ can be achieved if $x \mapsto p(t, x)$ is estimated sufficiently fast.\footnote{An assumption of sufficient smoothness on the covariates' density to ensure third-order terms are negligible has been imposed, for example, in \cite{kennedy2024minimax} to show minimax optimality of their second-order conditional treatment effect estimator (condition 1, Theorem 2). It is also considered by \cite{liu2021adaptive} when studying data adaptive tuning parameter selection for estimation of pathwise differentiable functionals.} Notice that, in the regime $(\gamma_1 + \gamma_2)/2 < d / 4$, the estimators from Section \ref{sec:dr_est} are typically not oracle efficient as their rate is driven by $\rhat$. Without further corrections, $\rhat$ is bounded by the product of the MSEs for estimating $w$ and $\mu$; this is of bigger order than the term $\Norm{v}\Norm{q}\Norm{f}_\infty$, which is of the same order as $\Norm{v}_g\Norm{q}_g\Norm{f}_\infty$ because $\int K_{ht}(a) \pi(a \mid x) da$ is assumed uniformly bounded. 

Figure \ref{fig::rates} compares the rate achieved by the second-order estimator, assuming the third-order remainder term is negligible, to the rate of the estimators from Section \ref{sec:dr_est}. We plot the rates as a function of $s = \gamma_1 = \gamma_2$.  That is, $s$ refers to the smoothness of $x \mapsto \mu(a, x)$ and $x \mapsto \pi(a \mid x)$. For illustration, we set $\alpha = \beta = 2$ and $\dim(X) = 20$, where $\alpha$ is the smoothness of $a \mapsto \mu(a, x)$ and $a \mapsto \pi(a \mid x)$. In this setting, the optimal rate for estimating the anisotropic functions $\mu(a, x)$ and $\pi(a \mid x)$ is $n^{-2s / (2s + s / \alpha + d)}$. This is also the rate generally inherited by the plug-in estimator (black line) $\Pn\{\widehat\mu(t, X)\}$, without further corrections. The oracle rate is $n^{-2\alpha / (2\alpha + 1)}$. The estimators from Section \ref{sec:dr_est} (red line) achieve a rate of order $n^{-2\alpha/(2\alpha + 1)} \lor n^{-4s/(2s + s/\alpha + d)}$. The blue line refers to the rate $n^{-2 / \{1 + d/(4s) + 1/\alpha\}} \lor n^{-2\alpha/(2\alpha + 1)}$, which is obtained by the second-order estimator under the assumption that the covariates density is estimated well enough so that the term $\|v\|_g\|q\|_g\|f\|_\infty$ is negligible. Finally, for reference, we also plot the minimax lower bound for estimating the ATE, which is of order $n^{-2/\{1 + d/(4s)\}} \lor n^{-1}$ \citep{robins2009semiparametric}.
\begin{figure}[!h]
	\centering
	\includegraphics[scale = 0.3]{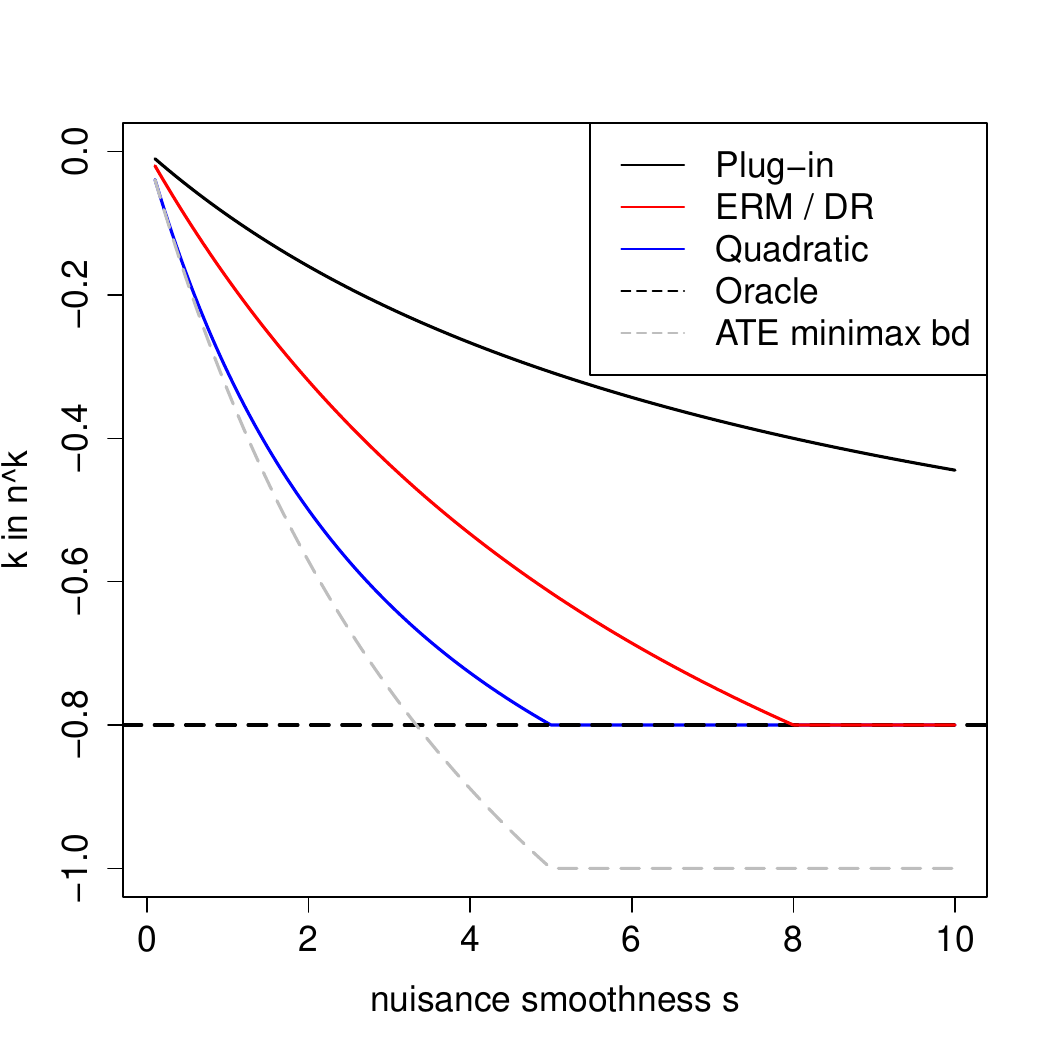}
	\caption{Illustration of the convergence rates in MSE for the estimator considered in this article, as a function of the smoothness $s = \gamma_1 = \gamma_2$. We take the smoothness of the dose-response to be $\alpha = 2$ and $\dim(X) = 20$.}
	\label{fig::rates}
\end{figure}
\subsection{Computational considerations and hyperparameter tuning}\label{sec:tuning}
In the previous section, we have derived the oracle efficiency bar from Theorem \ref{thm:mthorder} under the assumption that $m$, $\|q\|_g$, $\|v\|_g$ and $\|f\|_\infty$ are such that the term $\|q\|_g\|v\|_g\|f\|_\infty^{m-1}$ is negligible. In practice, if $\|q\|_g$, $\|v\|_g$ and $\|f\|_\infty$ are converging to zero slowly, one would need to compute $\widehat\theta_m(t)$ for large values of $m$. This would incur in substantial computational costs since the estimator would be a sum of $U$-statistics of order up to $m$. Remark 4.1 in \cite{liu2024assumption} also highlights the trade-off between bias reduction with higher-order corrections and computational costs. We refer to \cite{chen2025computing} for recent developments on efficient computation of higher-order $U$ statistics.

In our experience, an even more pressing issue to be addressed for widespread adoption of higher-order methods is that of hyperparameter selection. Choosing the projection dimension $k$ and the estimator of $\Omega^{-1}$ appears to be particularly important. To the best of our knowledge, this remains challenging. Even for pathwise differentiable estimands, work on data-adaptive implementations of higher-order methods remains scarce, with the notable exception of \cite{liu2021adaptive} studying hyperparameter tuning for second-order estimators via Lepski's method. Data adaptive estimation of order greater than two would be even more involved.
	
An important difference between pathwise differentiable parameters and the dose-response function is that the risk of a given estimator of the latter, e.g., $R(\theta) = \E[\{ \theta(A) - \theta_0(A)\}^2]$, is a pathwise differentiable parameter that can be estimated with parametric rate of convergence even though the dose-response function itself cannot generally be estimated at the parametric rate \citep{diaz2013targeted}. However, the error incurred in estimating the risk is comparable to that of a first-order estimator of $\theta_0(t)$ and so this strategy is not directly viable if the goal is to choose the hyperparameters entering the higher-order corrections.

Finally, choosing an appropriate estimator of $\Omega^{-1}$ that is computationally feasible and numerically stable is also nontrivial. For $k$ sufficiently small, an alternative to $\widehat\Omega$ is the empirical Gram matrix $\widehat\Omega_{emp} = \Pn\{b(X)b(X)^T K_{ht}(A)\}$. For our analysis to go through, $\widehat\Omega_{emp}$ would still need to be computed on the independent sample $D^n$, although we have found that estimating it \textit{in-sample} leads to numerical stability and good performance; see \cite{liu2023new} for a discussion on estimating $\Omega$ in-sample. The number of basis elements $k$ needs to be small enough so that $\widehat\Omega$ is invertible; because of the presence of $K_{ht}(A)$, $k$ would not be allowed to be of order greater than or equal to $nh$. In the context of pathwise differentiable functionals, \cite{mukherjee2017semiparametric} provide an in-depth discussion on using estimators akin to $\widehat\Omega_{emp}$ in place of those based on integration with estimated weight $\widehat{g}$. While a complete analysis of the error incurred by $\widehat\theta_m(t)$ when $\widehat\Omega$ is replaced by $\widehat\Omega_{emp}$ (whether in-sample or out-of-sample) is out of the scope of this paper, we provide the following heuristics. Consider $\widehat\theta_m(t)$ with $\widehat\Omega$ replaced by $\widehat\Omega_{emp}$ (estimated on $D^n$). Following the proof of Theorem 3.1, the conditional bias of $\widehat\theta_2(t)$ is of order:
	\begin{align*}
		\|(I - \Pi)(v)\|_g \|(I - \Pi)\|_g + h^{\alpha \land \beta} + \int \{\widehat\varphi_2(z_1, z_2) - \varphi_2(z_1, z_2)\}d\Pb(z_1, z_2),
	\end{align*}
where, for $r_1$ and $r_2$ defined in the proof:
\begin{align*}
 & \int \{\widehat\varphi_2(z_1, z_2)d\Pb(z_1, z_2) - \varphi_2(z_1, z_2)\}d\Pb(z_1, z_2) \\
 & = \int r_1(x_1) b^T(x_1)g(x_1)dx_1 \left(\widehat\Omega^{-1}_{emp} - \Omega^{-1}\right) \int b(x_2)r_2(x_2)g(x_2)dx_2.
\end{align*}
By the reasoning in Section 6.1 in \cite{mukherjee2017semiparametric} and our Theorem 3.1, if $\widehat\Omega_{emp}$ is invertible, the term above can be upper bounded in absolute value by
\begin{align*}
	\|r_1\|_g\|r_2\|_g\|\widehat\Omega_{emp} - \Omega\|_{\text{op}} \lesssim h^{\alpha \land \beta} + \|q\|_{2, g}\|v\|_{2, g}\|\widehat\Omega_{emp} - \Omega\|_{\text{op}}
\end{align*}
where $\|\cdot \|_{\rm op}$ is the matrix operator norm and $\|\widehat\Omega_{emp} - \Omega\|_{\text{op}} = O_P(\sqrt{k\log k / (nh)})$ by Rudelson's inequality, see, e.g., Lemma 17 in \cite{mukherjee2017semiparametric}. For $m \geq 2$, based on the analysis in \cite{mukherjee2017semiparametric}, we expect the bound on the conditional bias to be the same as the one above except that $\|\widehat\Omega_{emp} - \Omega\|_{\text{op}}$ is replaced by $\|\widehat\Omega_{emp} - \Omega\|_{\text{op}}^{m-1}$. Taking $k \asymp (nh) / (\log^2 (nh))$ and $m \asymp \log (nh)$ effectively makes $\|q\|_{2, g} \|v\|_{2, g} \|\widehat\Omega_{emp} - \Omega\|^{m-1}_{\text{op}}$ shrink at a rate faster than  $(nh)^{-c}$ for any constant $c$. A formal analysis would also need to carefully analyze the variance of the estimator based on $\widehat\Omega_{emp}$, particularly when $m$ grows with the sample size (albeit slowly). Finally notice that a larger $k$ should make $\|(I - \Pi)(v)\|_{2, g}\|(I - \Pi)(q)\|_{2, g}$ smaller, the extent of which depends on the approximation properties of the basis elements relative to the function class where $v$ and $q$ reside. Thus, choosing $k \asymp (nh) / (\log^2 (nh))$ will have consequences in terms of the rate achieved by this approximation term, which may or may not be smaller than the oracle rate. An important avenue for future work is thus to study data-driven estimation of $\Omega$ (possibly with regularization) in the context of higher-order estimators.
\section{Simulation experiment}\label{sec:sims}
We conduct a simulation experiment to evaluate the performance of the first-, second- and third-order estimators in finite samples. Additional results can be found in the supplement. We generate data according to the following process, which is loosely based on the simulation study presented in \cite{colangelo2020double}:
\begin{align*}
& X \sim N_{20}(0, \Sigma), \\
& A \mid X \sim \text{TruncNorm}(a_\text{min} = -1.5, a_\text{max} = 1.5, \text{ mean} = \mu_x(X), \text{ sd} = 2), \\
 & \text{and } Y \mid A, X \sim N(\mu_{ax}(A, X), 1).
\end{align*}
The covariance matrix $\Sigma$ has $(i, i)$-entry equal to 1 and $(i, j)$-entry equal to 0.5 if $|i-j|=1$ and equal to 0 otherwise. We set $\mu_x(x) = x^T\xi$ where the $j$-entry of $\xi$ is $j^{-2}$; notice that $\mu_x(x)$ refers to the mean of the underlying normal distribution before truncation, so that the mean of $A$ given $X$ is not necessarily linear. Finally,
\begin{align*}
	 \mu_{ax}(A, X) = & 1.2A + A^2 + 3A[s(X; \delta_1) - \E\{s(X; \delta_1)\}] \\
	 & + [s(X; \delta_2) - \E\{s(X; \delta_2)\}].
\end{align*}
The function $s(x; \delta)$ is defined as follows:
\begin{align*}
	s(x; \delta) = \sum_{j=1}^d 1.5 x_j + 2.5\cdot \delta^T \texttt{bs}(x_j)
\end{align*}
with \texttt{bs} denoting a $B$-spline basis with 10 inner knots equally spaced in $[-2, 2]$ and boundary knots at $-10$ and $10$. The parameters $\delta_1$ and $\delta_2$ are coefficients sampled uniformly in $[-2, 2]$. The dose-response function is thus $\theta(t) = 1.2 \cdot t + t^2$. Importantly, we have also generated data using $s(x; \delta)$ with 20 inner knots spaced in $[-2, 2]$ and found that the relative performance of the first- vs second- order methods reported below was largerly the same.

To construct the first-order estimators, we proceed as follows. We consider three different estimators of $\mu_{ax}(a, x)$: 1) additive B-spline regression with interaction term $A\cdot X$ and main term $A + A^2$ (i.e., $\sum_j \text{bs}(X_j) + A\cdot \sum_j\text{bs}(X_j) + A + A^2$), 2) B-spline regression without the interaction $A\cdot X$ term (simply an additive basis plus the term $A + A^2$), and 3) Random Forest estimated on a separate, independent sample. To estimate $\mu_x$, we regress $A$ on $X$ by B-spline regression or by Random Forest on an independent sample. Notice that this effectively means that $\mu_x$ is misspecified since the conditional mean of $A$ is not necessarily equal to $\mu_x(X)$ due to the truncation in $[-1.5, 1.5]$. All B-spline regressions are fitted using the function \texttt{bs} (\texttt{splines} package), with degrees of freedom in each term chosen, by leave-one-out-cross-validation, from the set $\{1, 2, 3, 6, 8, 12, 16, 32, 64\}$ (all terms are forced to have the same degrees of freedom). Random Forest was fitted with default parameters using the package \texttt{ranger} in \texttt{R}. Finally, the estimator of $\pi(a \mid x)$ is the density of a truncated normal with parameters $a_{\min}=-1.5$, $a_{\max}=1.5$, $\texttt{sd}=2$, but with the mean parameter replaced by $\widehat\mu_x(x)$. The marginal density of $A$, $p(a)$, is estimated by kernel density estimation using the function \texttt{density} with default parameters from the \texttt{R} package \texttt{stats}. Except for the Random Forest models, all estimators are computed on the same sample, i.e., no sample splitting is performed.

As an example of first-order methods, we estimate the dose-response function by regressing the pseudo-outcome $\widehat\varphi(Z)$ onto $A$ via local linear smoothing (Section \ref{sec:dr_est}); we use the function \texttt{lprobust} with default parameters from the package \texttt{nprobust} \citep{calonico2019nprobust}. We use the selected bandwidth from this step to also construct $\widehat\theta_1(t) = \Pn \widehat{f}_0$ from Section \ref{sec:higher_order_est}. 

To construct the second- and third-order estimators, we consider an additive $B$-spline basis where we expand each covariate using the first $k$ basis elements excluding the intercept. This yields a total of $20 \cdot k  + 1$ columns in the design matrix; we rely on the function \texttt{bs} to place the knots at the relevant quantiles. We vary $k$ in $\{4, 8, 16, 32\}$; since $d = 20$, this corresponds to $81$, $161$, $321$ and $641$ basis elements, respectively. We use the same bandwidth as the one used to construct the first-order estimators. The matrix $\Omega$ is estimated in-sample using the empirical Gram matrix $\Pn (b(X)b(X)^T K_{ht}(A))$; in the supplementary material, we also report the results when $\Omega$ is estimated by $\int b(x)b(x)^T \widehat{p}(t, x) dx$ and $X$ is one-dimensional. 

Finally, we consider a sample size $n=5000$; when the estimated nuisances rely on Random Forest, we employ an additional training sample of size $n_{\text{tr}} = 5000$.  We consider 1000 Monte Carlo simulations. We evalute the dose-response function on 21 points, equally spaced between the 0.1- and 0.9- quantiles of the distribution of $A$ (approximately $-1.17$ and $1.17$). 

Figure \ref{fig:mse} reports the root-mean-square error for the estimators considered. We have included the RMSE of an oracle local-linear estimator that has access to the true nuisance functions and regresses $\varphi(Z)$ (Section \ref{sec:dr_est}) onto $A$ using the \texttt{lprobust} function. As expected, the performance of all the estimators is similar when the model for $\mu_{ax}$ is correctly specified. When the model is misspecified or estimated with Random Forest (without any hyperparameter tuning), the second- and third-order estimators have lower RMSE in regions of the support of $A$ where the first-order estimator exhibit larger errors. Around $t = 0$, i.e., when the interaction term $A \cdot X$ is less relevant, first-order estimators perform well and the second-order corrections considered here are not able to further reduce the estimation error. We note that for $K = 4$, the higher-order corrections include fewer basis terms than the underlying truth ($s(X; \delta)$ depends on a B-spline basis with 10 inner knots); however, they are still able to reduce the error from first-order estimation. In addition, the projection kernel aims to approximate $1/\widehat\pi(t \mid x) - 1/\pi(t \mid x)$ and $\mu(t, x) - \widehat\mu(t, x)$. When the nuisances are estimated by Random Forest, it is not clear that a projection kernel based on B-splines is entirely appropriate. In this light, these results are encouraging insofar that they suggest a certain degree of robustness with respect to the projection kernel specification. In the appendix, we also consider a modified data generating process that is less favorable to higher-order corrections.

We conclude with a word of caution. The simulation study considered here, while showing encouraging results regarding higher-order corrections, presents also some limitations. First, our current setup considers first-order methods depending on nuisance estimators with a varying degree of misspecification (e.g., the model for the conditional density of the treatment given the covariates is misspecified). In practice, such misspecification might be detected by the analyst who would be able to proceed with better modeling choices, thus potentially making higher-order corrections less compelling. In future work, it would be critical to design simulation studies where no nuisance misspecification is detectable from the data during the first-order estimation step. In addition, the higher-order corrections considered here are for a relatively small number of basis elements compared to the sample size; we consider at most 681 terms relative to an effective sample size, $(nh)$, which is comparable to $n$ because the bandwidth $h$ is large as $\theta_0(t)$ is rather smooth. We have noticed that choosing $k$ too large relative to $(nh)$ leads to numerical instability. Finally, we refer the reader to \cite{li2005robust, mukherjee2017semiparametric, liu2024assumption} for comprehensive simulation studies illustrating the superior performance of estimators based on higher order influence functions in the context of pathwise differentiable parameters. Their simulation setups require substantial computational resources so that we leave extending them to our settings as future work.
\begin{figure}[!h]
	\centering
	\includegraphics[scale=0.30]{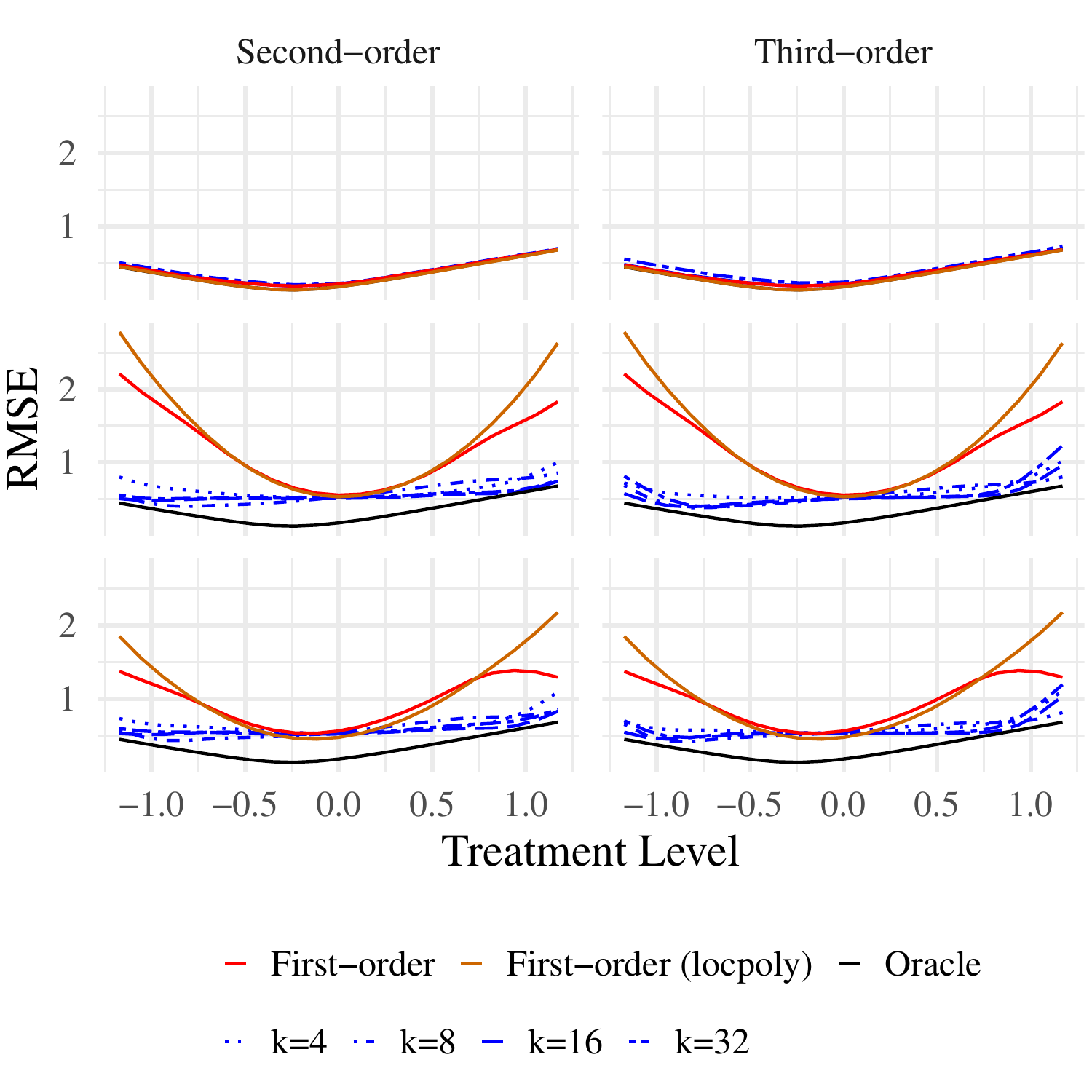}
	\caption{RMSE as a function of the evaluation points of the dose--response function in three settings. 
		Top row: the nuisance function estimators are correctly specified; middle row: the nuisance functions are estimated while omitting the interaction term; bottom row: the nuisance functions are estimated using Random Forest. 
		\textit{First-order (locpoly)} denotes the first-order estimator obtained by regressing $\widehat{\varphi}(Z)$ on $A$ using local linear smoothing, whereas \textit{First-order} denotes $\widehat{\theta}_1(t)=\Pn \widehat{f}_0$ from Section~\ref{sec:higher_order_est}. $k$ refers to the basis dimension in the higher-order corrections. \label{fig:mse}}
\end{figure}
\section{Data analysis}
In this section, we apply our methods to the dataset analyzed in \cite{colangelo2020double}\footnote{Available at \url{https://github.com/KColangelo/Double-ML-Continuous-Treatment}.} to study the effect of the Job Corps program. The goal is to understand the causal dose-response function linking the total number of hours of academic and vocational training that the participants attend in the first year of the program (treatment $A$) to the expected proportion of weeks employed in the second year (outcome $Y$). The dataset consists of 4,024 young individuals (age: 16-24 years) who completed at least one week of training (40 hours). We treat all the covariates $X$ from the dataset as baseline covariates and identify the dose-response function under conditional randomization, i.e., $Y^a \ind A \mid X$ for $Y^a$ the potential outcome if treatment is set to $A = a$ \citep{rubin1974estimating}. We refer to \cite{colangelo2020double} for further details.

We use Random Forest (from the \texttt{ranger} package in \texttt{R} with hyperparameter tuning by five-fold cross-validation, \citep{wright2019package}) as the main model for all nuisance functions and employ five-fold cross-fitting. For a fold $k$ consisting of a set of $n_k$ units $J_k$, we estimate $\int \mu(a, x) d\Pb(x)$ as $n_k^{-1} \sum_{i \in J_k} \widehat\mu^{(-k)}(a, X_i)$, where $\widehat\mu^{(-k)}$ denotes a Random Forest estimate of $\mu$ fitted excluding observations in fold $J_k$. That is, we simplify the implementation and avoid the additional splitting outlined in Algorithm \ref{alg:first_order} or the leave-one-out averaging discussed in \cite{semenova2017debiased}. We follow the conditional density model from \cite{kennedy2017nonparametric} whereby one assumes that $A = \E(A \mid X) + \sigma(X)\epsilon$ for some mean-zero, unit-variance $\epsilon$. On the training sample, we estimate $\E(A \mid X)$ and $\sigma(X)$ by Random Forest and compute $\widehat\epsilon = \{A - \widehat\E(A \mid X)\}/\widehat\sigma(X)$ on the test sample. We then estimate the density of $\widehat\epsilon$ by kernel density estimation (\texttt{kde} from the \texttt{ks} package in \texttt{R} with default tuning parameters, \citep{duong2024package}). The resulting estimate of $\pi(A \mid X)$ on the test sample is simply the density of $\widehat\epsilon$ rescaled by $\widehat\sigma(X)$. The marginal density of $A$ is estimated on the full sample by kernel density estimation. 

Next, we regress the cross-fitted pseudo outcome $\widehat\varphi(Z)$ onto $A$ by local linear regression using the function \texttt{lprobust} with default parameters from the \texttt{nprobust} package and Epanechnikov kernel. The package implements the bias-corrected and robust inferential procedures proposed in \cite{calonico2018effect}; their bias-correction is meant to subtract-off the smoothing bias arising from using local polynomials, which is a different bias than that arising from having to estimate the nuisances and that we propose reducing through higher-order methods. In this light, for simplicity, here we report only their standard estimates and not the bias-corrected ones. However, we note that, in subsequent work to ours, \cite{takatsu2024debiased} extended the works of \cite{calonico2018effect} to the dose-response setting.

In addition to the first-order method based on local linear smoothing, we consider a second- and a third-order estimator. To construct the higher-order estimators, we expand the continuous covariates using an additive B-spline basis (\texttt{bs} function in \texttt{R}) with degrees of freedom equal to 6 for all these covariates except for mean-gross-weekly-earnings for which we use 10 degrees of freedom. When the covariates appear to have a point mass at zero, we exclude this value when computing the quantiles to place the knots. For the categorical variables, we consider two options: they are either added to the B-spline basis as main terms or with all their two-way interactions. The first case leads to 79 terms in the basis, while second case to 1025, which makes the in-sample Gram matrix $\widehat\Omega$ non-invertible. In this non-invertibility setting, our theoretical analysis would prescribe estimating $\Omega$ by numerical integration against the estimated joint density of the covariates and the treatment at the evaluation point $A=t$. However, to avoid the substantial computational costs, we simply resort to the pseudo-inverse. While this lacks theoretical justification, it yields reasonable estimates that do not differ too much from those obtained with an additive basis.

The results are reported in Figure \ref{fig:da}. Consistent with the finding from \cite{colangelo2020double}, the estimated dose-response from regressing $\widehat\varphi(Z)$ on $A$ via local linear smoothing (black line) is roughly an inverted U-shape. However, the standard errors are large enough that the constant dose-response function, corresponding to a null effect, cannot be ruled out. The higher-order estimates are roughly in agreement with the first-order ones on most of the evaluation region, which we take to be $[200, \ 1800]$. As the first-order estimator is based on local linear smoothing, to construct the higher-order corrections we rely on $K_{ht}(a)$ being the first-order estimator's equivalent kernel, which can be negative. We estimate $\Omega$ by its empirical, in-sample counterpart; for constructing the weight in $\widehat\Omega$, to avoid certain numerical instability, we use the Epanechnikov kernel in place of the local linear equivalent kernel (using the same bandwidth). Finally, in the supplementary material, we also report the bounds on the dose-response function arising from the proposed sensitivity analysis to the no-unmeasured-confounding assumption (Section 5).
\begin{figure}[!h]
	\centering
	\includegraphics[scale=0.3]{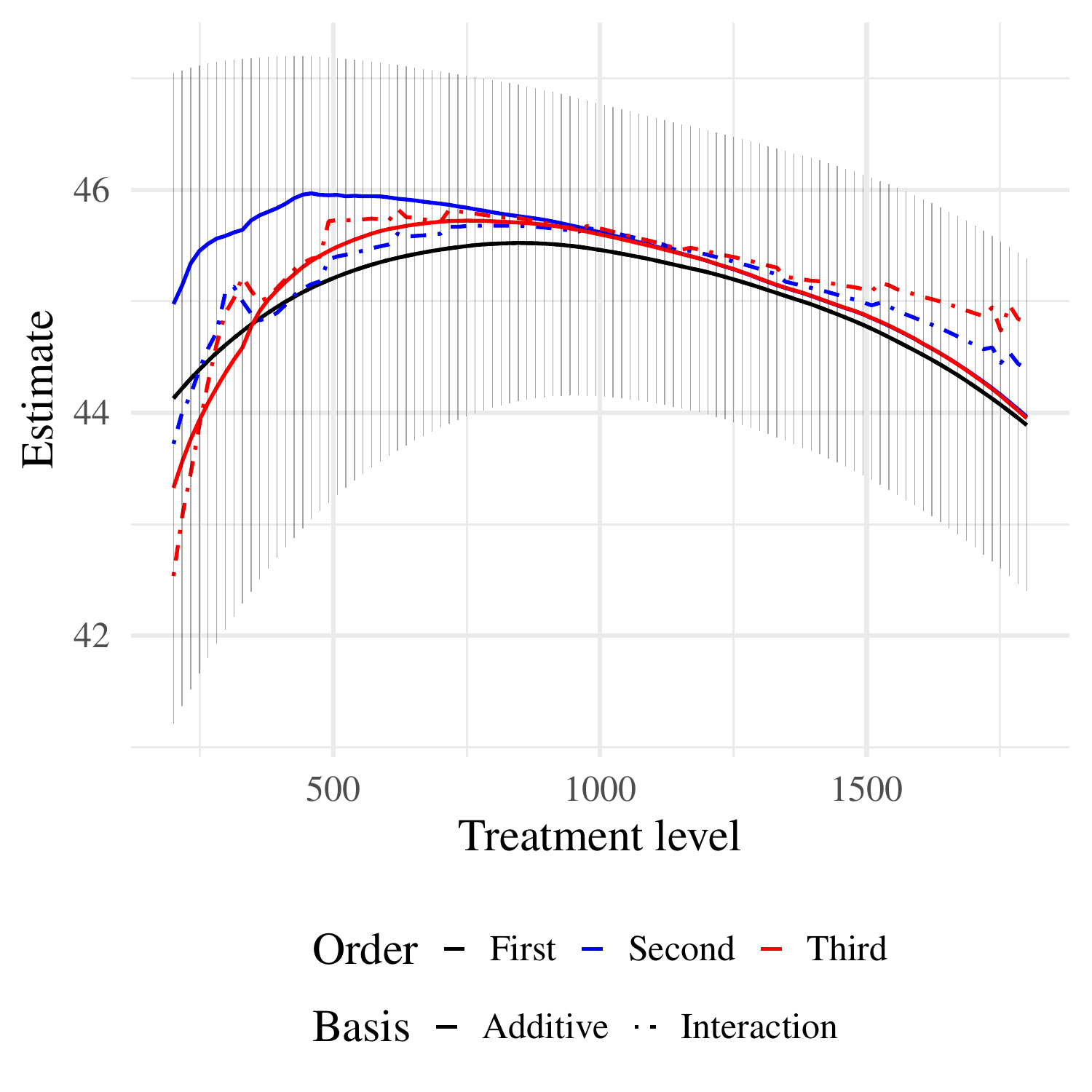}
	\caption{Estimates of $\E(Y^a)$ from our analysis of the Job Corps dataset from \cite{colangelo2020double}. The pointwise 95\% confidence bands are for the local linear estimator as calculated by the function \texttt{lprobust} with no smoothing bias correction. \label{fig:da}}
\end{figure}
\section{Conclusions and future directions}
In this work, we have explored the possibility of improving existing approaches to doubly-robust estimation of a dose–response curve. We specialize recent advances on regression with estimated outcomes to the dose–response setting, providing an analysis of a broad class of doubly-robust estimators. As our second and main contribution, we show that an estimator akin to the higher-order estimator of the average treatment effect described in \cite{robins2017higher} can improve upon existing estimators under certain smoothness conditions. In addition, a data analysis and simulation study illustrate and corroborate our theoretical findings in finite samples. In the supplementary material, we also describe a flexible method for bounding the causal dose–response function in the presence of unmeasured confounding.

Many open questions remain. First, and perhaps most importantly, a minimax lower bound for estimating the dose–response curve has not yet been established in the literature, to the best of our knowledge. Deriving such a bound would help clarify under what conditions, if any, the higher-order estimator proposed here can be improved.

Second, the higher-order estimator is currently unable to adapt to situations where the conditional density of the treatment given the covariates, viewed as a function of the treatment alone, is less smooth than the dose–response itself. It is unclear whether this reflects an intrinsic limitation of the estimator, a lack of tightness in the upper bound on the risk that we have derived, or a fundamental aspect of the minimax rate. A possible direction for future work is to investigate the construction of higher-order estimators based on regressions of suitable pseudo-outcomes onto $A$.

Finally, our results concern the convergence of the estimators in mean squared error. We leave the study of their inferential properties, as well as the development of data-driven procedures for selecting the hyperparameters, as important directions for future research.

\section{Acknowledgments}
The authors acknowledge the Texas Advanced Computing Center (TACC) at The University of Texas at Austin for providing computational resources that have contributed to the research results reported within this paper. URL: \url{http://www.tacc.utexas.edu}.
\appendix
\section{Proof of Proposition 2.1}
Suppose we observe two samples of $n$ iid observations from $\Pb$, say $D^n$ and $Z^n$. Denote $Z_1, \ldots, Z_n$ the observations in $Z^n$ that are iid copies of a generic random variable $Z$. Let $V$ denote a generic random variable such that $V \subset Z$. For example, in the dose-response settings, $Z = (Y, A, X)$ and $V = A$. Let $\theta_0(v) = \E\{f(Z) \mid V = v\}$ denote the true regression function that needs to be estimated. Recall that $\Norm{f}^2= \int f(z)^2 d\Pb(z) = \Pb\{f(Z)^2\}$ and $\theta^* \in \argmin_{\theta \in \Theta} \Norm{\theta - \theta_0}^2$, a fixed function. Let $\fhat(\cdot)$ denote an estimate of $f(\cdot)$ constructed using only observations in $D^n$. Let $\Pn$ denote the empirical average over sample $Z^n$. The estimator of $\theta_0(\cdot)$ considered is
\begin{align*}
	\thetahat \equiv \widehat\theta_{\text{erm}} = \argmin_{\theta \in \Theta} \frac{1}{n} \sum_{i = 1}^n \{\fhat(Z_i) - \theta(V_i)\}^2 \equiv \argmin_{\theta \in \Theta} \Pn\{\fhat(Z) - \theta(V)\}^2.
\end{align*}
Finally, let $\rhat(v) = \E\{\fhat(Z) \mid V  = v, D^n\} - \theta_0(v)$. 

The statement of the theorem follows after proving 
\begin{align} \label{eq:statement_1}
	\E(\Norm{\thetahat - \theta_0}^2 \mid D^n) \lesssim \Norm{\rhat}^2 + \Norm{\theta^* - \theta_0}^2 + \delta_n^2
\end{align}

Our proof is a specialization of that of Theorem 3 in \cite{foster2019orthogonal}. A useful reference for the arguments made in their proof is Chapter 14 in \cite{wainwright2019high}. To prove \eqref{eq:statement_1}, we need two lemmas. 
\begin{lemma}\label{lemma:base}
	The following inequality holds:
	\begin{align*}\label{eq:base}
		\Norm{\thetahat - \theta_0}^2 & \leq 8\Norm{\rhat}^2 + 3\Norm{\theta^* - \theta_0}^2 \\
		& \hphantom{\leq} - 2(\Pn - \Pb)\left[\{\fhat(Z) - \thetahat(V)\}^2 - \{\fhat(Z) - \theta^*(V)\}^2\right].
	\end{align*}
\end{lemma}
\begin{proof}
	Notice that
	\begin{align*}
		\Pb\left[\{\fhat(Z) - \thetahat(V)\}^2 - \{\fhat(Z) - \theta^*(V)\}^2\right] & = -2\Pb[\rhat(V)\{\thetahat(V) - \theta^*(V)\}] \\
		& \hphantom{=} + \Norm{\thetahat - \theta_0}^2 - \Norm{\theta^* - \theta_0}^2.
	\end{align*}
	By the AM-GM inequality we have, for any $\kappa > 0$\footnote{For any $x, y$ and $\kappa > 0$, $$\left(\frac{x}{\sqrt{2\kappa}} - y\sqrt{\frac{\kappa}{2}}\right)^2 \geq 0 \implies \frac{x^2}{2\kappa} + y^2\frac{\kappa}{2} \geq xy$$}: 
	\begin{align*}
		2\rhat(V)\{\thetahat(V) - \theta^*(V)\} & =  2\rhat(V)\{\thetahat(V) - \theta_0(V)\} + 2\rhat(V)\{\theta_0(V) - \theta^*(V)\} \\
		& \leq \frac{2}{\kappa}\rhat(V)^2 + \kappa\{\thetahat(V) - \theta_0(V)\}^2 + \kappa \{\theta^*(V) - \theta_0(V)\}^2
	\end{align*}
	By monotonicity of integration, it follows that
	\begin{align*}
		-2\Pb[\rhat(V)\{\thetahat(V) - \theta^*(V)\}]& \geq -\frac{2\|\rhat\|^2}{\kappa} - \kappa\Norm{\thetahat - \theta_0}^2 - \kappa\Norm{\theta^* - \theta_0}^2.
	\end{align*}
	Rearranging and choosing $\kappa = 1/2$, we have
	\begin{align*}
		\Norm{\thetahat - \theta_0}^2 \leq 8 \Norm{\rhat}^2 + 3\Norm{\theta^* - \theta_0}^2 + 2\Pb\left[\{\fhat(Z) - \thetahat(V)\}^2 - \{\fhat(Z) - \theta^*(V)\}^2\right].
	\end{align*}
	Because $\Pn[\{\fhat(Z) - \thetahat(V)\}^2 -  \{\fhat(Z) - \theta^*(V)\}^2] \leq 0$ since $\theta^* \in \Theta$ and $\thetahat$ is a minimizer, we also have
	\begin{align*}
		\Norm{\thetahat - \theta_0}^2 \leq & \ 8 \Norm{\rhat}^2 + 3\Norm{\theta^* - \theta_0}^2 \\
		& - 2(\Pn - \Pb)\left[\{\fhat(Z) - \thetahat(V)\}^2 - \{\fhat(Z) - \theta^*(V)\}^2\right],
	\end{align*}
	as desired. 
\end{proof}
\begin{lemma}\label{lemma:empirical_process}
	For some constant $L$, let
	\begin{align*}
		& \mathcal{E} = \left\{ \exists \theta \in \Theta: \Norm{\theta - \theta^*} \geq \delta_n \cap \vphantom{\left|(\Pn - \Pb)\left[\{\fhat(Z) - \theta(U)\}^2 -  \{\fhat(Z) - \theta^*(V)\}^2\right]\right|} \right. \\
		& \left. \hphantom{\mathcal{E} =} \quad \left|(\Pn - \Pb)\left[\{\fhat(Z) - \theta(V)\}^2 -  \{\fhat(Z) - \theta^*(V)\}^2\right]\right| \geq L\delta_n\Norm{\theta - \theta^*}\right\}
	\end{align*}
	Under the conditions of Proposition 2.1, $\Pb\left(\mathcal{E}\mid D^n\right) \leq c_1\exp(-c_2n\delta_n^2)$ for some constants $c_1$ and $c_2$. 
\end{lemma}
\begin{proof}
	Consider the sets
	\begin{align*}
		\mathcal{S}_m = \left\{\theta \in \Theta: 2^{m-1} \delta_n \leq \Norm{\theta - \theta^*} \leq 2^{m} \delta_n\right\}
	\end{align*}
	Because $\sup_{\theta \in \Theta} \Norm{\theta}_\infty \leq S$, $\Norm{\theta - \theta^*} \leq 2S$ for any $\theta \in \Theta$, which implies that any $\theta$ such that $\Norm{\theta - \theta^*} \geq \delta_n$ must belong to some $\mathcal{S}_m$ for $m \in \{1, \ldots, M\}$, where $M \leq \log_2(2S / \delta_n)$. By a union bound,
	\begin{align*}
		& \Pb(\mathcal{E} \mid D^n) \\
		& \leq \sum_{m = 1}^M \Pb(\mathcal{E} \cap \mathcal{S}_m \mid D^n) \\
		& \leq \sum_{m=1}^M \Pb\left( \vphantom{ \left|(\Pn - \Pb)\left[\{\fhat(Z) - \theta(V)\}^2\right]\right|} \exists \theta \in \Theta: \Norm{\theta - \theta^*} \leq 2^m \delta_n \right. \\
		& \hphantom{\leq \sum_{m=1}^M \Pb\left(\right.} \left. \cap \left|(\Pn - \Pb)\left[\{\fhat(Z) - \theta(V)\}^2 -  \{\fhat(Z) - \theta^*(V)\}^2\right]\right| \geq 2^{m-1}L\delta^2_n \mid D^n \right) \\
		& \leq \sum_{m = 1}^M\Pb(Z_n(2^m\delta_n) \geq 2^{m-1}L\delta_n^2 \mid D^n)
	\end{align*}
	where we define
	\begin{align*}
		Z_n(r) = \sup_{\theta \in \Theta: \|\theta - \theta^*\| \leq r}  \left|(\Pn - \Pb)\left[\{\fhat(Z) - \theta(V)\}^2 -  \{\fhat(Z) - \theta^*(V)\}^2\right]\right|
	\end{align*}
	Under the conditions of the proposition, we have
	\begin{align*}
		\sup_{\theta \in \Theta} \sup_{z \in \mathcal{Z}} \left|\{\fhat(z) - \theta(u)\}^2 -  \{\fhat(z) - \theta^*(v)\}^2 \right| \leq 8S^2 
	\end{align*}
	and
	\begin{align*}
		\left[\{\fhat(z) - \theta(v)\}^2 -  \{\fhat(z) - \theta^*(v)\}^2 \right]^2 & \leq 16S^2\{\theta(v) - \theta^*(v)\}^2
	\end{align*}
	Thus, we have
	\begin{align*}
		\sigma^2(r) \equiv \sup_{\theta: \Norm{\theta - \theta^*} \leq r} \Pb\left(\left[\{\fhat(Z) - \theta(V)\}^2 -  \{\fhat(Z) - \theta^*(V)\}^2 \right]^2\right) \leq 16S^2r^2
	\end{align*}
	By Theorem 3.27 in \cite{wainwright2019high} and subsequent discussion, viewing $\fhat(\cdot)$ as fixed given $D^n$, we have
	\begin{align*}
		& \Pb\left(Z_n(r) \geq \E\{Z_n(r) \mid D^n\} + v \mid D^n \right) \\
		& \leq 2\exp\left(- \frac{nv^2}{8e[16S^2r^2 + 16S^2\E\{Z_n(r) \mid D^n\}] + 32S^2v}\right)
	\end{align*}
	Next, we bound $\E\{Z_n(r) \mid D^n\}$. By a symmetrization argument, for $\epsilon$ a vector of iid Rademacher random variables independent of $Z_n$ and $D^n$, it holds that
	\begin{align*} 
		& \E\{Z_n(r) \mid D^n\} \\
		& \leq 2\E_{Z, \epsilon}\left(\sup_{\theta \in \Theta: \|\theta - \theta^*\| \leq r} \left|\Pn\left(\epsilon\left[\{\fhat(Z) - \theta(V)\}^2 -  \{\fhat(Z) - \theta^*(V)\}^2 \right]\right)\right| \left.\vphantom{\sum} \right| D^n \right)
	\end{align*}
	The Ledoux-Talagrand contraction inequality (see also pages 147 and 474 in \cite{wainwright2019high}) yields that, for non-random $x_i \in \mathcal{X}$, a class $
	\mathcal{F}$ of real-valued functions and a $L$-Lipschitz function $\phi: \R \to \R$, the following holds
	\begin{align*}
		\E\left(\sup_{f \in \mathcal{F}} \left|\sum_{i = 1} ^ n \epsilon_i\{\phi(f(x_i)) - \phi(f^*(x_i))\}\right| \right) \leq 2L \E\left(\sup_{f \in \mathcal{F}} \left| \sum_{i = 1} ^ n \epsilon_i\{f(x_i) - f^*(x_i)\} \right| \right)
	\end{align*}
	where $f^*: \mathcal{X} \to \R$ is any function. 
	
	Under the boundedness conditions of our proposition, we have
	\begin{align*}
		\left| \{\fhat(z) - \theta_1(v)\}^2 - \{\fhat(z) - \theta_2(v)\}^2 \right| \leq 4S|\theta_1(v) - \theta_2(v)|
	\end{align*}
	for any $z \in \mathcal{Z}$. Thus, the square-loss in this case is $4S$-Lipschitz for any $z \in \mathcal{Z}$. By the contraction inequality above, we have
	\begin{align*}
		& \E_{\epsilon}\left(\sup_{\theta \in \Theta: \|\theta - \theta^*\| \leq r} \left|\Pn\left(\epsilon\left[\{\fhat(Z) - \theta(u)\}^2 -  \{\fhat(Z) - \theta^*(V)\}^2 \right]\right)\right| \right)\\
		& \leq 2\E_{\epsilon}\left(\sup_{\theta \in \Theta: \|\theta - \theta^*\| \leq r} \left|\Pn[\epsilon\{\theta(V) -  \theta^*(V)\}] \right| \right)
	\end{align*} 
	Therefore, we have
	\begin{align*}
		& \E_{Z, \epsilon}\left(\sup_{\theta \in \Theta: \|\theta - \theta^*\| \leq r} \left|\Pn\left(\epsilon\left[\{\fhat(Z) - \theta(V)\}^2 -  \{\fhat(Z) - \theta^*(V)\}^2 \right]\right)\right| \left.\vphantom{\sum} \right| D^n \right) \\
		& \quad \leq 8S\E\left(\sup_{\theta \in \Theta: \|\theta - \theta^*\| \leq r} \left|\Pn[\epsilon\{\theta(V) -  \theta^*(V)\}]\right| \left.\vphantom{\sum} \right| D^n \right) \\
		& \quad \equiv  8S \mathcal{R}_n(\Theta^*, r)
	\end{align*}
	Next, we have assumed $\Theta^*$ to be star-shaped; by Lemma 13.6 in \cite{wainwright2019high} the function $r \mapsto \mathcal{R}_n(\Theta^*, r) / r $ is non-increasing. Therefore, because $\delta_n$ solves $\mathcal{R}_n(\Theta^*, \delta) \leq \delta^2$, we also have:
	\begin{align*}
		\mathcal{R}_n(\Theta^*, r) \leq r\delta_n \quad \text{ for all } r \geq \delta_n.
	\end{align*}
	Therefore, we conclude that $\E\{Z_n(r) \mid D^n\} \leq 16Sr\delta_n$ for all $r \geq \delta_n$. 
	
	Putting everything together, we have derived that
	\begin{align*}
		\Pb\left(Z_n(r) \geq 16Sr\delta_n + v \mid D^n \right) \leq 2\exp\left(- \frac{nv^2}{8e(16S^2r^2 + 16^2S^3r^2) + 32S^2v}\right)
	\end{align*}
	Let $L = 34S$; specializing this bound to our setting with $r = 2^m\delta_n$ and $u = S2^m \delta_n^2$, we have
	\begin{align*}
		\Pb\left(Z_n(2^m\delta_n) \geq  L \cdot 2^{m-1} \delta_n^2 \mid D^n\right) \leq 2\exp\left(- \frac{n\delta_n^2}{8e(16 + 16^2S) + 32S} \right)
	\end{align*}
	since $2^{-m} \leq 1$ for any $m \geq 1$. Finally,
	\begin{align*}
		\Pb(\mathcal{E} \mid D^n) & \leq \sum_{m=1}^M \Pb\left(Z_n(2^m\delta_n) \geq L 2^{m-1}\delta_n^2 \mid D^n\right) \\
		& \leq 2\exp\left(- \frac{n\delta_n^2}{8e(16 + 16^2S) + 32S} + \ln M \right)
	\end{align*}
	
	Recall that $M \leq \log_2(2S / \delta_n) \leq \log_2(2S\sqrt{2n})$ because we have assumed $\delta_n \geq 1/\sqrt{2n}$. Therefore, if
	\begin{align*}
		\delta_n^2 \geq \frac{2\ln\{\log_2(2S\sqrt{2n})\}\{8e(16 + 16^2S) + 32S\}}{n}
	\end{align*}
	we can conclude 
	\begin{align*}
		\Pb(\mathcal{E} \mid D^n) \leq 2\exp\left(- \frac{n\delta_n^2}{16e(16 + 16^2S) + 64S}\right)
	\end{align*}
	as desired.
\end{proof}
\subsection{Proof of Equation \eqref{eq:statement_1}}
Notice that Lemma \ref{lemma:empirical_process} implies that, with probability at least $1 - c_1\exp(-c_2n\delta_n^2)$, either of the following two events occur:
\begin{enumerate}
	\item Event 1:
	\begin{align*}
		\Norm{\thetahat - \theta^*} \leq \delta_n & \implies \Norm{\thetahat - \theta_0} \leq \delta_n +  \Norm{\theta^* - \theta_0} \\
		& \implies \Norm{\thetahat - \theta_0}^2 \leq 2\delta_n^2 + 2\Norm{\theta^* - \theta_0}^2
	\end{align*} 
	\item Event 2:
	\begin{align*}
		&	\left|(\Pn - \Pb)\left[\{\fhat(Z) - \widehat\theta(V)\}^2 -  \{\fhat(Z) - \theta^*(V)\}^2\right]\right| \\
		& \leq L\delta_n\Norm{\thetahat - \theta^*} \\
		& \leq \frac{L^2\delta_n^2}{\kappa} + \frac{\kappa \Norm{\thetahat - \theta_0}^2}{2} + \frac{\kappa \Norm{\theta^* - \theta_0}^2}{2}
	\end{align*}
	for any $\kappa > 0$.
\end{enumerate} 
Because of the result from Lemma \ref{lemma:base}, Event 2 (with $\kappa = 1 / 2)$ implies
\begin{align*}
	\Norm{\thetahat - \theta_0}^2 \leq 16\Norm{\rhat}^2 + 7\Norm{\theta^* - \theta_0}^2 + 8L^2 \delta_n^2
\end{align*}
This means that there exists a constant $C$ such that
\begin{align*}
	\Pb\left(\Norm{\thetahat - \theta_0}^2 \leq C\left(\Norm{\rhat}^2 + \Norm{\theta^* - \theta_0}^2 + \delta_n^2\right) \mid D^n\right) \geq 1 - c_1\exp(-c_2n\delta_n^2)
\end{align*}
Let $t_0 = C\left(\Norm{\rhat}^2 + \Norm{\theta^* - \theta_0}^2 + \delta_n^2 \right)$. This implies that
\begin{align*}
	& \E\left(\Norm{\thetahat - \theta_0}^2 \mid D^n \right)\\
	& = \int_0^\infty \Pb\left(\Norm{\thetahat - \theta_0}^2 > t \mid D^n\right) dt \\
	& = \int_0^{t_0} \Pb\left(\Norm{\thetahat - \theta_0}^2 > t \mid D^n\right) dt + \int_{t_0}^\infty\Pb\left(\Norm{\thetahat - \theta_0}^2 > t \mid D^n\right) dt \\
	& = \int_0^{t_0} \Pb\left(\Norm{\thetahat - \theta_0}^2 > t \mid D^n\right) dt + \int_{0}^\infty\Pb\left(\Norm{\thetahat - \theta_0}^2 > t_0 + t \mid D^n\right) dt \\
	& \leq t_0 + \int_0^\infty c_1\exp(-c_3nt) dt \\
	& = t_0 + \frac{c_1}{c_3 n}
\end{align*}
as desired. The last inequality holds because $\Pb\left(\Norm{\thetahat - \theta_0}^2 > t \mid D^n\right) \leq 1$ and because, whenever $\delta_n$ satisfies $\mathcal{R}_n(\delta_n; \Theta^*) / \delta_n \leq \delta_n$, then so does $\delta_n' = \sqrt{\delta_n^2 + t / C} > \delta_n $. This means that we can write
\begin{align*}
	t_0 + t = C\left(\Norm{\rhat}^2 + \Norm{\theta^* - \theta_0}^2 + {\delta_n'}^{2} \right)
\end{align*}
Thus, $$\Pb\left(\Norm{\thetahat - \theta_0}^2 > t_0 + t \mid D^n\right) \leq c_1\exp\{-c_2 n (\delta^2_n + t / C)\} \leq c_1\exp(-c_3 nt)$$ as Lemma \ref{lemma:empirical_process} holds for any $\delta'_n$ that solves $\mathcal{R}_n(\delta; \Theta^*) / \delta \leq \delta$.
\section{Proof of Proposition 2.2}
The proof of this proposition is based on Proposition 1 and Theorem 1 from \cite{kennedy2020optimal}. Their Theorem 1 together with consistency of $\widehat{f}(z)$ yields that
\begin{align*}
	|\widehat\theta_{\text{ls}}(t) - \theta_0(t)| \leq & \ | \widetilde\theta_{\text{ls}}(t) - \theta_0(t)| + \left| \frac{1}{n}\sum_{i = 1}^n W_i(t; V^n)\E\{ \widehat{f}(Z_i) - f(Z_i) \mid D^n, V_i\} \right| \\
	& + \color{black}{o_\Pb\left(\sqrt{\E\left[\left\{\widetilde\theta_{\text{ls}}(t) - \theta_0(t) \right\}^2\right]} \right)}.
\end{align*}
Under the assumptions of Proposition 2.2 (localized weights):
\begin{align*}
	\left| \frac{1}{n} \sum_{i = 1}^n W_i(t; V^n) \rhat(V_i) \right| \leq \frac{1}{n} \sum_{i = 1}^n \left| W_i(t; V^n) \right| |\rhat(V_i)| \one\{V_i \in N(t)\} \lesssim \sup_{v \in N(t)} |\rhat(v)|,
\end{align*}
as desired. 
\section{Extension of Proposition 2.2 to the $L_p$ risk}
Here, we report an extension to Proposition 2.2 that covers the $L_p$ risk. In particular, let $\|\widehat\theta_{\text{ls}} - \theta_0\|_p^p = \int |\widehat\theta_{\text{ls}}(v) - \theta_0(v)|^p d\Pb(v)$. Define the smoothed conditional bias:
\begin{align*}
	\widehat{b}(v) = \frac{1}{n}\sum_{i=1}^n W_i(v; V^n) \E\{\widehat{f}(Z) - f(Z) \mid V_i, D^n\}
\end{align*}
and recall the oracle estimator $\widetilde\theta_{\text{ls}}(v) = n^{-1}\sum_{i=1}^n W_i(v; V^n)f(Z_i)$. 
\begin{proposition}
	Under the condition of Proposition 2.2, it holds that
	\begin{align*}
		\|\widehat\theta_{\text{ls}} - \theta_0\|_p^p \lesssim  \|\widetilde\theta - \theta_0\|_p^p +  \|\widehat{b}\|_p^p + o_\Pb\left(\E\|\widetilde\theta_{\text{ls}} - \theta_0\|_p^p\right).
	\end{align*}
\end{proposition}
\begin{proof}
	We start by bounding the $L_2$ risk and then we derive a bound for the $L_p$ case with $2 \leq p < \infty$. The arguments follow by Section B of \cite{rambachan2022robust}; we restate their proof for the $L_2$ risk here for completeness. 
Let $\widetilde\theta_{\text{ls}}(v) = n^{-1} \sum_{i=1}^n W_i(v; V^n) f(Z_i)$ denote the oracle estimator that has access to the true pseudo-outcome $f(z)$. Further, let $\|W_i(\cdot; V^n)\|^2 = \int W^2_i(v; V^n)d\Pb(v)$ and define
\begin{align*}
	\|g\|_{w^2} = \frac{\sum_{i=1}^n \|W_i(\cdot; V^n)\|^2 \int g^2(Z) d\Pb(Z \mid V_i)}{\sum_{j=1}^n \|W_j(\cdot; V^n)\|^2}.
\end{align*}
We have
\begin{align*}
	\|\widehat\theta_{\text{ls}} - \theta_0\|^2 \leq \|\widehat\theta_{\text{ls}} - \widetilde\theta_{\text{ls}} - \widehat{b}\|_2^2 + \| \widehat{b} \|_2^2 + \|\widetilde\theta_{\text{ls}} - \theta_0\|_2^2.
\end{align*}
Notice that, since $\E\{f(Z) \mid V_i\} = \theta(V_i)$, we have:
\begin{align*}
	&  \E\|\widetilde\theta_{\text{ls}} - \theta_0\|_2^2 \\
	& =  \E\left( \int \left[ \frac{1}{n} \sum_{i=1}^n W_i(v; V^n)\{f(Z_i) - \theta(V_i)\}\right]^2d\Pb(v) \right) \\
	& \hphantom{=} + \E\left( \int \left[\frac{1}{n} \sum_{i=1}^n W_i(v; V^n)\{\theta_0(V_i) - \theta_0(v)\}\right]^2 d\Pb(v) \right) \\
	& \geq \E\left[\frac{1}{n^2}\sum_{i=1}^n \|W_i(\cdot; V^n)\|^2 \var\{f(Z) \mid V_i\} \right] \\
	& = \E \left\{\|\sigma\|_{w^2} \frac{1}{n^2}\sum_{j=1}^n \|W_j(\cdot; V^n)\|^2\right\}.
\end{align*}
where $\sigma(v) = \var\{f(Z) \mid V=v\}$. Next, letting $r(v) = \E\{\widehat{f}(Z) - f(Z) \mid V=v, D^n\}$, notice that
\begin{align*}
	\|\widehat\theta_{\text{ls}} - \widetilde\theta_{\text{ls}} - \widehat{b}\|^2 = \int \left[\frac{1}{n}\sum_{i=1}^n W_i(v; V^n) \left\{\widehat{f}(Z_i) - f(Z_i) - r(V_i)\right\}\right]^2 d\Pb(v)
\end{align*}
and that 
\begin{align*}
	\E\left[\frac{1}{n}\sum_{i=1}^n W_i(v; V^n) \left\{\widehat{f}(Z_i) - f(Z_i) - r(V_i)\right\} \mid V^n, D^n\right] = 0.
\end{align*}
Therefore,
\begin{align*}
	& \E\left(\|\widehat\theta_{\text{ls}} - \widetilde\theta_{\text{ls}} - \widehat{b}\|^2 \mid V^n, D^n\right) \\
	&  = \int \E \left( \left[\frac{1}{n}\sum_{i=1}^n W_i(v; V^n) \left\{\widehat{f}(Z_i) - f(Z_i) - r(V_i)\right\}\right]^2 \mid V^n, D^n\right) d\Pb(v) \\
	& = \int\frac{1}{n^2}\sum_{i=1}^n W^2_i(v; V^n) \var\left\{\widehat{f}(Z_i) - f(Z_i) \mid V^n, D^n\right\}d\Pb(v) \\
	& \leq n^{-2}\sum_{i=1}^n \|W_i(\cdot; V^n)\|^2 \int\left\{\widehat{f}(z) - f(z)\right\}^2d \Pb(z \mid V_i) \\
	& =\|\widehat{f} - f\|_{w^2} \frac{1}{n^2}\sum_{j=1}^n \|W_j(\cdot; V^n)\|^2
\end{align*}
Thus, we can conclude that
\begin{align*}
	& \Pb\left( \frac{\|\sigma\|_{w^2}\|\widehat\theta_{\text{ls}} - \widetilde\theta_{\text{ls}} - \widehat{b}\|}{\|\widehat{f} - f\|_{w^2} \sqrt{\E(\|\widetilde\theta_{\text{ls}} - \theta\|^2)}} > t  \right) \\
	& = \E\left\{  \Pb\left( \frac{\|\sigma\|^2_{w^2}\|\widehat\theta_{\text{ls}} - \widetilde\theta_{\text{ls}} - \widehat{b}\|^2}{\|\widehat{f} - f\|^2_{w^2} \E\|\widetilde\theta_{\text{ls}} - \theta\|^2} > t^2  \right) \mid V^n, D^n\right\} \\
	& \leq\frac{\E\left\{ \|\sigma\|^2_{w^2}\sum_{j=1}^n n^{-2} \|W_j(\cdot; V^n)\|^2\right\}}{t^2 \E\|\widetilde\theta - \theta_0\|^2} \\
	& \leq \frac{1}{t^2},
\end{align*}
so that
\begin{align*}
	\|\widehat\theta_{\text{ls}} - \widetilde\theta_{\text{ls}} - \widehat{b}\|^2 = O_\Pb\left(\frac{\|\widehat{f} - f\|_{w^2}\E\|\widetilde\theta_{\text{ls}} - \theta\|^2}{\|\sigma\|_{w^2}}\right).
\end{align*}
Putting everything together, we have
\begin{align*}
	\|\widehat\theta_{\text{ls}} - \theta_0\|^2 \leq \|\widehat{b}\|_2^2 + \|\widetilde\theta_{\text{ls}} - \theta_0\|_2^2 + O_\Pb\left(\frac{\|\widehat{f} - f\|_{w^2}\E\|\widetilde\theta_{\text{ls}} - \theta\|^2}{\|\sigma\|_{w^2}}\right).
\end{align*}
For a general $p \in [2, \infty)$, letting $\|g\|_p^p = \int |g(v)|^p d\Pb(v)$, we have
\begin{align*}
	& \E\|\widetilde\theta_{\text{ls}} - \theta_0\|^p \\
	& \geq \int \E\left| \frac{1}{n}\sum_{i=1}^n W_i(v; V^n)\{f(Z_i) - \theta(V_i)\}\right|^p d\Pb(v) \\
	& \geq \int \E\left(\E\left[\left| \frac{1}{n}\sum_{i=1}^n W_i(v; V^n)\{f(Z_i) - \theta(V_i)\}\right|^p \mid V^n\right] \right) d\Pb(v) \\
	& \geq 2^{-p}\int \E\left\{\E\left(\left[ \frac{1}{n^2}\sum_{i=1}^n W^2_i(v; V^n)\{f(Z_i) - \theta(V_i)\}^2\right]^{p/2}\mid V^n\right) \right\} d\Pb(v) \\
	& \geq 2^{-p}  \int \E\left\{\left(\E\left[\frac{1}{n^2}\sum_{i=1}^nW^2_i(v; V^n)\{f(Z_i) - \theta(V_i)\}^2\mid V^n\right] \right)^{p/2} \right\} d\Pb(v) \\
	&= 2^{-p}  \int \E\left[\left\{\frac{1}{n^2}\sum_{i=1}^nW^2_i(v; V^n)\sigma^2(V_i)\right\}^{p/2} \right]d\Pb(v) \\
	& \geq 2^{-p} \inf_v\sigma^p(v) \int \E\left[\left\{\frac{1}{n^2}\sum_{i=1}^nW^2_i(v; V^n)\right\}^{p/2}\right] d\Pb(v) 
\end{align*}
where the second inequality follows from the inequalities on line 2.16 in \cite{gine2000exponential}, while the third one by Jensen's inequality and convexity of $x \mapsto |x|^{p/2}$ for $p \geq 2$. 

Next, letting $\|g\|_\infty = \sup_{v} |g(v)|$:
\begin{align*}
	& \E\left(\|\widehat\theta_{\text{ls}} - \widetilde\theta_{\text{ls}} - \widehat{b}\|_p^p \mid D^n\right) \\
	& = \int \E\left( \left|\frac{1}{n}\sum_{i=1}^n W_i(v; V^n) \left\{\widehat{f}(Z_i) - f(Z_i) - \widehat{r}(V_i)\right\}\right|^p \mid D^n\right) d\Pb(u) \\
	& \leq 2^p(p-1)^{p/2}\|\widehat{f} - f\|^p_\infty \int \E\left[\left\{ \frac{1}{n^2}\sum_{i=1}^n W^2_i(v; V^n)\right\}^{p/2}\right] d\Pb(v)
\end{align*}
where the inequality follows again by the inequalities on line 2.16 in \cite{gine2000exponential}. Thus, we conclude that
\begin{align*}
	\|\widehat\theta_{\text{ls}} - \widetilde\theta_{\text{ls}} - \widehat{b}\|^p_p = O_\Pb\left( \frac{\|\widehat{f} - f\|^p_\infty (\E\|\widetilde\theta_{\text{ls}} - \theta_0\|^p)}{\inf_u \sigma^p(v)} \right)
\end{align*}
Finally, we have
\begin{align*}
	\|\widehat\theta_{\text{ls}} - \theta_0\|_p^p & \leq \left( \|\widehat\theta_{\text{ls}} - \widetilde\theta_{\text{ls}} - \widehat{b}\|_p + \|\widetilde\theta_{\text{ls}} -  \theta_0\|_p + \| \widehat{b}\|_p \right)^p \\
	& \lesssim \|\widehat\theta_{\text{ls}} - \widetilde\theta_{\text{ls}} - \widehat{b}\|^p_p + \|\widetilde\theta_{\text{ls}} -  \theta_0\|^p_p + \| \widehat{b}\|^p_p,
\end{align*}
which yields the result.
\end{proof}
\section{Proof of Theorem 3.1}
The proof of this theorem essentially follows from that of Theorem 8.1 in \cite{robins2017higher}, with the main difference that our estimator has $K_{ht}(a)$ in place of $\one(A = t)$ so that our analysis will need to keep track of terms of order $O(h^{\alpha \land \beta})$.

To simplify the notation, we let $v(a, x) = \mu(a, x) - \muhat(a, x), h(a, x) = 1 / \pi(a \mid x)$, $q(a, x) = \hhat(a, x) - h(a, x)$, $v(x) = v(t, x)$, $q(x) = q(t, x)$, and $g(x) = \int K_{ht}(a) p(a, x)da$. Also we define $\Norm{f}^2_g = \int f^2(x) g(x) dx$.

Before computing bias and variance of our estimator, we state some useful facts about orthogonal projections. More general versions of these statements can be found in the supplementary material to \cite{robins2017higher}. First, recall the definition of the orthogonal projection and its kernel in our context. For $g(x) = \int K_{ht}(a) p(a, x) da$, $\ghat(x) = \int K_{ht}(a) \phat(a, x) da$ and a function $f$:
\begin{align*}
& \Pi (f)(x_i) = \int \Pi_{i, j} f(x_j) g(x_j) dx_j = b(x_i)^T \Omega^{-1} \int b(x_j) f(x_j) g(x_j) dx_j \\
& \widehat\Pi (f)(x_i) = \int \widehat\Pi_{i, j} f(x_j) \ghat(x_j) dx_j = b(x_i)^T \widehat\Omega^{-1} \int b(x_j) f(x_j) \ghat(x_j) dx_j \\
& \Omega = \int b(u)b(u)^T g(u) du \text{ and } \widehat\Omega = \int b(u)b(u)^T \ghat(u) du.
\end{align*}
\begin{itemize}
\item Fact 1. Orthogonal projections do not increase length: for any function $f$ and projection in $L_2(\mu)$,
\begin{align*}
	& \Norm{\Pi (f)}^2 = \int \Pi(f)(x) \Pi(f)(x) d\mu = \int \Pi(f)(x) f(x) d\mu \leq \Norm{\Pi (f)} \Norm{f} \\
	& \implies \Norm{\Pi (f)} \leq \Norm{f}
\end{align*}
by Cauchy-Schwarz, where $\Norm{f}^2 = \int f^2 d\mu$ and 
\begin{align*}
	\Pi(f)(x) = b(x)^T \left\{\int b(u)b(u)^T d\mu\right\}^{-1} \int b(u) f(u) d\mu.
\end{align*}
\item Fact 2. Let $w$ denote some positive and bounded weight function and $\Pi_w$ and $\Pi$ projections in $L_2(\mu)$ onto some fixed $k$-dimensional space $L$ spanned by $b_1(x), \ldots, b_k(x)$, with weights $w$ and $1$ respectively. Then, for any $l \in L$, we have $\Pi_w(l) = \Pi(l)$:
\begin{align*}
	l(x) & = b(x)^T \beta \\
	& = \Pi_w(l)(x) \\
	& = b(x)^T \left\{\int b(u) b(u)^T w(u) d\mu\right\}^{-1} \int b(u) b(u)^T \beta w(u) d\mu \\
	& = b(x)^T \left\{\int b(u) b(u)^T d\mu \right\}^{-1} \int b(u) b(u)^T \beta d\mu \\
	& = \Pi (l) (x) 
\end{align*}
where $\beta \in \R^k$ is some vector of coefficients. 
\item Fact 3. Useful identities:
\begin{align*}
	& \int \Pi(x_i, x_j) \Pi(x_j, x_k) g(x_j) dx_j = \\
	& = b(x_i)^T \Omega^{-1}\int  b(x_j) b(x_j)^T g(x_j) dx_j\Omega^{-1} b(x_k) = \Pi(x_i, x_k), \\
	& \int \widehat\Pi(x_i, x_j) \widehat\Pi(x_j, x_k) \ghat(x_j) dx_j = \\
	& = b(x_i)^T \widehat\Omega^{-1} \int b(x_j) b(x_j)^T  \ghat(x_j) dx_j \widehat\Omega^{-1} b(x_k)  = \widehat\Pi(x_i, x_k), \\
	& \int \Pi(x_i, x_j) \widehat\Pi(x_j, x_k) g(x_j) dx_j \\
	& = b(x_i)^T \Omega^{-1} \int b(x_j) b(x_j)^T  g(x_j) dx_j \widehat\Omega^{-1} b(x_k)  = \widehat\Pi(x_i, x_k).
\end{align*}
\end{itemize}
\subsection{Bias}
We will divide the proof of the bias bound in several steps:
\begin{enumerate}
\item Prove that, for some functions $r_1$ and $r_2$ (defined in the proof) and
$$T= - \int r_1(x_1) \Pi(x_1, x_2) r_2(x_2) g(x_1) g(x_2) dx_1 dx_2$$
the following holds
\begin{align*}
	\left| \int \widehat\varphi_1(z) d\Pb(z) - \theta_0(t) + T \right| \lesssim \Norm{(I - \Pi) (v)}_g\Norm{(I - \Pi)(q)}_g + h^{\alpha \land \beta}
\end{align*}
\item Prove that
\begin{align*}
	& \sum_{j = 2}^m \int \widehat\varphi_j(z_1, \ldots, z_j) d\Pb(z_1, \ldots, z_j) - T \\
	& = (-1)^{m-1}\int r_1(x_1) (\widehat\Pi_{1, 2} - \Pi_{1,2}) \cdots (\widehat\Pi_{m-1, m} - \Pi_{m-1, m}) r_2(x_m) \\
	& \hphantom{= (-1)^{m-1}\int r_1(x_1) (\widehat\Pi_{1, 2} - \Pi_{1,2}) } \times g(x_1)\cdots g(x_m)d x_1\cdots d x_m \\
	& \equiv T_2
\end{align*}
\item Prove that
\begin{align*}
	\left| T_2 \right| \lesssim \Norm{r_1}_g\Norm{r_2}_g\Norm{\shat - 1}^{m-1}_\infty \lesssim \left(\Norm{v}_g\Norm{q}_g + h^{\alpha \land \beta}\right)\Norm{\shat - 1}^{m-1}_\infty 
\end{align*}
where $\shat =  g / \ghat$.
\end{enumerate}
\subsubsection{Step 1}
Let us define
\begin{align*}
& \Delta_1(x) \equiv \int K_{ht}(a) \{v(a, x) - v(t, x)\}\pi(a \mid x) da \\
& \Delta_2(x) \equiv \int K_{ht}(a) \{\mu(a, x) - \mu(t, x)\}\pi(a \mid x) da \\
& \Delta_3(x) \equiv \int K_{ht}(a) \{\hhat(a, x) - \hhat(t, x)\} \pi(a \mid x) da \\
& \Delta_4(x) \equiv \int K_{ht}(a) \pi(a\mid x)da - \pi(t \mid x) \\
& \Delta_5(x) \equiv h(t, x) - \frac{1}{\int K_{ht}(a) \pi(a \mid x) da} = \frac{\Delta_4(x)}{\pi(t \mid x)\{\pi(t \mid x) + \Delta_4(x)\}}
\end{align*}
We have
\begin{align*}
& \int \widehat\varphi_1(z) d\Pb(z) - \theta_0(t) \\
& =  \int K_{ht}(a) \hhat(t, x) \{\mu(a, x) - \muhat(t, x)\} \pi(a \mid x) da p(x) dx - \int v(x) p(x) dx  \\
& =  \int K_{ht}(a) \hhat(t, x)\{\mu(a, x) - \muhat(t, x)\}\pi(a \mid x) da p(x) dx \\
& \hphantom{=} \quad - \int \frac{v(x)}{\int K_{ht}(a) \pi(a \mid x)da} g(x)dx \\
& =  \int v(x) q(x) g(x)dx + \int v(x) \Delta_5(x) g(x) dx + \int \hhat(t, x)\Delta_2(x) p(x) dx
\end{align*}
Let
\begin{align*}
& r_1(x) = v(x) + \frac{\Delta_1(x)}{\int K_{ht}(a) \pi(a \mid x) da}, \\
& r_2(x) = q(x) + \Delta_5(x) + \frac{\Delta_3(x)}{\int K_{ht}(a) \pi(a \mid x) da},
\end{align*}
and notice that
\begin{align*}
& \int K_{ht}(a)\{y - \muhat(a, x)\} d\Pb(z \mid x) = r_1(x) \int K_{ht}(a) \pi(a \mid x) da \\
& \int \{K_{ht}(a)\widehat{h}(a, x) - 1\} d\Pb(z \mid x) = r_2(x) \int K_{ht}(a) \pi(a \mid x) da
\end{align*}
Define
\begin{align*}
T & \equiv - \int r_1(x_1) \Pi(x_1, x_2) r_2(x_2) g(x_1) dx_1 g(x_2) dx_2 \\
& = - \int v(x) \Pi(q)(x) g(x) dx - \int \frac{\Delta_1(x)}{\int K_{ht}(a) \pi(a \mid x) da} \Pi(r_2)(x) g(x) dx\\
& \hphantom{=} \ - \int \Pi(v)(x) \left\{\Delta_5(x) + \frac{\Delta_3(x)}{\int K_{ht}(a) \pi(a \mid x) da}\right\}g(x) dx
\end{align*}
Therefore, 
\begin{align*}
\int \widehat\varphi_1(z) d\Pb(z) + T  - \theta_0(t) = \int v(x) (I - \Pi)(q) (x) g(x) dx + \Delta
\end{align*}
where $\Delta$ groups all the terms involving $\Delta_j$ together:
\begin{align*}
\Delta & = \int v(x) \Delta_5(x) g(x) dx + \int \hhat(t, x)\Delta_2(x) p(x) dx \\
& \hphantom{=} \ - \int \frac{\Delta_1(x)}{\int K_{ht}(a) \pi(a \mid x) da} \Pi(r_2)(x) g(x) dx \\
& \hphantom{=} \ - \int \Pi(v)(x) \left\{\Delta_5(x) + \frac{\Delta_3(x)}{\int K_{ht}(a) \pi(a \mid x) da}\right\}g(x) dx
\end{align*}
The term $\Delta$ is controlled under the smoothness assumptions of the theorem, while 
\begin{align*}
\left| \int v(x) (I - \Pi)(q) (x) g(x) dx \right| \leq \Norm{(I- \Pi)(v)}_g\Norm{(I- \Pi)(q)}_g
\end{align*}
by Cauchy-Schwarz. In particular, we have assumed that $a \mapsto \mu(a, x)$, $a \mapsto \muhat(a, x)$ are H\"{o}lder-$\alpha$ and $a \mapsto h(a, x)$, $a \mapsto \hhat(a, x)$ (or, equivalently, $a \mapsto \pi(a \mid x)$ and $a \mapsto \pihat(a \mid x)$) are H\"{o}lder-$\beta$.
Thus, we have
\begin{align*}
& v(a, x) = \\
& = \sum_{j = 0}^{\lfloor \alpha \rfloor} v^{(j)}(t, x) \frac{(a - t)^j}{j!} + \{v^{(\lfloor \alpha \rfloor)} (t + \tau_1(a - t), x) - v^{(\lfloor \alpha \rfloor)} (t, x) \}\frac{(a - t)^{\lfloor \alpha \rfloor}}{\lfloor \alpha \rfloor!}, \\
& \mu(a, x) \\
& = \sum_{j = 0}^{\lfloor \alpha \rfloor} \mu^{(j)}(t, x) \frac{(a - t)^j}{j!} + \{\mu^{(\lfloor \alpha \rfloor)} (t + \tau_2(a - t), x) - \mu^{(\lfloor \alpha \rfloor)} (t, x) \}\frac{(a - t)^{\lfloor \alpha \rfloor}}{\lfloor \alpha \rfloor!}, \\
& \pi(a \mid x) = \\
& = \sum_{j = 0}^{\lfloor \beta \rfloor} \pi^{(j)}(t \mid x) \frac{(a - t)^j}{j!} + \{\pi^{(\lfloor \beta \rfloor)} (t + \tau_3(a - t) \mid x) - \pi^{(\lfloor \beta \rfloor)} (t \mid x) \}\frac{(a - t)^{\lfloor \beta \rfloor}}{\lfloor \beta \rfloor!}, \\
& h(a, x) = \\
& = \sum_{j = 0}^{\lfloor \beta \rfloor} h^{(j)}(t,  x) \frac{(a - t)^j}{j!} + \{h^{(\lfloor \beta \rfloor)} (t + \tau_4(a - t), x) - h^{(\lfloor \beta \rfloor)} (t, x) \}\frac{(a - t)^{\lfloor \beta \rfloor}}{\lfloor \beta \rfloor!}
\end{align*}
for some $\tau_1, \tau_2, \tau_3, \tau_4 \in [0, 1]$. Then, for example, we have the following
\begin{align*}
& \Delta_1(x) \equiv \int K_{ht}(a) \{v(a, x) - v(t, x)\}\pi(a \mid x) da \\
& = \sum_{i = 1}^{\lfloor \alpha \rfloor}\sum_{j = 0}^{\lfloor \beta \rfloor} h^{i + j} \frac{v^{(i)}(t, x)}{i!} \frac{\pi^{(j)}(t \mid x)}{j!} \int_0^1  u^{i+j} K(u)du \\
& \hphantom{=} + \sum_{j = 1}^{\lfloor \alpha \rfloor} h^{ \lfloor \beta \rfloor + j} \frac{v^{(j)}(t, x)}{\lfloor \beta \rfloor! j!} \int_0^1 K(u) u^{\lfloor \beta \rfloor + j} \{\pi^{(\lfloor \beta \rfloor)} (t + \tau_3uh \mid x) - \pi^{(\lfloor \beta \rfloor)} (t \mid x) \}du \\
& \hphantom{=} + \sum_{j = 0}^{\lfloor \beta \rfloor} h^{\lfloor \alpha \rfloor + j}\frac{\pi^{(j)}(t \mid x)}{\lfloor \alpha \rfloor! j!} \int_0^1 K(u) u^{\lfloor \alpha \rfloor + j} \{v^{(\lfloor \alpha \rfloor)} (t + \tau_1uh, x) - v^{(\lfloor \alpha \rfloor)}(t, x)\} du \\
& \hphantom{=} + \frac{h^{\lfloor \alpha \rfloor + \lfloor \beta \rfloor }}{\lfloor \alpha\rfloor ! \lfloor \beta \rfloor!} \int_0^1 K(u) u^{\lfloor \alpha \rfloor + \lfloor \beta \rfloor} \{v^{(\lfloor \alpha \rfloor)} (t + \tau_1uh, x) - v ^{(\lfloor \alpha \rfloor)}(t, x)\} \\
& \hphantom{= \frac{h^{\lfloor \alpha \rfloor + \lfloor \beta \rfloor }}{\lfloor \alpha\rfloor ! \lfloor \beta \rfloor!} \int_0^1 K(u) u^{\lfloor \alpha \rfloor + \lfloor \beta \rfloor}} \quad \times \{ \pi^{(\lfloor \beta \rfloor )} (t + \tau_3uh \mid x) - \pi^{(\lfloor \beta \rfloor )} (t\mid x)\} du
\end{align*}
so that $\Norm{\Delta_1}_\infty \lesssim h^{ \beta + 1} + h^\alpha$. Similarly, $\Norm{\Delta_2}_\infty \lesssim h^{\beta + 1} + h^{\alpha}$ and
\begin{align*}
& \Delta_3(x) \equiv \int K_{ht}(a) \{\widehat{h}(a, x) - \widehat{h}(t, x)\}\pi(a \mid x) da \\
& = \sum_{i = 1}^{\lfloor \beta \rfloor }\sum_{j = 0}^{\lfloor \beta \rfloor} h^{i + j} \frac{\widehat{h}^{(i)}(t, x)}{i!} \frac{\pi^{(j)}(t \mid x)}{j!} \int_0^1  u^{i+j} K(u)du \\
& \hphantom{=} + \sum_{j = 1}^{\lfloor \beta \rfloor } h^{\lfloor \beta \rfloor + j} \frac{\widehat{h}^{(j)}(t, x)}{\lfloor \beta \rfloor ! j!} \int_0^1  K(u)  u^{\lfloor \beta \rfloor + j}\{\pi^{(\lfloor \beta \rfloor)} (t + \tau_3uh \mid x) - \pi^{(\lfloor \beta \rfloor)} (t \mid x)\} du \\
& \hphantom{=} + \sum_{j = 0}^{\lfloor \beta \rfloor} h^{\lfloor \beta \rfloor + j}\frac{\pi^{(j)}(t \mid x)}{\lfloor \beta \rfloor! j!} \int_0^1  K(u) \widehat{h}^{(\lfloor \beta \rfloor)} (t + \tau_4uh, x) - \widehat{h}^{(\lfloor \beta \rfloor)}(t, x)\} u^{\beta + j} du \\
& \hphantom{=} + \frac{h^{2\lfloor \beta \rfloor}}{\lfloor \beta \rfloor ! \lfloor \beta \rfloor !} \int_0^1 K(u) u^{2\lfloor \beta \rfloor }\{\widehat{h}^{(\lfloor \beta \rfloor } (t + \tau_4uh, x) - \widehat{h}^{(\lfloor \beta \rfloor } (t, x)\}\\
& \hphantom{= \frac{h^{2\lfloor \beta \rfloor}}{\lfloor \beta \rfloor ! \lfloor \beta \rfloor !} \int_0^1 K(u) u^{2\lfloor \beta \rfloor }} \quad \times \{ \pi^{(\lfloor \beta \rfloor)} (t + \tau_3uh \mid x) - \pi^{(\lfloor \beta \rfloor)} (t \mid x) \} du
\end{align*}
Therefore, $\Norm{\Delta_3}_\infty \lesssim h^{\beta}$ and, similarly, $\Norm{\Delta_4}_\infty \lesssim h^{\beta}$ and $\Norm{\Delta_5}_\infty \lesssim h^{\beta}$. In this light, it holds that $\Norm{\Delta_j}_\infty \lesssim h^{\alpha \land \beta}$, for $j = 1,2,3,4,5$. This concludes our proof that
\begin{align*}
\left| \int \widehat\varphi_1(z) d\Pb(z) + T  - \theta_0(t) \right| \lesssim \Norm{(I-\Pi)(v)}_g\Norm{(I-\Pi)(q)}_g + h^{\alpha \land \beta}.
\end{align*}
\subsubsection{Step 2}
We will show that
\begin{align*}
T_2 & = \sum_{j = 2}^m \int \widehat\varphi_j(z_1, \ldots, z_j) d\Pb(z_1, \ldots, z_j) - T \\
& = \sum_{j = 2}^m \int \widehat\varphi_j(z_1, \ldots, z_j) d\Pb(z_1, \ldots, z_j) \\
& \hphantom{=} \quad + \int r_1(x) \Pi(x_1, x_2)r_2(x_2) g(x_1) g(x_2) dx_1 dx_2 \\
& = (-1)^{m-1}\int r_1(x_1) (\widehat\Pi_{1, 2} - \Pi_{1,2}) \cdots (\widehat\Pi_{m-1, m} - \Pi_{m-1, m}) r_2(x_m) \\
& \hphantom{=(-1)^{m-1}\int r_1(x_1) (\widehat\Pi_{1, 2} - \Pi_{1,2})} \times g(x_1)\cdots g(x_m) dx_1\cdots d x_m
\end{align*}
The result is true for $m = 2$, so we proceed by induction. Relative to the $m^{\text{th}}$ term, the term $m+1$ receives the contribution from 
\begin{align*}
& \int \widehat\varphi_{m + 1}(z_1, \ldots, z_{m + 1}) d\Pb(z_1, \ldots, z_m) \\
& = (-1)^m\sum_{i = 0}^{m - 1}  {{m-1} \choose i} (-1)^{i}\int r_1(x_1) \widehat{\Pi}_{1, 2} \cdots \widehat{\Pi}_{m - i, m - i + 1} r_2(x_{m - i + 1}) \\
& \hphantom{= = (-1)^m\sum_{i = 0}^{m - 1}  {{m-1} \choose i} (-1)^{i}\int r_1(x_1) } \times g(x_1) \cdots g(x_{m-i + 1}) dx_1 \cdots d x_{m-i + 1} \\
& \equiv (-1)^m T_3
\end{align*}
Thus to prove the claim we need to show that
\begin{align*}
T_3 & = \int r_1(x_1) (\widehat{\Pi}_{1,2} - \Pi_{1,2}) \cdots (\widehat{\Pi}_{m-1, m} - \Pi_{m-1, m}) r_2(x_m) \\
& \hphantom{=\int r_1(x_1) (\widehat{\Pi}_{1,2} - \Pi_{1,2}) \cdots } \quad \times g(x_1) \cdots g(x_m) dx_1 \cdots dx_m \\
& \hphantom{=} + \int r_1(x_1) (\widehat{\Pi}_{1,2} - \Pi_{1,2}) \cdots (\widehat{\Pi}_{m, m + 1} - \Pi_{m, m+1}) r_2(x_{m+1}) \\
& \hphantom{=+ \int r_1(x_1) (\widehat{\Pi}_{1,2} - \Pi_{1,2}) \cdots } \quad\times g(x_1) \cdots g(x_{m + 1}) d x_1 \cdots d x_{m + 1} \\
& \equiv  T_4 + T_5
\end{align*}
Notice that $T_4$ can be written as a sum of terms of the form
\begin{align*}
B_l = (-1)^{m - 1 - l} \int r_1(x_1) B_{1, 2} \cdots B_{m - 1, m} r_2(x_{m}) g(x_1) \cdots g(x_m)  dx_1 \cdots dx_m
\end{align*}
where $B_{i, j}$ equals either $\widehat{\Pi}_{i, j}$ or $\Pi_{i, j}$ and $l$ denotes the number of terms in the product for which $B_{i,j} = \widehat{\Pi}_{i, j}$. Similarly, $T_5$ is a sum of terms of the form
\begin{align*}
C_l = (-1)^{m - l} \int r_1(x_1) B_{1, 2} \cdots B_{m, m + 1} r_2(x_{m + 1}) g(x_1) \cdots g(x_{m + 1}) dx_1 \cdots d x_{m + 1}
\end{align*}
Fact 3 is the reason why we only need to keep track of the number of $\widehat\Pi_{i, j}$ terms and not specifically which $B_{ij}$ equals $\Pi_{i,j}$ or $\widehat\Pi_{i, j}$. In fact, for $B_{ij} = \Pi_{i,j}$ or $\widehat\Pi_{i,j}$, we have 
\begin{align*}
& \int \Pi(x_{j - 1}, x_{j}) B(x_j, x_{j + 1}) g(x_j) dx_j\\
& = B_{j - 1, j + 1} = \int B(x_{j - 1}, x_{j}) \Pi(x_j, x_{j + 1}) g(x_j) dx_j.
\end{align*}
In this light, we can simplify as
\begin{align*}
& B_l = (-1)^{m - 1 - l} \int r_1(x_1) \widehat{\Pi}_{1, 2} \cdots \widehat{\Pi}_{l, l + 1} r_2(x_{l + 1}) dg(x_1)\cdots dg(x_{l + 1}) \text{ for } l \geq 1 \\
& B_0 = (-1)^{m-1}\int r_1(x_1) \Pi_{1, 2} r_2(x_2) dg(x_1) d(x_2) \\
& C_l = (-1)^{m - l} \int r_1(x_1) \widehat{\Pi}_{1, 2} \cdots \widehat{\Pi}_{l, l + 1} r_2(x_{l + 1}) dg(x_1)\cdots dg(x_{l + 1})\text{ for } l \geq 1 \\
& C_0 = (-1)^{m}\int r_1(x_1) \Pi_{1, 2} r_2(x_2) dg(x_1) d(x_2)  = -B_0
\end{align*}
For $l \in \{1, \ldots, m - 1\}$, we have $C_l = - B_l$. Thus, we have reached
\begin{align*}
& T_4 = - C_0 - \sum_{l = 1}^{m - 1} {m - 1 \choose l} C_l \quad \text{ and } \quad T_5 = C_0 + \sum_{l = 1}^{m - 1} {m \choose l} C_l + C_m
\end{align*}
and this implies
\begin{align*}
T_4 + T_5 & = \sum_{l = 1}^{m - 1}  \left\{{m \choose l} - {m - 1\choose l} \right\} C_l  + C_m  = \sum_{l = 1}^{m - 1} {m - 1 \choose l - 1} C_l + C_m = T_3
\end{align*}
as desired. We have thus shown that
\begin{align*}
T_2 = & (-1)^{m-1}\int r_1(x_1) (\widehat\Pi_{1, 2} - \Pi_{1,2}) \cdots (\widehat\Pi_{m-1, m} - \Pi_{m-1, m}) r_2(x_m) \\
& \hphantom{(-1)^{m-1}\int r_1(x_1) (\widehat\Pi_{1, 2} - \Pi_{1,2}) \cdots } \quad\times g(x_1)\cdots g(x_m) d x_1\cdots dx_m
\end{align*}
\subsubsection{Step 3}
We need to show that $|T_2| \ \leq \Norm{r_1}_g \Norm{r_2}_g \Norm{\widehat{s} - 1}_\infty^{m-1}$. This statement is essentially a direct consequence of Lemma 13.7 in the Supplementary material to \cite{robins2017higher}. For the sake of completeness, we give a proof here that is less general (and more verbose) than that in \cite{robins2017higher}, although it uses the same arguments. 

Define $\shat = g(x) / \ghat(x) $ and let $M_{\widehat{s}}$ denoting multiplication by $\widehat{s}$. We have
\begin{align*}
\int (\widehat{\Pi}_{m-1, m} - \Pi_{m-1, m}) r_2(x_m) g(x_m) dx_m = \left(\widehat\Pi M_{\widehat{s}} - \Pi\right) (r_2) (x_{m-1})
\end{align*}
Continuing with this calculation, we get
\begin{align*}
T_2 = (-1)^{m - 1} \int r_1(x) \left(\widehat\Pi M_{\widehat{s}} - \Pi\right)^{m-1} (r_2) (x) g(x) dx 
\end{align*}
Let $\Norm{f}^2_{2, \ghat} = \int f^2(u) \ghat(u) du$ and bound $|T_2|$ as
\begin{align*}
|T_2| \ \lesssim \ \|r_1\|_{g} \left\| \left(\widehat\Pi M_{\widehat{s}} - \Pi\right)^{m-1} (r_2) \right\|_{2, \widehat{g}}
\end{align*}
Define
\begin{align*}
l(x) = \left(\widehat\Pi M_{\widehat{s}} - \Pi\right)^{m-2} (r_2) (x) \equiv b(x)^T \beta
\end{align*}
We can write $l(x)$ as a linear combination of the truncated basis because both $\widehat\Pi$ and $\Pi$ project a function onto the same finite dimensional subspace. Notice that we can view $\Pi$ as a weighted projection in $L_2(\ghat)$ with weight $\shat$, i.e.
\begin{align*}
\Pi(f)(x) & = b(x)^T \Omega^{-1} \int b(u) f(u) g(u) du \\
& = b(x)^T \left\{\int b(u) b(u)
\shat(u) \ghat(u) du \right\}^{-1} \int b(u) f(u) \shat(u) \ghat(u) du
\end{align*}
Therefore, by Fact 2, we have
\begin{align*}
\left(\widehat\Pi M_{\widehat{s}} - \Pi\right)^{m-1} (r_2) (x) & = \left(\widehat\Pi M_{\widehat{s}} - \Pi\right) (l) (x) \\
& = \int \widehat\Pi(x, u) \{\shat(u) - 1\} l(u) \ghat(u) du \\
& = \widehat\Pi\left((\shat - 1) l\right) (x).
\end{align*}
By Fact 1, we have
\begin{align*}
\left\|\left(\widehat\Pi M_{\widehat{s}} - \Pi\right)^{m-1} (r_2)\right\|_{2, \ghat} & = \left\|\widehat\Pi\left((\shat - 1) l\right)\right\|_{2, \ghat} \\
& \leq \Norm{(\shat - 1) l}_{2, \ghat} \leq \Norm{\shat - 1}_\infty \Norm{l}_{2, \ghat}.
\end{align*}
Repeating this argument $m-3$ times applied to $\Norm{l}_{2, \ghat}$, we obtain
\begin{align*}
\left\|\left(\widehat\Pi M_{\widehat{s}} - \Pi\right)^{m-1} (r_2)\right\|_{2, \ghat} \leq \Norm{\shat - 1}^{m-2}_\infty \left\| \left(\widehat\Pi M_{\widehat{s}} - \Pi\right) (r_2) \right\|_{2, \ghat}
\end{align*}
Furthermore,
\begin{align*}
\left\| \left(\widehat\Pi M_{\widehat{s}} - \Pi\right) (r_2) \right\|^2_{2, \ghat} & = \int \left(\widehat\Pi M_{\widehat{s}} - \Pi\right) (r_2)(x) \left(\widehat\Pi M_{\widehat{s}} - \Pi\right) (r_2)(x) \ghat(x) dx \\
& = \int \left(\widehat\Pi M_{\widehat{s}} - \Pi\right) (r_2)(x) \Pi \left(r_2\right)(x)\{\shat(x) - 1\} \ghat(x) dx
\end{align*}
The second line follows because $\left(\widehat\Pi M_{\widehat{s}} - \Pi\right) (r_2)$ belongs to the finite dimensional subspace and can be expressed as $b(x)^T\beta$ for some $\beta$. Therefore,
\begin{align*}
& \int \left(\widehat\Pi M_{\widehat{s}} - \Pi\right) (r_2)(x)\widehat\Pi M_{\shat} (r_2)(x) \ghat(x) dx \\
& = \beta^T\int b(x)\widehat\Pi M_{\shat} (r_2)(x) \ghat(x) dx \\
& = \beta^T\int b(x) b(x)^T \widehat\Omega^{-1} \int b(u) \widehat{s}(u)r_2(u)\widehat{g}(u)du \ghat(x) dx \\
& = \beta^T\int b(u)r_2(u) g(u) du \\
& = \beta^T \int b(x) b(x)^T \Omega^{-1} \int b(u) r(u) g(u) du \shat(x) \widehat{g}(x) dx \\
& = \int \beta^Tb(x) \Pi (r_2) (x) \widehat{s}(x) \widehat{g}(x) dx.
\end{align*}
By Cauchy-Schwarz:
\begin{align*}
\left\| \left(\widehat\Pi M_{\widehat{s}} - \Pi\right) (r_2) \right\|^2_{2, \ghat}  \leq \left\|\left(\widehat\Pi M_{\widehat{s}} - \Pi\right) (r_2) \right\|_{2, \ghat} \left\| \Pi \left(r_2\right)(\shat - 1) \right\|_{2, \ghat},
\end{align*}
implying
\begin{align*}
\left\|\left(\widehat\Pi M_{\widehat{s}} - \Pi\right) (r_2) \right\|_{2, \ghat} \lesssim \Norm{\shat - 1}_\infty \Norm{r_2}_g.
\end{align*}
This then yields
\begin{align*}
\left| T_2 \right| \lesssim \Norm{r_1}_g \Norm{r_2}_g \Norm{\shat -1}^{m-1}_\infty.
\end{align*}
The bounds on the terms involving $\Delta_j$ derived in Step 1 finally yield the result:
\begin{align*}
|T_2| \ \lesssim \left(\Norm{v}_g \Norm{q}_g + h^{\alpha \land \beta} \right) \Norm{\shat - 1}_\infty^{m-1} 
\end{align*}
\subsection{Variance}
The proof of the variance bound follows as in \cite{robins2017higher}. The key lemma that we rely on is their Lemma 10.2 (in their corrected version of the manuscript\footnote{\url{https://arxiv.org/pdf/1512.02174}}), which we adjust to keep track of the bandwidth $h$. 

Let $D_p$ denote the operator that makes a function degenerate:
\begin{align*}
D_p f(Z_1, \ldots, Z_j) = \sum_{A \subset \{1, \ldots, j\}} (-1)^{j-|A|} \E\{f(Z_1, \ldots, Z_j) \mid Z_i, i \in A\}
\end{align*}
Because $f = f_1(Z_1)\Pi_{1, 2} K_{ht}(A_2)\Pi_{2, 3} K_{ht}(A_3) \cdots \Pi_{j-1, j} \cdot f_2(Z_j)$ is already degenerate in $Z_1$ and $Z_j$, we can further write
\begin{align*}
& D_p f(Z_1, \ldots, Z_j) \\
& = \sum_{l = 0}^{j-2} (-1)^{j-l}\sum_{1 < i_1 < \cdots < i_l < j} f_1(Z_1) \Pi_{1, i_1} K_{ht}(A_{i_1})\Pi_{i_1, i_2} K_{ht}(A_{i_2}) \cdots \Pi_{i_l, j} \cdot f_2(Z_j)
\end{align*}
For example, $\varphi_4$ can be written as
\begin{align*}
\varphi_4(Z_1, Z_2, Z_3, Z_4) = f_1(Z_1) \cdot D_p [\Pi_{1, 2} K_{ht}(A_2)\Pi_{2, 3} K_{ht}(A_3)\Pi_{3, 4}] \cdot f_2(Z_4).
\end{align*}
The estimator of $\widehat\varphi_j$ is
\begin{align*}
\widehat\varphi_j(Z_1, \ldots, Z_j) = \widehat{f}_1(Z_1) \cdot D_{\widehat{p}} [\widehat\Pi_{1, 2} K_{ht}(A_2) \cdots K_{ht}(A_{j-1})\widehat\Pi_{{j-1}, j}] \cdot \widehat{f}_2(Z_j),
\end{align*}
so that $\widehat\varphi_j$ is degenerate relative to $\widehat{p}$. In bounding the variance, expectations are taken with respect to $p$, so that one cannot exploit directly the degeneracy of $\widehat\varphi_j$. 

The bound on the variance proceeds in three steps:

\begin{enumerate}
\item By Lemmas 14.1 and 14.2 in \cite{robins2017higher}, we have
\begin{align*}
	& \var\left[\Pn \widehat{f}_0 + \sum_{j = 2}^m \Un \{\widehat\varphi_j(Z_1, \ldots, Z_j) \mid D^n \right] \\
	& \leq 2 \var(\Pn \widehat{f}_0 \mid D^n) \\
	& \hphantom{\leq} + \sum_{j = 2}^m 2^j \sum_{l = 1}^j \frac{2^l j^{2l}}{n^l} \E\left(\left[\E\left\{S_j\widehat\varphi_j(Z_1, \ldots, Z_j) \mid Z_1, \ldots, Z_l, D^n \right\}\right]^2 \mid D^n \right),
\end{align*}
where $S_j$ symmetrizes the function $\widehat\varphi_j$. Notice that we rely on the fact that $\Un \varphi_j = \Un S_j \varphi_j$, since the empirical U-statistic measure takes already care of the symmetrization. However, the formula for the bound on the variance in their Lemma 14.1 requires the kernel of the U-statistic to be symmetric, which then requires a bound on $\E\{\widehat\varphi_j(Z_1, \ldots, Z_j) \mid Z_B\}$ for an arbitrary set of indices $B$ of size $l$ instead of the first $l$ indices $1, \ldots, l$, as outlined in the following step. Further, because the kernel and the observations are bounded, a change of variables argument yields that
\begin{align*}
	\var(\Pn \widehat{f}_0 \mid D^n) \lesssim (nh)^{-1}.
\end{align*}
\item As described in the proof of Theorems 8.1 and 8.2 in \cite{robins2017higher}, because the observations are i.i.d., we have that
\begin{align*}
	& \E\left( \left[ \E\left\{ S_j\widehat\varphi_j(Z_1, \ldots, Z_j) \mid Z_1, \ldots, Z_l, D^n\right\} \right]^2 \mid D^n \right) \\
	& \leq \max_{\sigma} \E\left( \left[ \E\left\{ \widehat\varphi_j(Z_{\sigma(1)}, \ldots, Z_{\sigma(j)}) \mid Z_1, \ldots, Z_l, D^n\right\} \right]^2 \mid D^n \right) \\
	& = \max_{\substack{B \subset \{1, \ldots, j\} \\  |B| = l}}\E\left( \left[ \E\left\{\widehat\varphi_j(Z_1, \ldots, Z_j) \mid Z_B, D^n\right\} \right]^2 \mid D^n \right)
\end{align*}
\item By Lemma \ref{robins_102_adjusted}, and letting $\epsilon_n = h^{\alpha \land \beta} \ \lor \ \|v\|_{4, g} \ \lor \ \|q\|_{4, g} \ \lor \ \|f\|_\infty$, the variance is further upper bounded by a constant multiple of
\begin{align*}
	\sum_{j = 1}^m c^j \sum_{l = 1}^j j^{2l} \epsilon_n^{2(j - l)} \frac{k^{l-1}}{(nh)^{l}} = \sum_{l = 1}^m \left(\sum_{j = l}^m c^j j^{2l} \epsilon_n^{2(j-l)} \right) \frac{k^{l-1}}{(nh)^l},
\end{align*}
where 
\begin{align*}
	& f = \widehat{p}(t, x) - p(t, x), \quad v(x) = \widehat\mu(t, x) - \mu(t, x) \\
	& \text{and } q(x) = 1/\widehat\pi(t \mid x) - 1/\pi(t \mid x).
\end{align*}
\end{enumerate}
To prove Lemma \ref{robins_102_adjusted} below, we rely on Lemma 10.1 in \cite{robins2017higher}, which we rewrite below adjusted to our setting. Suppose $Z_B$ contains both $Z_1$ and $Z_j$ and let $1 < b_1 < \cdots < j$ such that $B = \{1, \ldots, j\} \setminus \{b_1, \ldots, b_j\}$. Then
\begin{align*}
& \E_p\left\{ \widehat\varphi_j(Z_1, \ldots, Z_j) \mid Z_B, D^n \right\} \\
& = f_1(Z_1) \E_p\left\{ D_{\widehat{p}}\widehat\Pi_{1, 2}K_{ht}(A_2)\widehat\Pi_{2, 3} \cdots K_{ht}(A_{j-1})\widehat\Pi_{j-1, j} \mid Z_B, D^n \right\}f_2(Z_j) \\
& = \int \cdots \int D^{2, b_1 - 1}_{\widehat{p}}\left\{\widehat\Pi_{1, 2}\prod_{i = 2}^{b_1 - 1} K_{ht}(A_i)\widehat\Pi_{i, i+1} \right\}(\widehat{s}_{b_1} - 1) \times \\
& \hphantom{= \int \cdots \int } \times D^{b_1+1, b_2 - 1}_{\widehat{p}}\left\{\widehat\Pi_{b_1, b_1 +1}\prod_{i = b_1+1}^{b_2 - 1} K_{ht}(A_l)\widehat\Pi_{i, i+1} \right\}(\widehat{s}_{b_2} - 1) \times \cdots \times \\
& \hphantom{= \int \cdots \int} \times D^{b_s+1, j - 1}_{\widehat{p}}\left\{\widehat\Pi_{b_s, b_s +1}\prod_{i = b_s+1}^{j - 1} K_{ht}(A_i)\widehat\Pi_{i, i+1} \right\}\times \\
& \hphantom{= \int \cdots \int} \quad\quad \times\widehat{g}(x_{b_1}) \cdots \widehat{g}(x_{b_s}) dx_{b_1} \cdots dx_{b_s}
\end{align*}
where $D_{\widehat{p}}^{k, l}$ is the operator that makes a function degenerate relative to the variables $Z_k, \ldots, Z_l$ with respect to the distribution $\widehat{p}$. For example, for $j = 3$ and $B = \{1, 3\}$, we have
\begin{align*}
& \widehat\varphi_3(Z_1, Z_2, Z_3) = - f_1(Z_1) \widehat\Pi_{1, 2} K_{ht}(A_2)\widehat\Pi_{2, 3} f_3(Z_3) + f_1(Z_1)\widehat\Pi_{1, 3}f_2(Z_3)
\end{align*}
and
\begin{align*}
&  \E\{\widehat\varphi_3(Z_1, Z_2, Z_3) \mid Z_1, Z_3, D^n\} \\
& = -f_1(Z_1) \int \widehat\Pi_{1, 2} \widehat\Pi_{2, 3} \frac{g(x_2)}{\widehat{g}(x_2)} \widehat{g}(x_2) dx_2 + f_1(Z_1) \widehat\Pi_{1, 3} f_2(Z_3) \\
& = -f_1(Z_1) \int \widehat\Pi_{1, 2} (\widehat{s}_2 - 1) \widehat\Pi_{2, 3} \widehat{g}(x_2) dx_2f_3(Z_3),
\end{align*}
where $\widehat{s}_2 = g(x_2) / \widehat{g}(x_2) - 1$. As in the proof of Lemma 10.2 in \cite{robins2017higher}, when computing
\begin{align*}
\E_{\widehat{p}} \left(\left[ \E_{p}\left\{ \widehat\varphi_j(Z_1, \ldots, Z_j) \mid Z_B, D^n\right\} \right]^2\right),
\end{align*}
we will ignore the degeneracy operation as it would not increase second moments. We will also ignore it when $Z_B$ does not include either $Z_1$ or $Z_j$ because the integration with respect to $\widehat{g}(x_1) dx_1$ or $\widehat{g}(x_j) dx_j$ when computing
\begin{align*}
\E_p\left\{ \widehat\varphi_j(Z_1, \ldots, Z_j) \mid Z_B, D^n\right\}
\end{align*}would commute with the operators $D_{\widehat{p}}^{k, l}$ where $k$ and $l$ are never equal to either $1$ or $j$. 
\begin{lemma}\label{robins_102_adjusted}
For any $B \subset \{1, \ldots, j\}$, there exists a constant $M$ such that the following holds:
\begin{align*}
	& \E_p\left( \left[\E_p\left\{ \widehat\varphi_j(Z_1, \ldots, Z_j) \mid Z_B, D^n \right\} \right]^2 \right) \\
	& \leq M^j \left( h^{\alpha \land \beta} \ \lor \ \|v\|_{4, g} \ \lor \ \|q\|_{4, g} \ \lor \ \|f\|_\infty \right)^{2(j - |B|)} \frac{k^{|B| - 1}}{h^{|B|}}.
\end{align*}
\end{lemma}
\begin{proof}
First, notice that, by a change of measure from $p$ to $\widehat{p}$:
\begin{align*}
	\E_p\left( \left[\E_p\left\{ \widehat\varphi_j(Z_1, \ldots, Z_j) \mid Z_B \right\} \right]^2 \right) \leq \left\|\frac{p}{\widehat{p}} \right\|^{|B|}_\infty \E_{\widehat{p}}\left( \left[\E_p\left\{ \widehat\varphi_j(Z_1, \ldots, Z_j) \mid Z_B \right\} \right]^2 \right).
\end{align*}
We will rely on Fact 3 repeatedly using the bound 
\begin{align*}
	\int \widehat\Pi^2_{i, i+1} \widehat{g}(x_{i+1}) dx_{i+1} = \widehat\Pi_{i, i} \lesssim k.
\end{align*} 
In particular, we will rely on
\begin{align*}
	& \int \widehat\Pi^2_{i, i+1} K^2_{ht}(a_{i+1}) \widehat{p}(a_{i + 1}, x_{i+1}) da_{i+1} dx_{i+1} \\
	& = \int \widehat\Pi^2(x_i, x_{i+1}) \frac{\int K_{ht}^2(a_{i+1}) \widehat{p}(a_{i + 1}, x_{i + 1}) da_{i+1}}{\widehat{g}(x_{i + 1})} \widehat{g}(x_{i+1}) dx_{i+1}  \\
	& \lesssim \frac{1}{h} \int \widehat\Pi^2_{i, i+1} \widehat{g}(x_{i+1}) dx_{i + 1} \\
	& \lesssim \frac{k}{h}.
\end{align*}

There are four cases to consider, depending on whether $B$ includes either $1$ or $j$. 

\textbf{Case 1: $B$ includes 1 and $j$}. Let $\{b_1, \ldots, b_s\} = \{1, \ldots, j\} \setminus B$. Lemma 10.1 in \cite{robins2017higher}, applied with their $\overline{A}_j$ and $\widehat{g}$ replaced by our $K_{ht}(A_j)$ and $\widehat{g}$, respectively, and the fact that a projection cannot increase the second moment yield that it is sufficient to bound the expectation of the square of
\begin{align*}
	& f_1(Z_1) \int \cdots \int \widehat\Pi_{1, 2} \prod_{i = 2}^{b_1 - 1} K_{ht}(A_i) \widehat\Pi_{i, i+1} (\widehat{s}_{b_1} - 1)\widehat\Pi_{b_1, b_1 + 1} \\
	& \hphantom{f_1(Z_1) \int \cdots \int} \quad \times \prod_{i = b_1+1}^{b_2 - 1} K_{ht}(A_i) \widehat\Pi_{i, i+1}(\widehat{s}_{b_2} - 1) \times \cdots \\
	& \hphantom{f_1(Z_1) \int \cdots \int} \quad \times \prod_{i = b_{s-1} + 1}^{b_s - 2} K_{ht}(A_i)\widehat\Pi_{i, i+1} \times \widehat\Pi_{b_s-1, b_s}(\widehat{s}_{b_s} - 1) \widehat\Pi_{b_s, b_s + 1} \\
	& \hphantom{f_1(Z_1) \int \cdots \int} \quad \times \widehat{g}(x_{b_1}) \cdots \widehat{g}(x_{b_s}) dx_{b_1} \cdots dx_{b_s} \\
	& \quad \times \prod_{i = b_s + 1}^{j-1} K_{ht}(A_i) \widehat\Pi_{i, i+1} f_2(Z_j)
\end{align*}
We proceed by peeling from right to left, taking the expectation with respect to $Z_j, Z_{j-1}, \ldots, Z_{b_s + 2}$:
\begin{align*}
	&\int \cdots \int \left\{\prod_{i = b_s + 1}^{j-1} K_{ht}(a_i) \widehat\Pi_{i, i+1} f_2(z_j)\right\}^2 \widehat{p}(z_{b_s + 2}) \cdots \widehat{p}(z_{j}) dz_{b_s + 2} \cdots dz_j \\
	& \lesssim \frac{k^{j - b_s - 1}}{h^{j-b_s - 1}} K^2_{ht}(A_{b_s + 1}).
\end{align*}
Next, we consider integration with respect to $x_{b_s}$ in
\begin{align*}
	& f_1(Z_1) \int \cdots \int \widehat\Pi_{1, 2} \prod_{i = 2}^{b_1 - 1} K_{ht}(A_i) \widehat\Pi_{i, i+1} (\widehat{s}_{b_1} - 1)\widehat\Pi_{b_1, b_1 + 1} \\
	& \hphantom{f_1(Z_1) \int } \times  \prod_{i = b_1+1}^{b_2 - 1} K_{ht}(A_i) \widehat\Pi_{i, i+1}(\widehat{s}_{b_2} - 1) \times \cdots \\
	& \hphantom{f_1(Z_1) \int } \times \prod_{i = b_{s-1} + 1}^{b_s - 2} K_{ht}(A_i)\widehat\Pi_{i, i+1} \times K_{ht}(A_{b_s - 1}) \widehat\Pi_{b_s-1, b_s}(\widehat{s}_{b_s} - 1) \widehat\Pi_{b_s, b_s + 1} \\
	& \hphantom{f_1(Z_1) \int } \times \widehat{g}(x_{b_1}) \cdots \widehat{g}(x_{b_s}) dx_{b_1} \cdots dx_{b_s} \\
	& \equiv T(Z_1, \ldots, Z_{b_{s} - 1}) \int \widehat\Pi_{b_s - 1, b_s}(\widehat{s}_{b_s} - 1) \widehat\Pi_{b_s, b_s + 1} \widehat{g}(x_{b_s}) dx_{b_s},
\end{align*}
where 
\begin{align*}
	& T(Z_1, \ldots, Z_{b_{s} - 1})  = \\
	& f_1(Z_1) \int \cdots \int \widehat\Pi_{1, 2} \prod_{i = 2}^{b_1 - 1} K_{ht}(A_i) \widehat\Pi_{i, i+1} (\widehat{s}_{b_1} - 1)\widehat\Pi_{b_1, b_1 + 1} \\
	& \hphantom{f_1(Z_1) \int \cdots \int} \times \prod_{i = b_1+1}^{b_2 - 1} K_{ht}(A_i) \widehat\Pi_{i, i+1}(\widehat{s}_{b_2} - 1) \times \cdots \\
	& \hphantom{f_1(Z_1) \int \cdots \int} \times \widehat{g}(x_{b_1}) \cdots \widehat{g}(x_{b_{s-1}}) dx_{b_1} \cdots dx_{b_{s-1}} \\
	& \hphantom{f_1(Z_1) \int \cdots \int} \times \prod_{i = b_{s-1} + 1}^{b_s - 2} K_{ht}(A_i)\widehat\Pi_{i, i+1} \times K_{ht}(A_{b_s - 1}) 
\end{align*}
Notice that 
\begin{align*}
	& \int T^2(Z_1, \ldots, Z_{b_{s} - 1})\left\{ \int \widehat\Pi_{b_s - 1, b_s}(\widehat{s}_{b_s} - 1) \widehat\Pi_{b_s, b_s + 1} \widehat{g}(x_{b_s}) dx_{b_s}\right\}^2 \\
	& \hphantom{\int T^2(Z_1, \ldots, Z_{b_{s} - 1})\left\{ \right\}} \times \frac{\int K^2_{ht}(a_{b_s + 1}) \widehat{p}(a_{b_s + 1},  x_{b_s + 1}) da_{b_s + 1}}{\widehat{g}(x_{b_s + 1})} \widehat{g}(x_{b_s + 1}) dx_{b_s + 1} \\
	& \lesssim h^{-1}\int T^2(Z_1, \ldots, Z_{b_{s} - 1})\left\{ \int \widehat\Pi_{b_s - 1, b_s}(\widehat{s}_{b_s} - 1) \widehat\Pi_{b_s, b_s + 1} \widehat{g}(x_{b_s}) dx_{b_s}\right\}^2 \\
	& \hphantom{\lesssim h^{-1}\int } \quad \times \widehat{g}(x_{b_s + 1}) dx_{b_s + 1} \\
	& \leq  h^{-1} T^2(Z_1, \ldots, Z_{b_{s} - 1}) \int \widehat\Pi^2_{b_s - 1, b_s} (\widehat{s}_{b_s} - 1)^2 \widehat{g}(x_{b_s}) dx_{b_s} \\
	& \leq T^2(Z_1, \ldots, Z_{b_{s} - 1}) \frac{k}{h} \|\widehat{s} - 1\|_\infty^2
\end{align*}
Thus, we have removed the variable $x_{b_s}$ and reached the bound:
\begin{align*}
	& \int \cdots \int \left\{ f_1(Z_1) \int \cdots \int \widehat\Pi_{1, 2} \prod_{i = 2}^{b_1 - 1} K_{ht}(A_i) \widehat\Pi_{i, i+1} (\widehat{s}_{b_1} - 1)\widehat\Pi_{b_1, b_1 + 1} \right. \\
	& \left. \hphantom{f_1(Z_1) \int } \times \prod_{i = b_1+1}^{b_2 - 1} K_{ht}(A_i) \widehat\Pi_{i, i+1}(\widehat{s}_{b_2} - 1) \times \cdots \right. \\
	& \left. \hphantom{f_1(Z_1) \int } \times \prod_{i = b_{s-1} + 1}^{b_s - 2} K_{ht}(A_i)\widehat\Pi_{i, i+1} \times K_{ht}(A_{b_s - 1}) \widehat\Pi_{b_s-1, b_s}(\widehat{s}_{b_s} - 1) \widehat\Pi_{b_s, b_s + 1} \right. \\
	& \left. \hphantom{f_1(Z_1) \int } \vphantom{\prod_{i = b_{s-1} + 1}^{b_s - 2}} \times \widehat{g}(x_{b_1}) \cdots \widehat{g}(x_{b_s}) dx_{b_1} \cdots dx_{b_s} \right\}^2 \\
	& \hphantom{\int \cdots \int}\times \widehat{p}(z_{b_s+1}) \cdots \widehat{p}(z_j) dz_{b_s + 1}\cdots  dz_{j} \\
	& \lesssim \frac{k^{j -b_s}}{h^{j - b_s}} \|\widehat{s} - 1\|_\infty^2  T^2(Z_1, \ldots, Z_{b_s - 1})
\end{align*}
Next, we peel again from right to left to remove the term
\begin{align*}
	\prod_{i = b_{s-1} + 1}^{b_s - 2} K^2_{ht}(A_i)\widehat\Pi^2_{i, i+1} \times K^2_{ht}(A_{b_s - 1}) 
\end{align*}
from $T^2(Z_1, \ldots, Z_{b_s} - 1)$. When we integrate the display above with respect to variables $Z_{b_s - 1}, \ldots, Z_{b_{s-1} + 2}$, we gain a factor $$k^{b_s - b_{s-1} - 2} h^{-b_s + b_{s-1} - 2} K^2_{ht}(A_{b_{s-1} + 1}).$$ When integrating with respect to $Z_{b_{s-1} + 1}$ and $Z_{b_{s-1}}$, we gain a factor $(k/h)\|\widehat{s} - 1\|_\infty^2$. Repeating this process for up until variable $Z_{b_1}$, we accumulate in total a factor 
\begin{align*}
	\frac{k^{j -b_1 - (s-1)}}{h^{j - b_1- (s-1)}} \|\widehat{s} - 1\|_\infty^{2s}.
\end{align*}
This means that the expected square that we need to bound can be bounded by a constant multiple of
\begin{align*}
	& \E_{\widehat{p}} \left[ \widehat{f}^2_1(Z_1) \widehat\Pi_{1, 2} \prod_{i = 2}^{b_1 - 2} K^2_{ht}(A_i) \widehat\Pi^2_{i, i+1} K^2_{ht}(A_{b_1 - 1}) \right] \times \frac{k^{j -b_1 - (s-1)}}{h^{j - b_1- (s-1)}} \|\widehat{s} - 1\|_\infty^{2s} \\
	& \lesssim \frac{k^{b_1 - 2}}{h^{b_1 - 1}} \times \frac{k^{j -b_1 - (s-1)}}{h^{j - b_1- (s-1)}} \|\widehat{s} - 1\|_\infty^{2s} \\
	& = \frac{k^{j - s - 1}}{h^{j - s}} \|\widehat{s} - 1\|_\infty^{2s}.
\end{align*}
The integration with respect to $Z_1$ is adds an additional $h^{-1}$ factor.

\textbf{Case 2: $B$ includes $j$ but not $1$}. In this case, we first need to integrate with respect to $Z_1$ (recalling $v(a, x) = \mu(a, x) - \widehat\mu(a, x)$ and $v(x) = v(t, x)$):
\begin{align*}
	& \int \widehat{f}_1(z_1) \widehat\Pi_{1, 2} p(a_1, x_1) d a_1 dx_1 \\
	& = \int K_{ht}(a_1)v(a_1, x_1) \widehat\Pi(x_1, x_2) p(a_1, x_1) da_1 dx_1 \\
	& = \int v(x_1) \widehat\Pi(x_1, x_2) g(x_1) dx_1 \\
	& \hphantom{=} + \int K_{ht}(a_1) \{v(a_1, x_1) - v(t, x_1)\} \pi(a_1 \mid x_1) d a_1 \widehat\Pi(x_1, x_2) p(x_1) dx_1 \\
	& = \widehat\Pi \left( v \widehat{s}\right)_2 + \int \frac{\Delta_1(x_1)}{\int K_{ht}(a) \pi(a \mid x_1) da} \widehat\Pi(x_1, x_2) g(x_1) dx_1 \\
	& = \widehat\Pi\left( (v + \Delta_v) \widehat{s}\right)_2,
\end{align*}
where
\begin{align*}
	& \Delta_1(x) = \int K_{ht}(a)\{v(a, x) - v(t, x)\} \pi(a \mid x) da, \\
	& \Delta_v(x) = \frac{\Delta_1(x)}{\int K_{ht}(a) \pi(a \mid x) da},
\end{align*}
and the notation $(f)_j$ denotes evaluation of the function $f$ at point $x_j$. As shown in proving the bound on the bias (Step 1), $\|\Delta_1\|_\infty \lesssim h^{b+1} + h^\alpha \lesssim h^{\alpha \land \beta}$. Under the assumption that $\int K_{ht}(a) \pi(a \mid x) da $ is bounded away from zero, we also have that $\|\Delta_v\|_\infty \lesssim h^{\alpha \land \beta}$.  
Notice that
\begin{align*}
	\int \left\{\widehat\Pi\left( (v + \Delta_v) \widehat{s}\right)_2\right\}^2 \widehat{g}(x_2) dx_2 & \leq \int \left[\left\{v(x_2) + \Delta_v(x_2)\right\} \widehat{s}(x_2) \right]^2  \widehat{g}(x_2) dx_2 \\
	& \lesssim \|v\|_g^2 \ + \ h^{2(\alpha \land \beta)}
\end{align*}
because projections do not increase the second moment. 

If $b_1 = 2$, then we also integrate with respect to $Z_2$ yielding 
\begin{align*}
	\int \widehat\Pi\left((v + \Delta_v)\widehat{s}\right)_2(\widehat{s}_2 - 1) \widehat\Pi_{2, 3} \widehat{g}(x_2) dx_2 = \widehat\Pi\left(\left\{\widehat\Pi\left((v + \Delta_v) \widehat{s}\right)\right\}(\widehat{s}- 1)\right)_3.
\end{align*}
Notice that
\begin{align*}
	& \int \left[\widehat\Pi\left(\left\{\widehat\Pi\left((v + \Delta_v) \widehat{s}\right)\right\}(\widehat{s}- 1)\right)_3 \right]^2 \widehat{g}(x_3) dx_3 \\
	& \leq  \int \left[\left\{\widehat\Pi\left((v + \Delta_v) \widehat{s}\right)_3\right\}(\widehat{s}(x_3) - 1)\right]^2 \widehat{g}(x_3) dx_3 \\
	& \leq \|\widehat{s} - 1\|_\infty^2 \int \left\{\widehat\Pi\left( (v + \Delta_v) \widehat{s}\right)_3\right\}^2 \widehat{g}(x_3) dx_3 \\
	& \lesssim \left\{\|v\|^2_{g} \ + \ h^{2(\alpha \land \beta)}\right\}\|\widehat{s} - 1\|_\infty^2.
\end{align*}
We repeat this until we reach the first variable not belonging to $Z_B$, after which the peeling process proceeds right to left as in Case 1. That is, we stop as soon as $b_r > r+1$. Thus, we obtain
\begin{align*}
	& \E_{\widehat{p}}\left[\left\{\widehat\Pi \left( \cdots \widehat\Pi\left((v + \Delta_v)\widehat{s}\right)(\widehat{s} - 1) \right)_{r + 1} \right\}^2 K^2_{ht}(A_{r + 1})\right] \\
	& \quad\quad\quad \times \frac{k^{j - b_{r-1} - (s - r) - 1}}{h^{j - b_{r-1} - (s - r) - 1}}\|\widehat{s} - 1\|_\infty^{2s - 2r} \\
	& = \E_{\widehat{p}}\left[\left\{\widehat\Pi \left( \cdots \widehat\Pi\left((v + \Delta_v)\widehat{s}\right)(\widehat{s} - 1) \right)_{r + 1} \right\}^2  \right. \\
	& \hphantom{= \E_{\widehat{p}}} \left. \vphantom{\left\{\widehat\Pi \left( \cdots \widehat\Pi\left((v + \Delta_v)\widehat{s}\right)(\widehat{s} - 1) \right)_{r + 1} \right\}^2 } \quad \times  \int K_{ht}(a) \widehat\pi(a \mid X_{r+1}) da \times \frac{\int K^2_{ht}(a) \widehat{p}(a_{r+1} \mid X_{r+1}) d a}{\int K_{ht}(a) \widehat\pi(a \mid X_{r+1}) da} \right] \\
	& \hphantom{=} \quad\quad\quad\quad \times \frac{k^{j - b_{r-1} - (s - r) - 1}}{h^{j - b_{r-1} - (s - r) - 1}}\|\widehat{s} - 1\|_\infty^{2s - 2r} \\
	& \lesssim \E_{\widehat{p}}\left[\left\{\widehat\Pi \left( \cdots \widehat\Pi\left((v + \Delta_v)\widehat{s}\right)(\widehat{s} - 1) \right)_{r + 1} \right\}^2  \int K_{ht}(a) \widehat\pi(a \mid X_{r+1}) da)\right]\\
	& \hphantom{\E_{\widehat{p}}\left[\left\{\widehat\Pi\right.\right.} \quad\quad \times \frac{k^{j - b_{r-1} - (s - r) - 1)}}{h^{j - b_{r-1}- (s - r)}}\|\widehat{s} - 1\|_\infty^{2s - 2r} \\
	& \lesssim \left\{\|v\|^2_{g} \ + \ h^{2(\alpha \land \beta)}\right\} \|\widehat{s} - 1\|^{2r-2}_\infty \frac{k^{j - b_{r-1} - (s  - r) - 1}}{h^{j - b_{r-1} - (s-r)}}\|\widehat{s} - 1\|_\infty^{2s - 2r}  \\
	& \lesssim \left\{\left(\|v\|_{g} \ + \ h^{\alpha \land \beta}\right)  \ \lor \ \|\widehat{s} - 1\|_\infty\right\}^{2s} \frac{k^{j - s -1}}{h^{j - s}}
\end{align*}
since $b_{r-1} = r$. 

\textbf{Case 3: $B$ includes 1 but not $j$.} This case is similar to Case 2. The main difference is that we start by integrating with respect to the variable $Z_j$ (recalling that $q(a, x) = 1/\widehat\pi(a \mid x) - 1/\pi(a \mid x)$ and $q(x) = q(t, x)$):
\begin{align*}
	& \int \widehat\Pi_{j-1, j} f_2(z_j) \widehat{p}(z_j) dz_j \\
	& = \int K_{ht}(a_j) q(a_j, x_j) \widehat\Pi(x_{j-1}, x_{j})  p(a_j, x_j) da_jdx_j \\
	& = \int q(x_j) \widehat\Pi(x_{j-1}, x_j) dx_j \\
	& \hphantom{=} + \int K_{ht}(a_j)\{q(a_j, x_j) - q(t, x_j)\} \pi(a_j \mid x_j) da_j \widehat\Pi(x_{j-1}, x_j) p(x_j) dx_j \\
	& = \widehat\Pi(q \widehat{s})_{j-1} + \int \frac{\Delta_1(x)}{\int K_{ht}(a) \pi(a \mid x_{j}) da} \widehat\Pi(x_{j-1}, x_j) g(x_j) dx_j \\
	& = \widehat\Pi((q + \Delta_q) \widehat{s})_{j-1},
\end{align*}
where this time
\begin{align*}
	& \Delta_1(x) = \int K_{ht}(a) \{q(a, x) - q(t, x)\} \pi(a \mid x) da \\
	& \Delta_q(x) = \frac{\Delta_1(x)}{\int K_{ht}(a) \pi(a \mid x) da}.
\end{align*}
As shown in the bound for the bias when controlling $\Delta_3(x)$ in Step 1, and under the assumption that $\int K_{ht}(a) \pi(a \mid x) da$ is bounded away from zero, we have $\|\Delta_q\|_\infty \lesssim h^{\beta}$.

Repeating the arguments as in Case 2, we reach a bound of order
\begin{align*}
	\left\{\left(\|q\|_g \ + \ h^{\beta}\right) \ \lor \ \|\widehat{s} - 1\|_\infty\right\}^{2s} \frac{k^{j - s - 1}}{h^{j - s}}.
\end{align*}
\textbf{Case 4: $B$ includes 1 and $j$.} To analyze this case, we can first integrate from left to right as done in case 2 and from right to left as done in case 3. We need to bound the expectation of the square of
\begin{align*}
	& \widehat\Pi \left(\cdots \widehat\Pi\left( \widehat\Pi\left((v + \Delta_v)\widehat{s}\right)(\widehat{s} - 1) \right) \right)_{r + 1} \\
	& \times \int \cdots \int \prod_{i = b_{r-1} + 1}^{b_{r} - 1} K_{ht}(A_i) \widehat\Pi_{i, i + 1} \widehat\Pi_{b_{r}, b_{r} + 1} \times \cdots \\
	& \quad\quad \times \prod_{i = b_{r' - 1} + 1}^{b_{r'} - 2}K_{ht}(A_i) \widehat\Pi_{i, i + 1} K_{ht}(A_{b_{r'} - 1}) \prod_{i = r}^{r^{'}-1}(\widehat{s}_{b_i} - 1)\widehat{g}(x_{b_i})dx_{b_i} \\
	& \times \widehat\Pi \left(\cdots \widehat\Pi\left( \widehat\Pi\left((q + \Delta_q)\widehat{s}\right)(\widehat{s} - 1) \right) \right)_{b_{r'} - 1}.
\end{align*}
By Cauchy-Schwarz, we can further bound the expected value of the square by the square root of the product of:
\begin{align*}
	& \E_{\widehat{p}} \left[ \left\{ \widehat\Pi \left(\cdots \widehat\Pi\left( \widehat\Pi\left((v + \Delta_v)\widehat{s}\right)(\widehat{s} - 1) \right) \right)_{r + 1} \right\}^4 \right. \\
	& \quad \quad \times \left\{ \left. \int \cdots \int \prod_{i = b_{r-1} + 1}^{b_{r} - 1} K_{ht}(A_i) \widehat\Pi_{i, i + 1} \widehat\Pi_{b_{r}, b_{r} + 1} \times \cdots \right. \right. \\
	& \hphantom{\quad \quad \times \left\{ \right.} \left.\left. \quad \times \prod_{i = b_{r' - 1} + 1}^{b_{r'} - 2}K_{ht}(A_i) \widehat\Pi_{i, i + 1} K_{ht}(A_{b_{r'} - 1})\prod_{i = r}^{r^{'}-1}(\widehat{s}_{b_i} - 1) \widehat{g}(x_{b_i}) dx_{b_i}\right\}^2 \right]
\end{align*}
times
\begin{align*}
	& \E_{\widehat{p}} \left[ \left\{ \widehat\Pi \left(\cdots \widehat\Pi\left( \widehat\Pi\left((q + \Delta_q)\widehat{s}\right)(\widehat{s} - 1) \right) \right)_{b_{r'} - 1} \right\}^4 \right. \\
	& \quad \quad \times \left\{ \left. \int \cdots \int \prod_{i = b_{r-1} + 1}^{b_{r} - 1} K_{ht}(A_i) \widehat\Pi_{i, i + 1} \widehat\Pi_{b_{r}, b_{r} + 1} \times \cdots \right. \right. \\
	& \hphantom{\quad \quad \times \left\{ \right.} \left.\left. \quad \times \prod_{i = b_{r' - 1} + 1}^{b_{r'} - 2}K_{ht}(A_i) \widehat\Pi_{i, i + 1} K_{ht}(A_{b_{r'} - 1})\prod_{i = r}^{r^{'}-1}(\widehat{s}_{b_i} - 1) \widehat{g}(x_{b_i}) dx_{b_i}\right\}^2 \right].
\end{align*}
The square root of the product can be seen to be bounded by
\begin{align*}
	& \left\{\|v\|_{4, g}^2 \ + \ h^{2(\alpha \land \beta)}\right\} \|\widehat{s} - 1\|_\infty^{2(r - 1)} \frac{k^{b_{r'} - b_{r-1} - (r{'} - r) - 2}}{h^{b_{r'} - b_{r-1} - (r^{'} - r)-1}} \|\widehat{s} - 1 \|_\infty^{2(r' - r)} \\
	& \quad \times \left\{\|q\|_{4, g}^2 \ + \ h^{2\beta}\right\} \|\widehat{s} - 1\|_\infty^{2(s - r{'} - 1)}
\end{align*}
where $\|v\|^2_{4, g} = \int v^4(x) g(x) dx$. Notice that $b_{r'} = j - s + r' + 1$ and $b_{r-1} = r$, so the quantity above can be further upper bounded as
\begin{align*}
	(h^{\alpha \land \beta} \ \lor \ \|v\|_{4, g} \ \lor \ \|q\|_{4, g} \ \lor \ \|\widehat{s} - 1\|_\infty)^{2s} \frac{k^{j - s - 1}}{h^{j - s}}.
\end{align*}
Finally, notice that, for $f(x) = \widehat{p}(t, x) - p(t, x)$, we have
\begin{align*}
	\left| \shat(x) - 1 \right| = \left| \frac{\int K_{ht}(a) \{p(a, x) - \widehat{p}(a, x)\} da dx}{\widehat{g}(x)}\right| \lesssim |p(t, x) - \widehat{p}(t, x)| \ + \ h^\beta,
\end{align*}
so that replacing $\|\widehat{s} - 1 \|_\infty$ with $\|f\|_\infty$ does not change the order of the bound since $h^\beta$ is always included.
\end{proof}

\section{Sensitivity analysis to the no-unmeasured-confounding assumption \label{section:sens}}
In this section, we briefly outline a simple pseudo-outcome regression method to carry out flexible, nonparametric sensitivity analysis to the no-unmeasured-confounding assumption, i.e., when $Y^t \not\ind A \mid X$ so that $\int \mu(t, x) d\Pb(x)$ can no longer be interpreted as the dose-response curve. To the best of our knowledge, this is the first nonparametric sensitivity analysis method for continuous treatment effects. \cite{bonvini2022MSM} propose an extension to Rosenbaum's sensitivity model for binary treatments as follows. Let $U$ be such that $Y^a \ind A \mid (X, U)$ and recall that $\E(Y^t) = \E\{Yp(t) / \pi(t \mid X, U) \mid A = a\}$. Let $\gamma \geq 1$ be a user-specified sensitivity parameter. Departures from the no-unmeasured-confounding assumption are parametrized by considering all densities of $A$ given $(X, U)$, $\pi(a \mid x, u)$, in the class
\begin{align*}
\Pi(\gamma) = \left\{ \pi(a \mid x, u): \frac{1}{\gamma} \leq \frac{\pi(a \mid x, u)}{\pi(a \mid x)} \leq \gamma \right\}
\end{align*}
When $\gamma = 1$, corresponding to the case when the measured covariates are sufficient to characterize the treatment selection process, one has the usual identification formula $$\E(Y^t) = \E\{w(t, X)Y \mid A = t\} = \int \mu(t, x) d\Pb(x).$$ 
Throughout, we assume that the outcome $Y$ is continuous. Lemma 2 in \cite{bonvini2022MSM} shows that valid bounds on $\E(Y^a)$ under the sensitivity model $\Pi(\gamma)$ are
\begin{align*}
& \theta_l(t; \gamma) = \int \E[Y\gamma^{\sgn\{q_l(t, x) - Y\}} \mid A = t, X = x]d\Pb(x) \\
& \theta_u(t; \gamma) = \int \E[Y\gamma^{\sgn\{Y - q_u(t, x)\}} \mid A = t, X = x]d\Pb(x)
\end{align*}
where $q_l(A, X)$ (resp. $q_u(A, X)$) is the $1/(1 + \gamma)$ (resp. $\gamma / (1 + \gamma)$)-quantile of $Y$ given $(A, X)$. In other words, for a given, user-specified $\gamma$, if $\pi(a \mid x, u) \in \Pi(\gamma)$, then $\theta_l(a; \gamma) \leq \E(Y^a) \leq \theta_u(a; \gamma)$. 

A DR-Learner estimator of the bounds above can be computed by appropriately modifying the original pseudo-outcome $\varphi(Z) = w(A, X) \{Y - \mu(A, X)\} + \int \mu(A, x) d\Pb(x)$ and regressing it onto $A$. For $j = \{l, u\}$, define
\begin{align*}
& \varphi_j(Z; \gamma) \equiv \varphi_j(Z; w, \kappa_j, q_j, \gamma) \\
& \hphantom{\varphi_j(Z; \gamma)} = w(A, X)\{s_j(Z; q_j) - \kappa_j(A, X; q_j)\} \\
& \hphantom{\varphi_j(Z; \gamma) \equiv} \quad + \int \kappa_j(A, x; q_j) d\Pb(x), \text{ for } \\
& s_l(Z; q_l) = q_l(A, X) + \{Y - q_l(A, X)\}\gamma^{\sgn\{q_l(A, X) - Y\}} \\
& s_u(Z; q_u) = q_u(A, X) + \{Y - q_u(A, X)\}\gamma^{\sgn\{Y - q_u(A, X)\}} \\
& \kappa_j(A, X; q_j) = \E\{s_j(Z; q_j)\mid A, X\}
\end{align*}
Following the sample splitting scheme whereby all nuisance functions are estimated on a separate, independent sample $D^n$, a DR-Learner estimator of $\theta_j(t; \gamma)$ regresses an estimate of $\varphi_j(Z; \gamma)$ onto $A$ on the test set. For example, if the second stage regression is done via linear smoothing, then $\widehat\theta_j(t; \gamma) = n^{-1}\sum_{i = 1}^n W_i(t; A_i) \widehat\varphi_j(Z; \gamma)$. It can be shown that $\varphi_j(Z)$ is just part of the influence function of $\int \theta_j(a; \gamma) d\Pb(a)$, which is a pathwise-differentiable parameter. Furthermore, $\varphi_l(Z; 1) = \varphi_u(Z; 1) = \varphi(Z)$. 

The error analysis of the DR-Learners $\widehat\theta_l(t; \gamma)$ and $\widehat\theta_u(t; \gamma)$ follows from Propositions 2.1 and 2.2. In this light, it only remains to calculate $\E\{\widehat\varphi_j(Z; \gamma) - \varphi_j(Z; \gamma) \mid A = t, D^n\}$. We do so in the following lemma, proved in Section \ref{app:sens_bias} below, which plays the role of Lemma 2.3 in the no-unmeasured-confounding case. 
\begin{lemma}\label{lemma:sens_bias}
Let $\widehat{r}_j(t) = \E\{\widehat\varphi_j(Z; \gamma) - \varphi_j(Z; \gamma) \mid A = t, D^n\}$. It holds that
\begin{align*}
	|\widehat{r}_j(t)| \ \lesssim \|w - \widehat{w}\|_t\|\kappa_j - \widehat\kappa_j\|_t + \|q_j - \widehat{q}_j\|^2_t + \left|(\Pn - \Pb) \widehat\kappa_j(t, X; \widehat{q}_j)  \right|
\end{align*}
\end{lemma}
The result of Lemma \ref{lemma:sens_bias} is similar to that of Lemma 2.3, so that the discussion regarding estimating $\omega$ and, in this case, the regression function $\kappa$ applies here as well. There are two main differences. First, the upper bound on the conditional bias now involves the additional term $\|q_j - \widehat{q}_j\|_t^2$, so that consistent estimation of the bounds relies on the consistency of the conditional quantiles estimators. Second, the error $\|\kappa_j - \widehat\kappa_j\|$ is nonstandard since it pertains to a regression function with estimated outcome ($s_j(Z; \widehat{q}_j))$. The problem of estimating conditional quantiles is well-studied; see, e.g., \cite{belloni2017high, athey2019generalized, koenker1978regression, white1992nonparametric}. Because the function $q \mapsto s_j(Z; q)$ is Lipschitz, one may expect the term $\|\kappa_j - \widehat\kappa_j\|_t$ to be of order the maximum between $\|\widehat{q}_j - q_j\|_t$ and $\|\widetilde\kappa_j - \kappa_j\|_t$, where $\widetilde\kappa_j = \widehat\E_n(s_j(Z; q_j) \mid A, X)$ is the oracle estimator regressing the true outcome $s_j(Z; q)$ onto $(A, X)$ using an estimator $\widehat\E_n(\cdot \mid A, X)$. This would be the case if the second-stage regression estimator $\widehat\E_n(\cdot \mid A, X)$ posseses some stability in the form of 
\begin{align*}
\left|\widehat\E_n\{s_j(Z; \widehat{q}_j) \mid A, X\} -  \widehat\E_n\{s_j(Z; q_j) \mid A, X\}\right| \lesssim \left|\widehat{q}_j(A, X) - q_j(A, X)\right|
\end{align*}
We refer to Section B in the Appendix of \cite{dorn2024doubly} for a more comprehensive discussion on the propogation of the error in estimating the conditional quantile in the context of binary treatments.
The centered empirical average term is of order $O_\Pb(n^{-1/2})$, under mild boundedness conditions, and thus negligible in nonparametric models for which the convergence rate is slower than $n^{-1/2}$.

We conclude by establishing that $\varphi_j(Z; \gamma)$ satisfies the \textit{doubly-valid} structure discovered by \cite{dorn2021doubly} in a similar sensitivity model for binary treatments. In particular, the bounds remain valid even if the conditional quantiles are not correctly specified. While \cite{dorn2021doubly} focused on binary treatments, their observation extends to the continuous treatment case as well, as summarized in the following proposition.
\begin{proposition}\label{prop:doubly_valid}
Let $\overline{w}$, $\overline\kappa_l$, $\overline\kappa_u$, $\overline{q}_l$ and $\overline{q}_u$ be some fixed-functions such that all the expectations below are well defined. If either $\overline\kappa_j = \kappa_j(a, x; \overline{q}_j)$ or $\overline{w}(a, x) = w(a, x)$, but not necessarily both, then 
\begin{gather*}
	\E\{\varphi_l(Z; \overline{w}, \overline\kappa_l, \overline{q}_l, \gamma) \mid A = t\} \\
	\leq \theta_l(t; \gamma) \leq \theta_u(t; \gamma) \leq \\
	\E\{\varphi_u(Z; \overline{w}, \overline\kappa_u, \overline{q}_u, \gamma) \mid A = t\}
\end{gather*}
\end{proposition}
\begin{proof}
If either $\overline\kappa_j = \kappa_j(a, x; \overline{q}_j)$ or $\overline{w}(a, x) = w(a, x)$, then 
\begin{align*}
	& \E\{\varphi_l(Z; \overline{w}, \overline\kappa_l, \overline{q}_l, \gamma) \mid A = t\} = \\
	& = \int \left(\overline{q}_l(t, x) + \E[\{Y - \overline{q}_l(A, X)\}\gamma^{\sgn\{\overline{q}_l(A, X) - Y\}} \mid A = t, X = x] \right) d\Pb(x) \\
	& \E\{\varphi_u(Z; \overline{w}, \overline\kappa_u, \overline{q}_u, \gamma) \mid A = t\} = \\
	& = \int \left(\overline{q}_u(t, x) + \E[\{Y - \overline{q}_u(A, X)\}\gamma^{\sgn\{Y - \overline{q}_u(A, X)\}} \mid A = t, X = x] \right) d\Pb(x) 
\end{align*} 
The result follows because it holds that $$\E\left[\gamma^{\sgn\{q_l(A, X) - Y\} }\mid A, X\right] = \E\left[\gamma^{\sgn\{Y - q_u(A, X)\} }\mid A, X\right] = 1$$
and, deterministically, that 
\begin{align*}
	& \{Y - \overline{q}_l(A, X)\}\gamma^{\sgn\{\overline{q}_l(A, X) - Y\}} \leq \{Y - \overline{q}_l(A, X)\}\gamma^{\sgn\{q_l(A, X) - Y\}} \\
	& \{Y - \overline{q}_u(A, X)\}\gamma^{\sgn\{Y - \overline{q}_u(A, X)\}} \geq \{Y - \overline{q}_u(A, X)\}\gamma^{\sgn\{Y - q_u(A, X)\}}
\end{align*}
\end{proof}
Proposition \ref{prop:doubly_valid} establishes the doubly-valid structure of $\varphi_l(Z; \gamma)$ and $\varphi_u(Z; \gamma)$. Just like in the sensitivity model studied by \cite{dorn2021doubly} for binary treatments, the bounds on $\E(Y^a)$ remain valid even if the conditional quantiles are not correctly specified as long as either $w(a, x)$ or the second stage regression of $s_j(Z; \overline{q})$ onto $(A, X)$ are.

In the next proposition, we provide the sample analog of Proposition \ref{prop:doubly_valid} when the estimator of the bounds is a DR-Learner based on localized linear smoothing. Let $\overline\kappa_j(a,x) \equiv \kappa_j(a, x; \overline{q}_j) = \E\{s_j(Z; \overline{q}_j) \mid A = a, X = x\}$. Further, let $R^2_j(t)$ be the mean-square-error of an oracle estimator of $\theta_j(t; \gamma)$ regressing the pseudo-outcome $\overline\varphi_j(Z; w, \overline\kappa_j, \overline{q}_j, q_j)$ onto $A$, defined as
\begin{align*}
& \overline\varphi_u(Z; w, \overline\kappa_u, \overline{q}_u, q_u) = \\
& = w(A, X)\{s_u(Z; \overline{q}_u) - \overline\kappa_u(A, X; \overline{q}_u)\} + \int \overline\kappa_u(A, x; \overline{q}_u) d\Pb(x) \\
& \hphantom{=} - w(A, X)\{Y - \overline{q}_u(A, X)\}\left[\gamma^{\sgn\{Y - \overline{q}_u(A, X)\}} - \gamma^{\sgn\{Y - q_u(A, X)\}} \right] \\
& \overline\varphi_l(Z; w, \overline\kappa_l, \overline{q}_l, q_l) = \\
& = w(A, X)\{s_l(Z; \overline{q}_l) - \overline\kappa_l(A, X; \overline{q}_l)\} + \int \overline\kappa_l(A, x; \overline{q}_l) d\Pb(x) \\
& \hphantom{=} - w(A, X)\{Y - \overline{q}_l(A, X)\}\left[\gamma^{\sgn\{\overline{q}_l(A, X) - Y\}} - \gamma^{\sgn\{q_l(A, X) - Y\}} \right]
\end{align*}
It can be shown that $\E\{\overline\varphi_j(Z; w, \overline\kappa_j, \overline{q}_j, q_j) \mid A = t\} = \theta_j(t; \gamma)$ for $j = \{l, u\}$.
\begin{proposition}\label{prop:sample_doubly_valid}
Let $\widehat\theta_j(t; \gamma)$ be an DR-Learner estimator of $\theta_j(t; \gamma)$ based on linear smoothing (Sections 2 in the main text and \ref{section:sens} in the supplement). Further, let the sample splitting scheme be the same as in Algorithm 1 and assume that the following conditions hold:
\begin{enumerate}
	\item If $T_i \leq V_i$ for all $i \in \{1, \ldots, n\}$, then the weights satisfy $\sum_{i = 1}^nW_i(t; A^n) T_i \leq \sum_{i = 1}^n W_i(t; A^n) V_i$;
	\item $\|\widehat{w} - w\|_\infty$, $\|\widehat\kappa_j - \overline\kappa_j\|_\infty$ and $\|\widehat{q}_j - \overline{q}_j\|_\infty$ are all $o_\Pb(1)$, where $\overline{q}_j(a, x)$ does not need to equal $q_j(a, x)$;
	\item $\var\{\overline\varphi_j(Z; w, \overline\kappa_j, \overline{q}_u, w, q_u) \mid A = a\} \geq c > 0$ for all $a \in \mathcal{A}$ and some constant $c$. 
	\item The outcome $Y$ has a uniformly bounded conditional density given any values of $(A, X)$;
	\item The linear smoother weights $W_i(t; A^n)$ are localized as in Proposition 2.2 in a neighborhood $N_t$ around $A = t$.
\end{enumerate}
Then, the following inequalities hold
\begin{align*}
	& \widehat\theta_l(t; \gamma) \leq \theta_l(t; \gamma) + O_\Pb\left(R_l(t) + \sup_{a \in N_t} \widehat{r}_l(a) \right) \\
	& \widehat\theta_u(t; \gamma) \geq \theta_u(t; \gamma) + O_\Pb\left(R_u(t) + \sup_{a \in N_t} \widehat{r}_u(a) \right)
\end{align*}
where, for $\|f\|^2_t = \int f^2(z) d\Pb(z \mid A = t)$:
\begin{align*}
	\widehat{r}_j(t) = \|\widehat{w} - w\|_t\|\widehat\kappa_j - \overline\kappa_j\|_t \ + \ \|\widehat{w} - w\|_t\|\widehat{q}_j-\overline{q}_j\|_t \ + \ \left|(\Pn - \Pb)\widehat\kappa_j(t, X; \widehat{q}_j)\right|
\end{align*}
\end{proposition}
Proposition \ref{prop:sample_doubly_valid} shows that, even if the conditional quantiles of $Y$ given $(A, X)$ are not well estimated, the estimators of the bounds can still converge to functions that contain the region $[\theta_l(t; \gamma), \theta_u(t; \gamma)]$ and, in this sense, are ``valid bounds.'' The result holds under mild conditions. For instance, conditions 1 and 5 are a mild stability conditions on the second-stage linear smoother. Conditions 3 and 4 are mild regularity conditions on the data generating process and the nuisance functions' estimators. The speed at which $\widehat\theta_j(t; \gamma)$ converges to valid bounds depends on the structural properties of $\theta_j(t; \gamma)$, encoded in the oracle MSE $R^2_j(t)$, as well as the accuracy in estimating $w$ and $\overline\kappa$. In particular, the conditional bias term $\widehat{r}_j(t)$, defined in Proposition \ref{prop:sample_doubly_valid}, is similar to that defined in Lemma 5.1; it consists of a sum of product of errors plus a $O_P(n^{-1/2})$-term. The main difference is that here the quantile estimator $\widehat{q}_j$ is allowed to converge uniformly to a fixed function $\overline{q}_j(a, x)$ that does not need to equal the true $q_j(a, x)$. In this sense, depending on the quantile estimator, it may be reasonable to assume that $\|\widehat{q}_j - \overline{q}_j\|_t$ converges much faster to zero than $\|\widehat\kappa_j - \overline\kappa_j\|_t$ so that, in $\widehat{r}_j(t)$, the first term is the dominant one. Notice that $\overline\kappa_j$ depends on $\overline{q}_j(a, x)$ and not $q_j(a, x)$. The proof of Proposition \ref{prop:sample_doubly_valid} extends the strategy of \cite{dorn2021doubly} to the case of non-root-$n$ estimable parameters.
\subsection{Proof of Lemma \ref{lemma:sens_bias} \label{app:sens_bias}}
We prove the result for the upper bound, as that for the lower bound can be proven with a similar argument. The derivative of the map $\overline{q}(a, x) \mapsto \E\{s_u(Z; \overline{q}) \mid A = a, X = x)$ is
\begin{align*}
& \frac{d}{d \overline{q}} \left\{ \overline{q} + \frac{1}{\gamma} \int_{-\infty}^{\overline{q}} (y - \overline{q}) f(y \mid A = a, X = x) dy \right. \\
& \hphantom{\frac{d}{d \overline{q}} \left\{ \overline{q} \right. } \left. + \gamma \int_{\overline{q}}^{\infty} (y - \overline{q}) f(y \mid A = a, X = x) dy \right\} \\
& = 1 - \frac{1}{\gamma}\Pb\left(Y \leq \overline{q} \mid A = a, X = x\right) - \gamma\Pb\left(Y \geq \overline{q} \mid A = a, X = x\right)
\end{align*}
Similarly, the second derivative is
\begin{align*}
& \frac{d^2}{d \overline{q}^2} \left\{ \overline{q} + \frac{1}{\gamma} \int_{-\infty}^{\overline{q}} (y - \overline{q}) f(y \mid A = a, X = x) dy + \right. \\
& \hphantom{\frac{d^2}{d \overline{q}^2} \left\{ \overline{q}\right. } \left. \gamma \int_{\overline{q}}^{\infty} (y - \overline{q}) f(y \mid A = a, X = x) dy \right\} \\
& = \frac{d }{d \overline{q}} \left\{1 - \frac{1}{\gamma} \int_{-\infty}^{\overline{q}} f(y \mid A = a, X = x)dy -\gamma\int_{\overline{q}}^{\infty} f(y \mid A = a, X = x)dy\right\} \\
& = -\frac{1}{\gamma}f(\overline{q}\mid A = a, X = x) + \gamma f(\overline{q}\mid A = a, X = x) \\
& = O(1)
\end{align*}
Notice that the first derivative vanishes at the true quantile $\overline{q}(a, x) = q_u(a, x)$. Therefore, by a second order Taylor expansion, it holds that
\begin{align*}
\left| \E\{s(Z; \widehat{q}_u) - s(Z; q_u) \mid A = a, X = x\} \right| \lesssim \left\{\widehat{q}_u(a, x) - q_u(a, x) \right\}^2
\end{align*}
Next, notice that
\begin{align*}
\rhat_u(t) & = \int \widehat{w}(t, x)[\E\{ s(Z; \widehat{q}_u) \mid A = t, X = x\} - \widehat\kappa_u(t, x)] d\Pb(x \mid A = t) \\
& \hphantom{=} + \int w(t, x)\{\widehat\kappa_u(t, x) - \kappa_u(t, x) \} d\Pb(x \mid A = t) + (\Pn - \Pb) \widehat\kappa_u(t, X; \widehat{q}_u) \\
& = \int \widehat{w}(t, x)[\E\{ s(Z; \widehat{q}_u) \mid A = t, X = x\} - \kappa_u(t, x)] d\Pb(x \mid A = t) \\
& \hphantom{=} + \int \{w(t, x) - \widehat{w}(t, x)\}\{\widehat\kappa_u(t, x) - \kappa_u(t, x) \} d\Pb(x \mid A = t) \\
& \hphantom{=} + (\Pn - \Pb) \widehat\kappa_u(t, X; \widehat{q}_u)
\end{align*}
The bound then follows by the Cauchy-Schwarz inequality. 
\subsection{Proof of Proposition \ref{prop:sample_doubly_valid}}
We prove the result for the upper bound, as the proof for the lower bound is analogous. Let $\widehat\E_n(\cdot \mid A = t)$ denote the second-stage regression based on linear smoothing. Define
\begin{align*}
& \widetilde\varphi_u(Z; \widehat{w}, \widehat\kappa_u, \widehat{q}_u, w, q_u)\\
& = \varphi_u(Z; \widehat{w}, \widehat\kappa_u, \widehat{q}_u) \\
& \hphantom{=} - w(A, X)[Y - \widehat{q}_u(A, X)]\left[\gamma^{\sgn\{Y - \widehat{q}_u(A, X)\}} - \gamma^{\sgn\{Y - q_u(A, X)\}} \right],
\end{align*}
where
\begin{align*}
\varphi_u(Z; \widehat{w}, \widehat\kappa_u, \widehat{q}_u) = \widehat{w}(A, X)\{s_u(Z; \widehat{q}_u) - \widehat\kappa_u(A, X; \widehat{q}_u)\} + \frac{1}{n} \sum_{i = 1}^n \widehat\kappa_u(A, X_i; \widehat{q}_u)
\end{align*}
We have $\varphi_u(Z; \widehat{w}, \widehat\kappa_u, \widehat{q}_u) \geq \widetilde\varphi_u(Z; \widehat{w}, \widehat\kappa_u, \widehat{q}_u, w, q_u)$ and, deterministically by assumption, $$\widehat\E_n\{\varphi_u(Z; \widehat{w}, \widehat\kappa_u, \widehat{q}_u)\mid A = t\} \geq \widehat\E_n\{\widetilde\varphi_u(Z; \widehat{w}, \widehat\kappa_u, \widehat{q}_u, w, q_u) \mid A = t\}.$$ 
Let
\begin{align*}
& \overline\varphi_u(Z; w, \overline\kappa_u, \overline{q}_u, q_u) \\
& = \widetilde\varphi_u(Z; w, \overline\kappa_u, \overline{q}_u, w, q_u) - \frac{1}{n} \sum_{i = 1}^n \widehat\kappa_u(A, X_i; \widehat{q}_u) + \int \overline\kappa_u(A, x; \overline{q}_u) d\Pb(x) \\
& = w(A, X)\{s_u(Z; \overline{q}_u) - \overline\kappa(A, X; \overline{q}_u)\} + \int \overline\kappa(A, x; \overline{q}_u) d\Pb(x) \\
& \hphantom{=} - w(A, X)[Y - \overline{q}_u(A, X)]\left[\gamma^{\sgn\{Y - \overline{q}_u(A, X)\}} - \gamma^{\sgn\{Y - q_u(A, X)\}} \right]
\end{align*}
and notice that, because $\E[\gamma^{\{Y - q_u(t, x)\}} \mid A = t, X = x] = 1$:
\begin{align*}
& \E\left\{ \overline\varphi_u(Z; w, \overline\kappa_u, \overline{q}_u, q_u) \mid A = t,  X = x\right\} \\
& = \int \overline\kappa(t, x; \overline{q}_u) d\Pb(x) \\
& \hphantom{=} - w(t, x) \E\left[ \{Y - \overline{q}_u(t, x)\}\gamma^{\sgn\{Y - \overline{q}_u(t,x)\}} \mid A = t, X = x\right] \\
& \hphantom{=} + w(t, x)\E\left[Y\gamma^{\sgn\{Y - q_u(t, x)\}} \mid A = t, X = x\right] - w(t, x)\overline{q}_u(t, x) \\
& = \int \overline\kappa(t, x; \overline{q}_u) d\Pb(x) - w(t, x) \overline\kappa(t, x; \overline{q}_u) \\
& \hphantom{=} + w(t, x)\E\left[Y\gamma^{\sgn\{Y - q_u(t, x)\}} \mid A = t, X = x\right]
\end{align*}
Therefore, $\E\left\{ \overline\varphi_u(Z; w, \overline\kappa_u, \overline{q}_u, q_u) \mid A = t\right\} = \theta_u(t; \gamma)$. 

By the reasoning in \cite{kennedy2020optimal} and used to prove Proposition 2.2, one has
\begin{align*}
& \widehat\E_n\{\widetilde\varphi_u(Z; \widehat{w}, \widehat\kappa_u, \widehat{q}_u, q_u) \mid A = t\} - \theta_u(t; \gamma) \\
& = \widehat\E_n\{\widetilde\varphi_u(Z; \widehat{w}, \widehat\kappa_u, \widehat{q}_u, q_u) \mid A = t\} - \E\{\overline\varphi_u(Z; w, \overline\kappa_u, \overline{q}_u, q_u) \mid A = t\} \\
& = O_\Pb\left(R_u(t) \right) \\
& \hphantom{=} + \frac{1}{n}\sum_{i = 1}^nW_i(t; A^n)\E\left\{\widetilde\varphi_u(Z; \widehat{w}, \widehat\kappa_u, \widehat{q}_u, w, q_u) - \overline\varphi_u(Z; w, \overline\kappa_u, \overline{q}_u, q_u) \mid A_i, D^n\right\}
\end{align*}
provided that $\sup_z| \widetilde\varphi_u(z; \widehat{w}, \widehat\kappa_u, \widehat{q}_u, q_u) - \overline\varphi_u(z; w, \overline\kappa_u, \overline{q}_u, q_u)| = o_\Pb(1)$. This is the case, because $\widehat{w}$ is consistent for $w$, $\widehat\kappa_u$ is consistent for $\overline\kappa_u$ and $\widehat{q}$ is consistent for $\overline{q}$.

Next, recall that 
\begin{align*}
\theta_u(A; \gamma) & = \E\{\overline\varphi_u(Z; w, \overline\kappa_u, \overline{q}_u, q_u) \mid A\} \\
& = \int W(A, x)\E\{Y \gamma^{\sgn\{Y - q_u(A, X)\}} \mid A, X = x\} d\Pb(x \mid A)
\end{align*} 
so that, because $\E[\gamma^{\sgn\{Y - q_u(A, X)\}} \mid A, X] = 1$:
\begin{align*}
& \E\left(w(A, X)[Y - \widehat{q}_u(A, X)]\gamma^{\sgn\{Y - q_u(A, X)\}} - \overline\varphi_u(Z; w, \overline\kappa_u, \overline{q}_u, q_u) \mid A_i, D^n \right) \\
& = -\int w(A_i, x) \widehat{q}_u(A_i, x) d\Pb(x \mid A_i)
\end{align*}
In turns, this means that
\begin{align*}
& \E\left(-w(A, X)\{Y - \widehat{q}_u(A, X)\}\left[\gamma^{\sgn\{Y - \widehat{q}_u(A, X)\}} - \gamma^{\sgn\{Y - q_u(A, X)\}} \right] \right. \\
& \hphantom{\E\left(-w(A, X)\right.} \left. - \overline\varphi_u(Z; w, \overline\kappa_u, \overline{q}_u, q_u) \mid A_i, D^n\right) \\
& = -\E\{w(A_i, X)s_u(Z; \widehat{q}_u) \mid A_i, D^n)
\end{align*}
yielding
\begin{align*}
& \widehat{b}(A_i) \\
& \equiv \E\{\widetilde\varphi_u(Z; \widehat{w}, \widehat\kappa_u, \widehat{q}_u, w, q_u) - \overline\varphi_u(Z; w, \overline\kappa_u, \overline{q}_u, q_u) \mid A_i, D^n\} \\
& = \int \{\widehat{w}(A_i, x) - w(A_i, x)\} [ \E\{s_u(Z; \widehat{q}_u) \mid A_i, x\} - \widehat\kappa_u(A_i, x; \widehat{q}_u)] d\Pb(x \mid A_i) \\
& \hphantom{=} + (\Pn - \Pb) \widehat\kappa_u(A_i, X; \widehat{q}_u)\\
& = \int \{\widehat{w}(A_i, x) - w(A_i, x)\} [ \E\{s_u(Z; \widehat{q}_u) \mid A_i, x\} - \overline\kappa_u(A_i, x; \overline{q}_u)] d\Pb(x \mid A_i) \\
& \hphantom{=} + \int \{\widehat{w}(A_i, x) - w(A_i, x)\} \{\overline\kappa_u(A_i, x; \overline{q}_u) - \widehat\kappa_u(A_i, x; \widehat{q}_u)\} d\Pb(x \mid A_i) \\
& \hphantom{=} + (\Pn - \Pb) \widehat\kappa_u(A_i, X; \widehat{q}_u)
\end{align*}
As shown in \cite{dorn2021doubly} (Lemma 5), the map $q \mapsto s_u(Z; q)$ is Lipschitz. Therefore, by Cauchy-Schwarz:
\begin{align*}
\left| \widehat{\E}_n\{\widehat{b}(A) \mid A = t\} \right | & \lesssim \sup_{a \in N_t} \left[\|\widehat{w}-w\|_a \{ \|\widehat{q}_u - \overline{q}_u\|_a + \|\widehat\kappa_u - \overline\kappa_u\|_a\} \right. \\
& \hphantom{\lesssim\sup_{a \in N_t} \left[\right. } + \left. \left|(\Pn - \Pb) \widehat\kappa_u(a, X; \widehat{q}_u)\right|\right]
\end{align*}
\section{Additional simulation results}
\subsection{Challenging setup for second-order corrections}
In this section, we investigate a setup whereby pursuing higher-order corrections can potentially deteriorate the estimator's performance. The setup is the same as in Section 4 in the main text except that: 1) the sample size is reduced to $n=1000$, and 2) the function $\mu_{ax}(a, x)$ is replaced by
\begin{align*}
& \mu_{ax}(a, x) = 1.2A + A^2 + \Delta \cdot A[s(X; \delta_1) - \E\{s(X; \delta_1)\}] \\
& \hphantom{\mu_{ax}(a, x) =} + [s(X; \delta_2) - \E\{s(X; \delta_2)\}], \\
& s(x; \delta) = \sum_{j=1}^d 1.5\cdot\one(x_j > \delta_j),
\end{align*}
with $\delta_j$ sampled uniformly in $[-1, 1]$ and $\Delta \in \{0, 3\}$. Estimation of $\mu_{ax}$ and $\pi(a \mid x)$ follows the procedure described in the main text for the Random Forest case. In implementing second-order corrections, \textcolor{black}{we rely on the pseudo-inverse to estimate $\Omega^{-1}$ when the empirical Gram matrix is non-invertible.}
Ideally, in this non-invertibility regime, one should use other estimators of $\Omega$, such as those based on $\int b(x)b(x)^T\widehat{g}(x)dx$, where $\widehat{g}(x) = \int K_{ht}(a) \widehat{p}(a, x) da$ or $\widehat{g}(x) = \widehat{p}(t, x)$, for an estimator $\widehat{p}(a, x)$ of the joint density of $(A, X)$. However, even in moderate dimensions, estimating a joint density nonparametrically, as well as carrying out multidimensional integration, incurs in substantial computational costs; we leave this to future work. Figure \ref{fig:mse_indicators} reports the root-mean-square-errors (RMSEs) for the estimators considered. We cap the RMSEs at 2 to improve the visualization. It can be seen that the first-order methods perform well in the setting when the interaction $A \cdot X$ is not present ($\Delta = 0$). For $k = 16$ or $k = 32$ (corresponding to 321 or 641 basis terms), pursuing higher-order corrections leads to a substantially deteriorated, unstable performance. For $\delta =3$, consistent with the results from Section 4 in the main text based on Random Forest for the nuisances, the first-order estimators suffer from larger errors in regions of $A$ far away from zero. However, while higher-order corrections based on smaller $k$ are still able to correct a portion of the first-order estimation error even at the current sample size $n=1000$, those based on larger $k$ become numerically unstable and lead to a deteriorated performance. \textcolor{black}{This simulation setting is particularly challenging for the higher-order estimators because the sample size is small relative to the number of basis terms $k$ included in the corrections. In fact, studying reliable, data-driven choices for $k$ is an important avenue for future work. Furthermore, the true function $x \mapsto \mu(a, x)$ is nonsmooth so that approximating it with B-splines can be less effective. However, the finding that, for low values of $k$, higher-order corrections are still able to deliver some improvements relative to first-order estimators suggest this second aspect is less relevant in this setting.}
\begin{figure}[!h]
\centering
\includegraphics[scale=0.4]{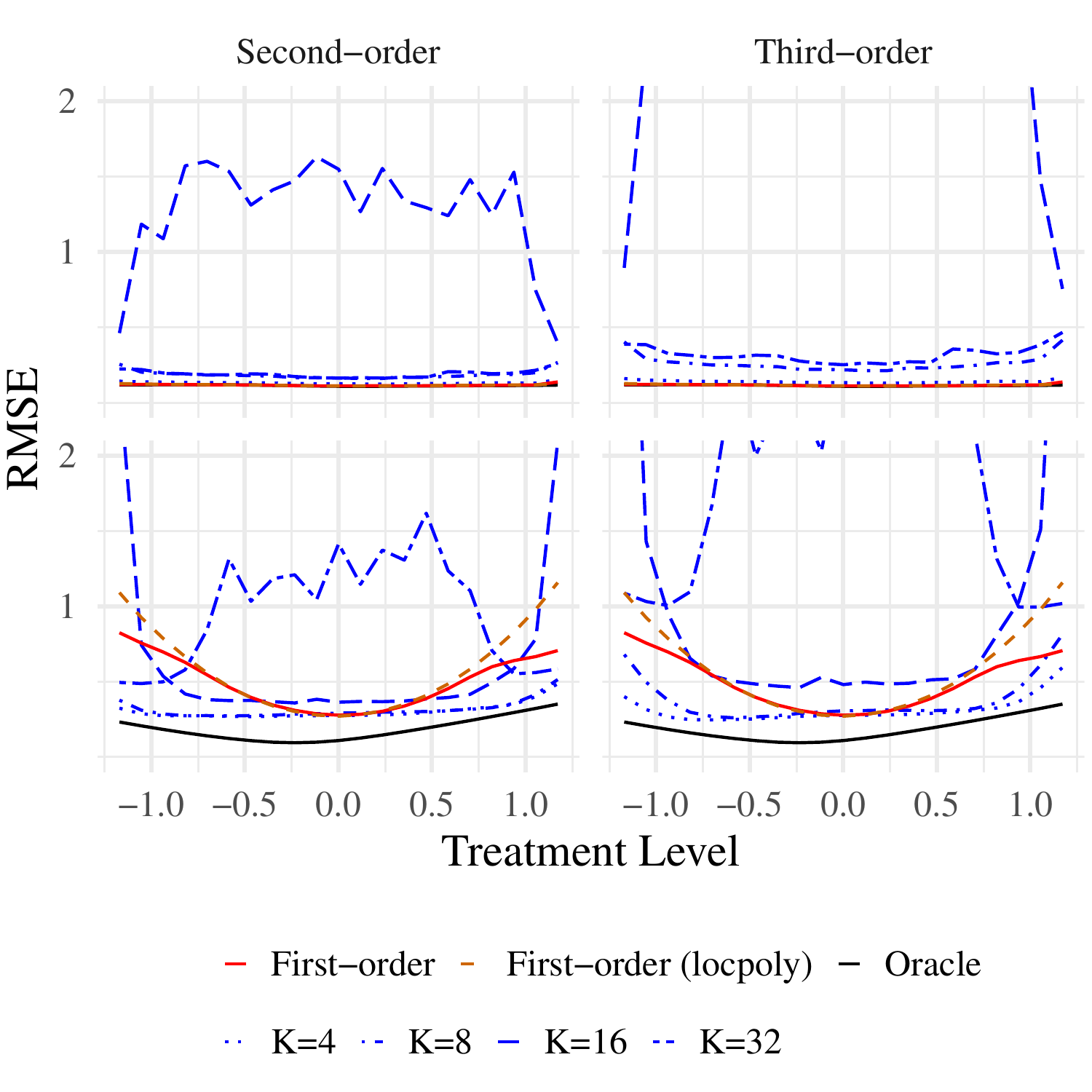}
\caption{RMSE as a function of the evaluation points of the dose--response function in three settings. 
\textit{First row}: no interaction terms between $A$ and $X$ in $\mu(a, x)$; \textit{second row}: $\mu(a, x)$ contains the interaction terms between $A$ and $X$.
\textit{First-order (locpoly)} denotes the first-order estimator obtained by regressing $\widehat{\varphi}(Z)$ on $A$ using local linear smoothing, whereas \textit{First-order} denotes $\widehat{\theta}_1(t)=\Pn \widehat{f}_0$ from Section~3.3. $k$ refers to the basis dimension in the higher-order corrections. \label{fig:mse_indicators}}
\end{figure}
\subsection{Different choice for $\widehat\Omega$}
In this section, we conduct a cursory exploration of the performance of the second-order estimator considered in Section 4 in the main text when the empirical Gram matrix estimating $\Omega$ is replaced by $\int b(x)b(x)^T \widehat{p}(t, x) dx$, where $\widehat\pi(t, x) = \widehat\pi(t \mid x) \widehat{p}(x)$ with $\widehat\pi(t \mid x)$ estimated as described in Section 4 (misspecified) and $\widehat{p}(x)$ estimated by kernel density estimation using the function \texttt{density} from the package \texttt{stats}. This estimator of $\Omega$ is essentially that based on $\int K_{ht}(a, t) \widehat{p}(a, x) da$ up to a term of order $h^\beta$. Because of computational costs, we consider the case when $X$ is one-dimensional. We focus on the setting where both $\mu_x$ and $\mu_{ax}$ are estimated by B-spline regression where the model for the latter does not include the $A \cdot X$ interaction term. In this respect, the nuisance models (including $\widehat{g}$) are misspecified. We note that, because $X$ is one-dimensional, such misspecification should be easily detected by standard model checking procedures; in this sense, the results below should be interpreted as overly pessimistic regarding first-order methods in this setting. Nevertheless, as shown in Figure \ref{fig:mse_ghat}, it is encouraging to note that the second-order corrections are still able to reduce the first-order estimation error. In particular, this is the case even if $\widehat{g}$ is not entirely correctly specified.
\begin{figure}[!h]
\centering
\includegraphics[scale=0.4]{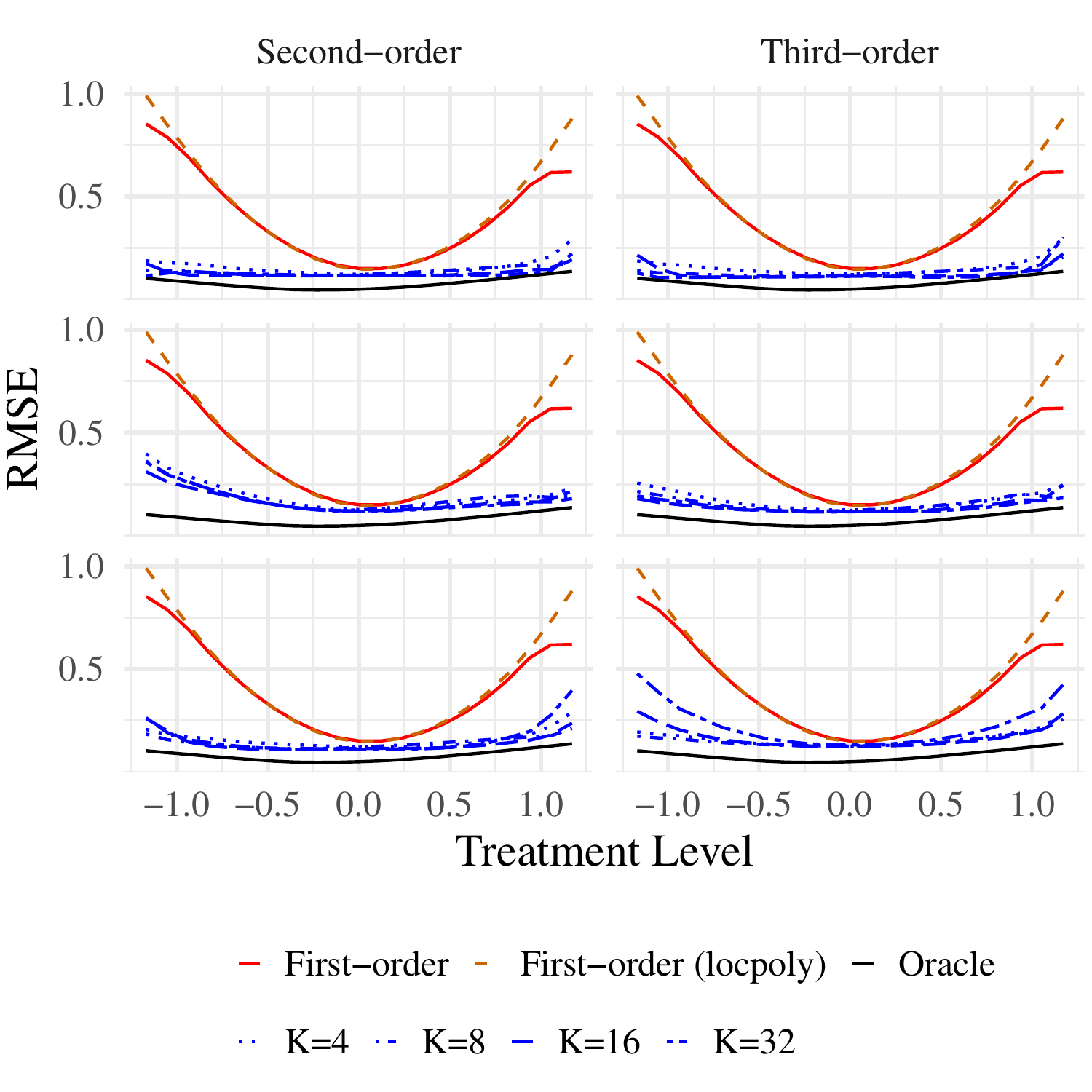}
\caption{RMSE as a function of the evaluation points of the dose--response function in three settings. 
\textit{Top row}: $\Omega$ is estimated by the empirical counterpart; \textit{middle row}: $\Omega$ is estimated based on the estimated weight $\widehat{g}(x)$; \textit{bottom row}: $\Omega$ is computed using the true weight $g(x)$. 
\textit{First-order (locpoly)} denotes the first-order estimator obtained by regressing $\widehat{\varphi}(Z)$ on $A$ using local linear smoothing, whereas \textit{First-order} denotes $\widehat{\theta}_1(t)=\Pn \widehat{f}_0$ from Section~3.3. $k$ refers to the basis dimension in the higher-order corrections. \label{fig:mse_ghat}}
\end{figure}
\subsection{Sensitivity model from Section \ref{section:sens}}
In this section, we explore the performance of a DR-learner based on local-linear second-stage regression as an estimator of the bounds on the dose-response function derived in Section \ref{section:sens} under the proposed sensitivity model. We let the set up as in Section 4 from the main text and consider the case $\gamma = 1.2$. The conditional quantiles are estimated as the quantile of a normal random variable with mean $\widehat\mu(a, x)$ and standard deviation 1. We follow the exact estimation procedure as in the main text, but, for simplicity, consider two choices for $\mu_{ax}$: 1) based on B-spline regression with correctly specified $A\cdot X$ interaction term, and 2) based on a Random Forest estimated on an independent sample. We consider the exact same models for $\kappa_j$ as well. As mentioned in the main text, the estimator of $\pi(a \mid x)$ is not correctly specified. Figure \ref{fig:mse_sens} reports the boxplots showing the estimated bounds across 1000 simulations, along with the dose-response function and the true bound. It can be seen that when the quantiles are based on $\mu_{ax}$ estimated by $B$-spline regression with the interaction term, and $\kappa_j$ is also based on the same $B$-spline model, the bounds are estimated well. When the quantiles are based on $\mu_{ax}$ estimated with Random Forest while $\kappa_j$ is still estimated by B-spline regression, the bounds remain valid but much conservative. Conversely, when both $\mu_{ax}$ and $\kappa_j$ are estimated with Random Forest, the estimation of the bounds is rather poor and, in particular, anti-conservative.
\begin{figure}[!h]
\centering
\begin{subfigure}[t]{0.32\textwidth}
\includegraphics[width=\textwidth]{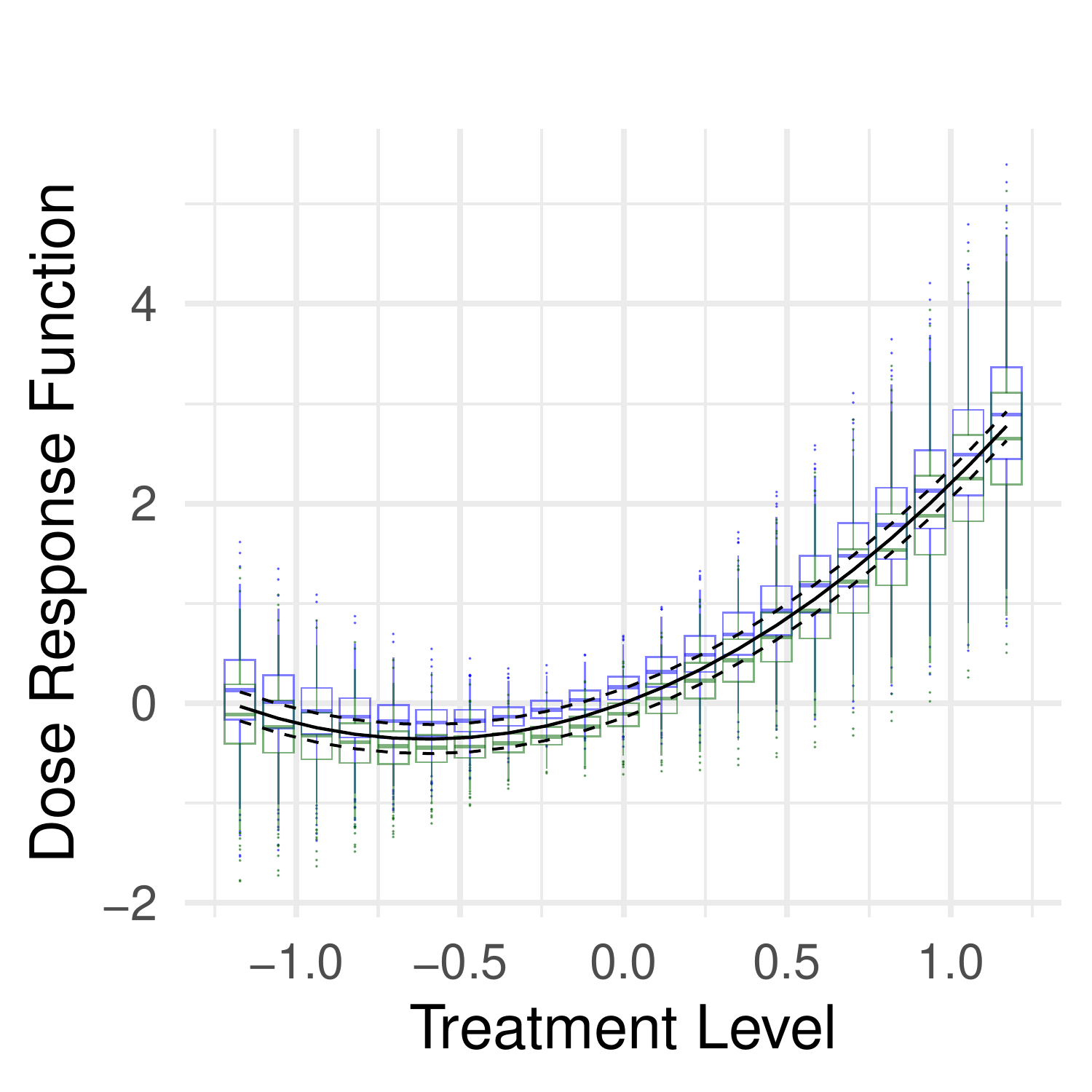}
\caption{$\widehat\mu_{ax}$ and $\widehat\kappa_j$ based on $B$-spline regression with interaction term $A \cdot X$.}
\end{subfigure}
\hfill
\begin{subfigure}[t]{0.32\textwidth}
\includegraphics[width=\textwidth]{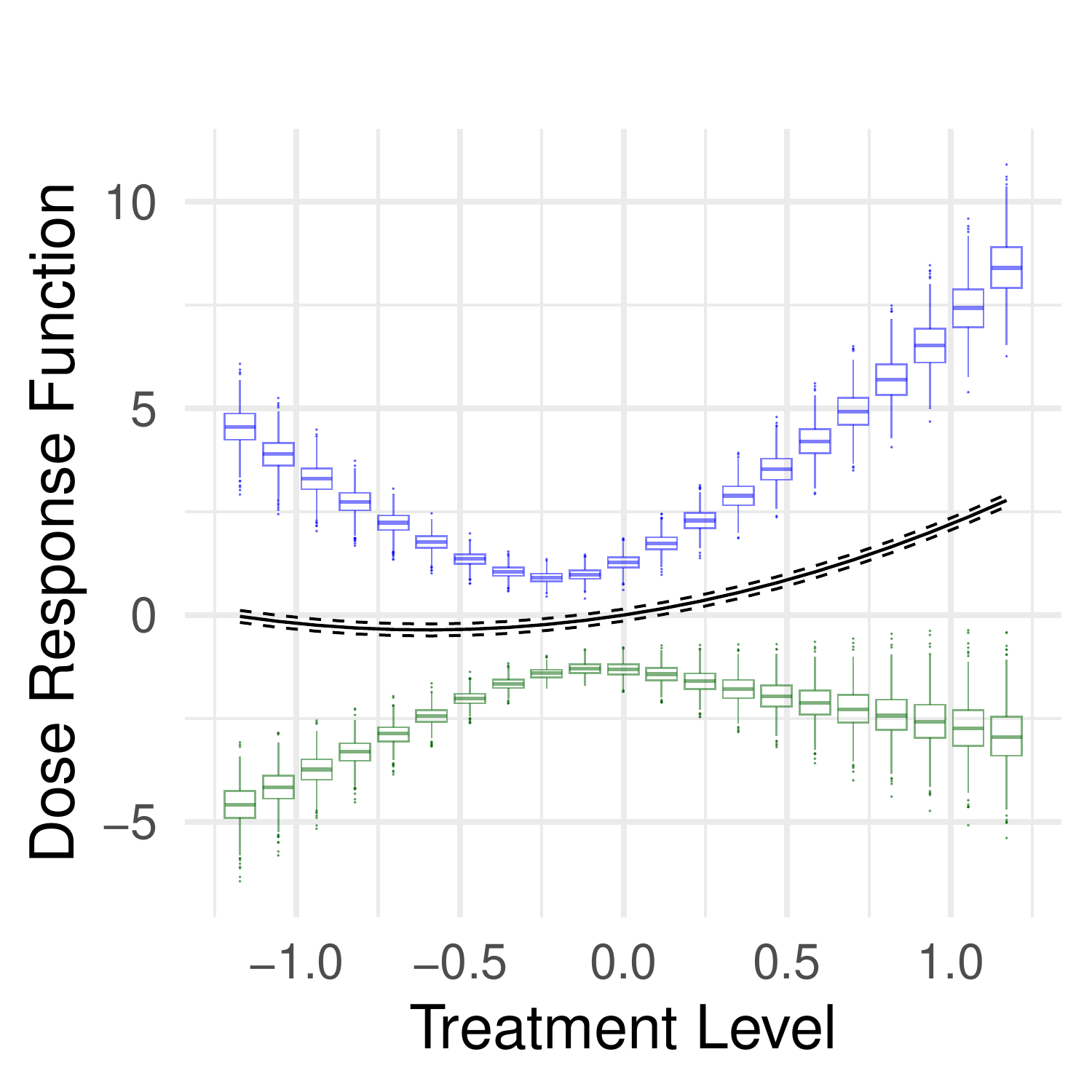}
\caption{$\widehat\mu_{ax}$ based on Random Forest and $\widehat\kappa_j$ based on $B$-spline regression with interaction term $A \cdot X$.}
\end{subfigure}
\begin{subfigure}[t]{0.32\textwidth}
\includegraphics[width=\textwidth]{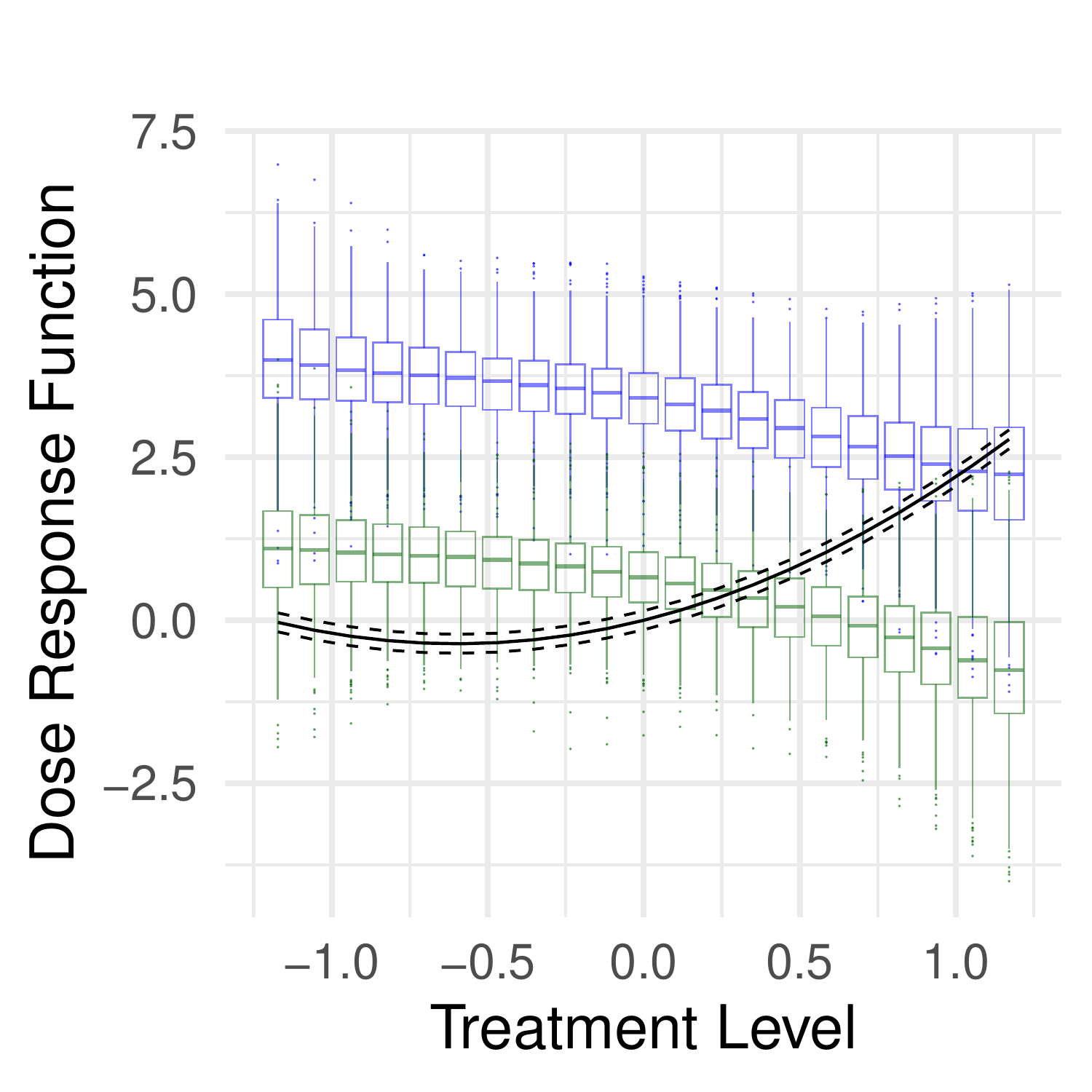}
\caption{$\widehat\mu_{ax}$ and $\widehat\kappa_j$ based on Random Forest.}
\end{subfigure}
\caption{Boxplots representing the estimated bounds across 1000 simulations, together with the true dose-response function and the true bounds. $\gamma = 1.2$. The conditional quantile is equal to that of a normal random variable with mean $\widehat\mu_{ax}$ and standard deviation equal to 1.}
\label{fig:mse_sens}
\end{figure}
\section{Application of the sensitivity analysis model to the data analysis}
In this section, we apply our proposed sensitivity model to the data analysis described in Section 5 of the main text. The standard error from the main analysis is already wide enough so that a pointwise 95\% confidence band includes the constant dose-response function; thus a sensitivity analysis in this case is less relevant. Nevertheless, we pursue it here for illustrative purposes. We set $\gamma = 1.2$ for simplicity, while future work is needed to properly calibrate the sensitivity analysis so that a range of plausible values for $\gamma$ could be better specified. We estimate the quantiles by Random Forest using the \texttt{quantile\_forest} function from the \texttt{grf} package in \texttt{R} \citep{tibshirani2018package}. The second-stage regression of $s_j(Z; \widehat{q}_j)$ onto $(A, X)$ is done by Random Forest too (\texttt{ranger} package, \cite{wright2019package}). We employ five-fold cross-fitting and select the tuning parameters by five-fold cross-validation. After constructing the updated pseudo-outcomes corresponding to the upper and lower bounds, we regress them on $A$ using local linear smoothing via the \texttt{lprobust} function with default parameters \citep{calonico2019nprobust}. Figure \ref{fig:sens_da} simply adds the estimated bounds to the original plot described in the main text. As the bounds are fairly wide, we do not add $\pm 1.96$ times their standard errors to the plot. We notice that the bounds are roughly constant shifts of the estimated dose-response (under no unmeasured confounding) along the $Y$-axis. 
\begin{figure}[!h]
\centering
\includegraphics[scale=0.3]{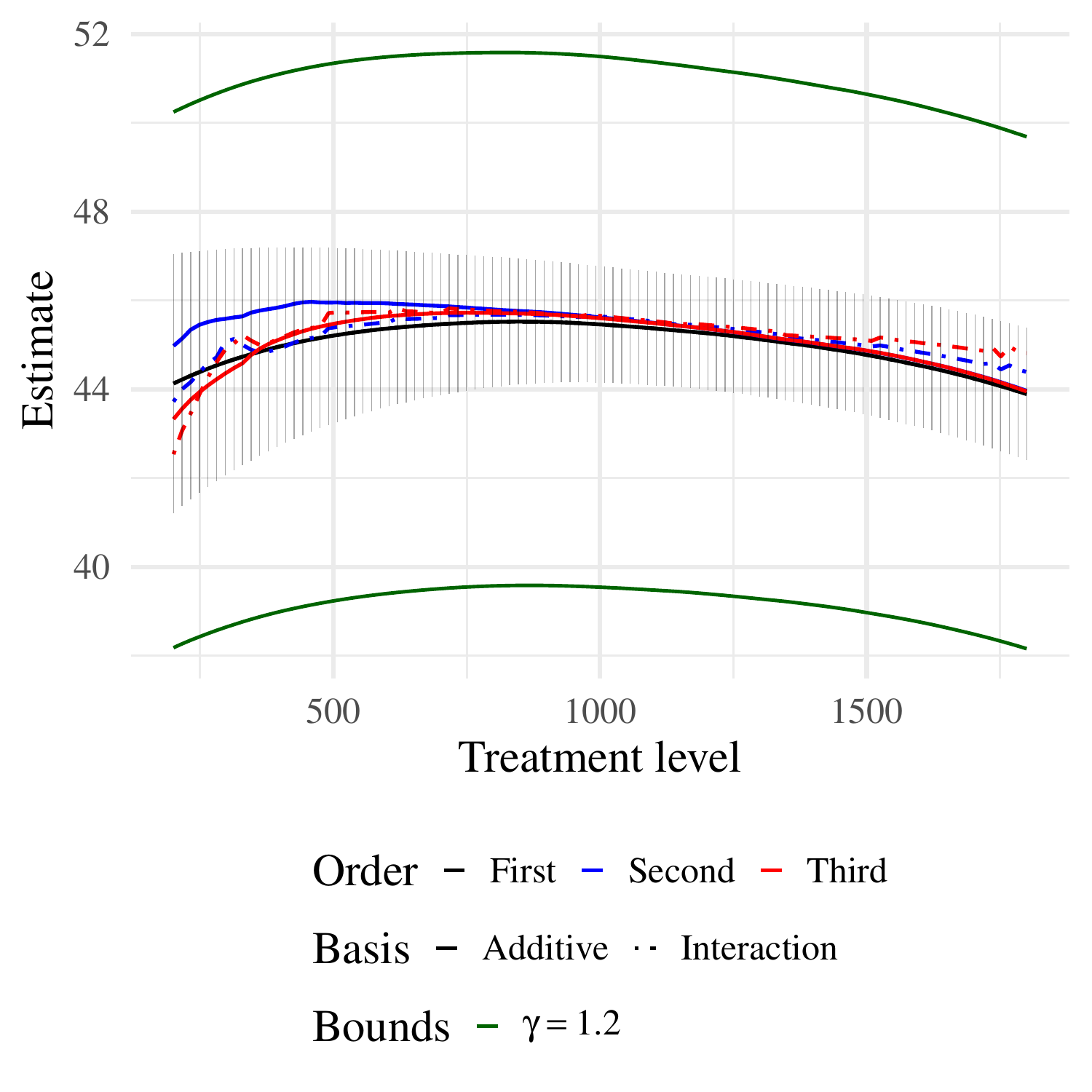}
\caption{Estimates of $\E(Y^a)$ along with the upper and lower bounds from the sensitivity model of Section \ref{section:sens} from our analysis of the Jobs Corps dataset from \cite{colangelo2020double}.\label{fig:sens_da}}
\end{figure}
\bibliographystyle{plainnat}
\bibliography{references}

@article{chen2025computing,
	title={On computing and the complexity of computing higher-order $ U $-statistics, exactly},
	author={Chen, Xingyu and Zhang, Ruiqi and Liu, Lin},
	journal={arXiv preprint arXiv:2508.12627},
	year={2025}
}

@article{koenker1978regression,
	title={Regression quantiles},
	author={Koenker, Roger and Bassett Jr, Gilbert},
	journal={Econometrica: journal of the Econometric Society},
	pages={33--50},
	year={1978},
	publisher={JSTOR}
}

@article{calonico2019nprobust,
	title={nprobust: Nonparametric kernel-based estimation and robust bias-corrected inference},
	author={Calonico, Sebastian and Cattaneo, Matias D and Farrell, Max H},
	journal={Journal of Statistical Software},
	volume={91},
	pages={1--33},
	year={2019}
}

@article{balakrishnan2023fundamental,
	title={The Fundamental Limits of Structure-Agnostic Functional Estimation},
	author={Balakrishnan, Sivaraman and Kennedy, Edward H. and Wasserman, Larry},
	journal={arXiv preprint arXiv:2305.04116},
	year={2023}
}

@article{cinelli2025challenges,
	title={Challenges in statistics: A dozen challenges in causality and causal inference},
	author={Cinelli, Carlos and Feller, Avi and Imbens, Guido and Kennedy, Edward and Magliacane, Sara and Zubizarreta, Jose},
	journal={arXiv preprint arXiv:2508.17099},
	year={2025}
}

@article{tibshirani2018package,
	title={Package ‘grf’},
	author={Tibshirani, Julie and Athey, Susan and Friedberg, Rina and Hadad, Vitor and Hirshberg, David and Miner, Luke and Sverdrup, Erik and Wager, Stefan and Wright, Marvin and Tibshirani, Maintainer Julie},
	journal={Comprehensive R Archive Network},
	year={2018}
}

@article{duong2024package,
	title={Package ‘ks’},
	author={Duong, Tarn and Duong, Maintainer Tarn and Suggests, MASS},
	journal={Website: https://cran. uvigo. es/web/packages/ks/ks. pdf},
	year={2024}
}

@article{wright2019package,
	title={Package ‘ranger’},
	author={Wright, Marvin N and Wager, Stefan and Probst, Philipp and Wright, Maintainer Marvin N},
	journal={Version 0.11},
	volume={2},
	year={2019}
}

@article{calonico2018effect,
	title={On the effect of bias estimation on coverage accuracy in nonparametric inference},
	author={Calonico, Sebastian and Cattaneo, Matias D and Farrell, Max H},
	journal={Journal of the American Statistical Association},
	volume={113},
	number={522},
	pages={767--779},
	year={2018},
	publisher={Taylor \& Francis}
}

@article{takatsu2024debiased,
	title={Debiased inference for a covariate-adjusted regression function},
	author={Takatsu, Kenta and Westling, Ted},
	journal={Journal of the Royal Statistical Society Series B: Statistical Methodology},
	year={2024},
	publisher={Oxford University Press UK}
}

@article{liu2023new,
	title={New root-n-consistent, numerically stable higher-order influence function estimators},
	author={Liu, Lin and Li, Chang},
	journal={arXiv preprint arXiv:2302.08097},
	year={2023}
}

@article{liu2024assumption,
	title={Assumption-lean falsification tests of rate double-robustness of double-machine-learning estimators},
	author={Liu, Lin and Mukherjee, Rajarshi and Robins, James M},
	journal={Journal of Econometrics},
	volume={240},
	number={2},
	pages={105500},
	year={2024},
	publisher={Elsevier}
}

@incollection{white1992nonparametric,
	title={Nonparametric estimation of conditional quantiles using neural networks},
	author={White, Halbert},
	booktitle={Computing Science and Statistics: Statistics of Many Parameters: Curves, Images, Spatial Models},
	pages={190--199},
	year={1992},
	publisher={Springer}
}

@article{athey2019generalized,
	title={Generalized Random Forests},
	author={Athey, Susan and Tibshirani, Julie and Wager, Stefan},
	journal={The Annals of Statistics},
	volume={47},
	number={2},
	pages={1148--1178},
	year={2019},
	publisher={JSTOR}
}

@incollection{belloni2017high,
	title={High-dimensional quantile regression},
	author={Belloni, Alexandre and Chernozhukov, Victor and Kato, Kengo},
	booktitle={Handbook of quantile regression},
	pages={253--272},
	year={2017},
	publisher={Chapman and Hall/CRC}
}

@article{liu2021adaptive,
	title={Adaptive estimation of nonparametric functionals},
	author={Liu, Lin and Mukherjee, Rajarshi and Robins, James M and Tchetgen, Eric Tchetgen},
	journal={Journal of Machine Learning Research},
	volume={22},
	number={99},
	pages={1--66},
	year={2021}
}

@incollection{gine2000exponential,
	title={Exponential and moment inequalities for U-statistics},
	author={Gin{\'e}, Evarist and Lata{\l}a, Rafa{\l} and Zinn, Joel},
	booktitle={High Dimensional Probability II},
	pages={13--38},
	year={2000},
	publisher={Springer}
}

@article{mcclean2024double,
	title={Double cross-fit doubly robust estimators: Beyond series regression},
	author={McClean, Alec and Balakrishnan, Sivaraman and Kennedy, Edward H and Wasserman, Larry},
	journal={arXiv preprint arXiv:2403.15175},
	year={2024}
}

@article{dorn2024doubly,
	title={Doubly-valid/doubly-sharp sensitivity analysis for causal inference with unmeasured confounding},
	author={Dorn, Jacob and Guo, Kevin and Kallus, Nathan},
	journal={Journal of the American Statistical Association},
	pages={1--12},
	year={2024},
	publisher={Taylor \& Francis}
}

@article{newey2018cross,
	title={Cross-fitting and fast remainder rates for semiparametric estimation},
	author={Newey, Whitney K and Robins, James R},
	journal={arXiv preprint arXiv:1801.09138},
	year={2018}
}

@article{rambachan2022robust,
	title={Robust design and evaluation of predictive algorithms under unobserved confounding},
	author={Rambachan, Ashesh and Coston, Amanda and Kennedy, Edward},
	journal={arXiv preprint arXiv:2212.09844},
	year={2022}
}

@article{rubin1974estimating,
  title={Estimating causal effects of treatments in randomized and nonrandomized studies.},
  author={Rubin, Donald B},
  journal={Journal of educational Psychology},
  volume={66},
  number={5},
  pages={688},
  year={1974},
  publisher={American Psychological Association}
}

@article{robins2009semiparametric,
	title={Semiparametric minimax rates},
	author={Robins, James and Tchetgen, Eric Tchetgen and Li, Lingling and van der Vaart, Aad},
	journal={Electronic journal of statistics},
	volume={3},
	pages={1305},
	year={2009},
	publisher={NIH Public Access}
}

@article{kennedy2020optimal,
	title={Towards optimal doubly robust estimation of heterogeneous causal effects},
	author={Kennedy, Edward H},
	journal={Electronic Journal of Statistics},
	volume={17},
	number={2},
	pages={3008--3049},
	year={2023},
	publisher={The Institute of Mathematical Statistics and the Bernoulli Society}
}

@incollection{robins2008higher,
	title={Higher order influence functions and minimax estimation of nonlinear functionals},
	author={Robins, James and Li, Lingling and Tchetgen, Eric and van der Vaart, Aad and others},
	booktitle={Probability and statistics: essays in honor of David A. Freedman},
	pages={335--421},
	year={2008},
	publisher={Institute of Mathematical Statistics}
}

@book{tsybakov2008introduction,
	title={Introduction to nonparametric estimation},
	author={Tsybakov, Alexandre B},
	year={2008},
	publisher={Springer Science \& Business Media}
}

@book{van2003unified,
	title={Unified methods for censored longitudinal data and causality},
	author={Van der Laan, Mark J and Laan, MJ and Robins, James M},
	year={2003},
	publisher={Springer Science \& Business Media}
}

@article{kennedy2024minimax,
	title={Minimax rates for heterogeneous causal effect estimation},
	author={Kennedy, Edward H and Balakrishnan, Sivaraman and Robins, James M and Wasserman, Larry},
	journal={The Annals of Statistics},
	volume={52},
	number={2},
	pages={793--816},
	year={2024},
	publisher={Institute of Mathematical Statistics}
}

@article{robins2009quadratic,
	title={Quadratic semiparametric von mises calculus},
	author={Robins, James and Li, Lingling and Tchetgen, Eric and van der Vaart, Aad W},
	journal={Metrika},
	volume={69},
	number={2},
	pages={227--247},
	year={2009},
	publisher={Springer}
}

@article{mukherjee2017semiparametric,
	title={Semiparametric efficient empirical higher order influence function estimators},
	author={Liu, Lin and Mukherjee, Rajarshi and Newey, Whitney K and Robins, James M},
	journal={arXiv preprint arXiv:1705.07577},
	year={2017}
}

@article{kennedy2017nonparametric,
	title={Nonparametric methods for doubly robust estimation of continuous treatment effects},
	author={Kennedy, Edward H and Ma, Zongming and McHugh, Matthew D and Small, Dylan S},
	journal={Journal of the Royal Statistical Society. Series B, Statistical Methodology},
	volume={79},
	number={4},
	pages={1229},
	year={2017},
	publisher={NIH Public Access}
}

@article{colangelo2020double,
	title={Double debiased machine learning nonparametric inference with continuous treatments},
	author={Colangelo, Kyle and Lee, Ying-Ying},
	journal={Journal of Business \& Economic Statistics},
	volume={44},
	number={1},
	pages={67--79},
	year={2026},
	publisher={Taylor \& Francis}
}

@article{semenova2017debiased,
	title={Debiased machine learning of conditional average treatment effects and other causal functions},
	author={Semenova, Vira and Chernozhukov, Victor},
	journal={The Econometrics Journal},
	volume={24},
	number={2},
	pages={264--289},
	year={2021},
	publisher={Oxford University Press}
}

@article{flores2007estimation,
	title={Estimation of dose-response functions and optimal doses with a continuous treatment},
	author={Flores, Carlos A},
	journal={University of Miami, Department of Economics, November},
	year={2007},
	publisher={Citeseer}
}

@article{galvao2015uniformly,
	title={Uniformly semiparametric efficient estimation of treatment effects with a continuous treatment},
	author={Galvao, Antonio F and Wang, Liang},
	journal={Journal of the American Statistical Association},
	volume={110},
	number={512},
	pages={1528--1542},
	year={2015},
	publisher={Taylor \& Francis}
}

@article{ai2018unified,
	title={A unified framework for efficient estimation of general treatment models},
	author={Ai, Chunrong and Linton, Oliver and Motegi, Kaiji and Zhang, Zheng},
	journal={Quantitative Economics},
	volume={12},
	number={3},
	pages={779--816},
	year={2021},
	publisher={Wiley Online Library}
}

@article{newey1994kernel,
	title={Kernel estimation of partial means and a general variance estimator},
	author={Newey, Whitney K},
	journal={Econometric Theory},
	pages={233--253},
	year={1994},
	publisher={JSTOR}
}

@book{wasserman2006all,
	title={All of nonparametric statistics},
	author={Wasserman, Larry},
	year={2006},
	publisher={Springer Science \& Business Media}
}

@article{neugebauer2007nonparametric,
	title={Nonparametric causal effects based on marginal structural models},
	author={Neugebauer, Romain and van der Laan, Mark},
	journal={Journal of Statistical Planning and Inference},
	volume={137},
	number={2},
	pages={419--434},
	year={2007},
	publisher={Elsevier}
}

@incollection{robins2000marginal,
	title={Marginal structural models versus structural nested models as tools for causal inference},
	author={Robins, James M},
	booktitle={Statistical models in epidemiology, the environment, and clinical trials},
	pages={95--133},
	year={2000},
	publisher={Springer}
}

@article{diaz2013targeted,
	title={Targeted data adaptive estimation of the causal dose--response curve},
	author={D{\'\i}az, Iv{\'a}n and van der Laan, Mark J},
	journal={Journal of Causal Inference},
	volume={1},
	number={2},
	pages={171--192},
	year={2013},
	publisher={De Gruyter}
}

@book{wainwright2019high,
	title={High-dimensional statistics: A non-asymptotic viewpoint},
	author={Wainwright, Martin J},
	volume={48},
	year={2019},
	publisher={Cambridge University Press}
}

@article{foster2019orthogonal,
	title={Orthogonal statistical learning},
	author={Foster, Dylan J and Syrgkanis, Vasilis},
	journal={The Annals of Statistics},
	volume={51},
	number={3},
	pages={879--908},
	year={2023},
	publisher={Institute of Mathematical Statistics}
}

@article{robins2017higher,
	author = {James M. Robins and Lingling Li and Rajarshi Mukherjee and Eric Tchetgen Tchetgen and Aad van der Vaart},
	title = {{Minimax estimation of a functional on a structured high-dimensional model}},
	volume = {45},
	journal = {The Annals of Statistics},
	number = {5},
	publisher = {Institute of Mathematical Statistics},
	pages = {1951 -- 1987},
	year = {2017},
}

@article{kennedy2022semiparametric,
  title={Semiparametric doubly robust targeted double machine learning: a review},
  author={Kennedy, Edward H},
  journal={arXiv preprint arXiv:2203.06469},
  year={2022}
}

@article{bertin2004asymptotically,
  title={Asymptotically exact minimax estimation in sup-norm for anisotropic H{\"o}lder classes},
  author={Bertin, Karine},
  journal={Bernoulli},
  volume={10},
  number={5},
  pages={873--888},
  year={2004},
  publisher={Bernoulli Society for Mathematical Statistics and Probability}
}

@article{hoffman2002random,
  title={Random rates in anisotropic regression (with a discussion and a rejoinder by the authors)},
  author={Hoffman, M and Lepski, Oleg},
  journal={The Annals of Statistics},
  volume={30},
  number={2},
  pages={325--396},
  year={2002},
  publisher={Institute of Mathematical Statistics}
}

@article{efromovich2007conditional,
  title={Conditional density estimation in a regression setting},
  author={Efromovich, Sam},
  journal={The Annals of Statistics},
  volume={35},
  number={6},
  pages={2504--2535},
  year={2007},
  publisher={Institute of Mathematical Statistics}
}

@inproceedings{li2005robust,
  title={Robust inference with higher order influence functions: Parts I and II},
  author={Li, Lingling and Tchetgen, Eric and Robins, J and van der Vaart, A},
  booktitle={Joint Statistical Meetings, Minneapolis, Minnesota},
  year={2005}
}

@article{van2006statistical,
  title={Statistical inference for variable importance},
  author={Van der Laan, Mark J},
  journal={The International Journal of Biostatistics},
  volume={2},
  number={1},
  year={2006},
  publisher={De Gruyter}
}

@article{westling2020causal,
  title={Causal isotonic regression},
  author={Westling, Ted and Gilbert, Peter and Carone, Marco},
  journal={Journal of the Royal Statistical Society: Series B (Statistical Methodology)},
  volume={82},
  number={3},
  pages={719--747},
  year={2020},
  publisher={Wiley Online Library}
}

@article{westling2020unified,
  title={A unified study of nonparametric inference for monotone functions},
  author={Westling, Ted and Carone, Marco},
  journal={Annals of statistics},
  volume={48},
  number={2},
  pages={1001},
  year={2020},
  publisher={NIH Public Access}
}

@article{bonvini2022MSM,
	title={Sensitivity analysis for marginal structural models},
	author={Bonvini, Matteo and Kennedy, Edward and Ventura, Valerie and Wasserman, Larry},
	journal={arXiv preprint arXiv:2210.04681},
	year={2022}
}

@article{dorn2021doubly,
	title={Doubly-valid/doubly-sharp sensitivity analysis for causal inference with unmeasured confounding},
	author={Dorn, Jacob and Guo, Kevin and Kallus, Nathan},
	journal={Journal of the American Statistical Association},
	pages={1--12},
	year={2024},
	publisher={Taylor \& Francis}
}

@book{fan2018local,
  title={Local polynomial modelling and its applications},
  author={Fan, Jianqing and Gijbels, Irene},
  year={2018},
  publisher={Routledge}
}

@article{singh2020reproducing,
  title={Kernel Methods for Causal Functions: Dose, Heterogeneous, and Incremental Response Curves},
  author={Singh, Rahul and Xu, Liyuan and Gretton, Arthur},
  journal={arXiv preprint arXiv:2010.04855},
  year={2020}
}

@article{van2014higher,
  title={Higher order tangent spaces and influence functions},
  author={van der Vaart, Aad},
  journal={Statistical Science},
  pages={679--686},
  year={2014},
  publisher={JSTOR}
}

@article{belloni2015some,
  title={Some new asymptotic theory for least squares series: Pointwise and uniform results},
  author={Belloni, Alexandre and Chernozhukov, Victor and Chetverikov, Denis and Kato, Kengo},
  journal={Journal of Econometrics},
  volume={186},
  number={2},
  pages={345--366},
  year={2015},
  publisher={Elsevier}
}
\end{document}